\documentclass[a4paper,11pt]{article} %11pt,a4paper,letterpaper,twocolumn 

\usepackage{amsmath, amssymb, amsthm, mathtools}
\usepackage[mathscr]{euscript}     %usage: \mathscr{ABC}
\usepackage{dsfont} % usage: \mathds{ABC} for nicer blackboard font than mathbb
\usepackage{graphicx} % graphics
\usepackage{xspace}
\usepackage{qip} % personal. Equation formatting.
\usepackage[section]{thmenvironments} %  personal. Definition, Theorem, etc.
\usepackage{layoutcommands} % personal. Everything that doesn't fit in
\usepackage{underscore} % need it for a doi which includes an underscore "_"
\usepackage[shortcuts]{extdash} % for hyphenation of words after dashes 
\usepackage{authblk} % author / affil commands and options
\usepackage{appendix} % appendix commands, e.g., \appendixpage
\usepackage{microtype} % improves line spacing with pdflatex
\usepackage{enumitem} % more itemize commands, in particular for spacing.
\usepackage[margin=10pt,font=small,labelfont=bf,labelsep=endash,hypcap=true]{caption}
\usepackage[margin=10pt,font=small,labelfont=bf,labelformat=parens,subrefformat=parens,labelsep=space,hypcap=false]{subcaption}
\usepackage{xcolor}
\usepackage{tikz} % inline graphics.
\usetikzlibrary{arrows.meta,fit,decorations.pathmorphing} %{arrows.meta,fit,positioning,shapes}
\usepackage[linktoc=all,colorlinks=true]{hyperref} 
\usepackage[open,numbered]{bookmark} % additional pdfbookmark
\usepackage[noadjust]{cite} % improves the cite command
\usepackage{hyperlinks} % personal. \eqnref{label}, \secref{label},
\usepackage{eprint} % downloaded from theoryofcomputing.org. automatically makes hyperlinks out of doi and arXiv numbers.

\title{(Quantum) Min-Entropy Resources}

\author[1,2]{Christopher Portmann\email{chportma@ethz.ch}}

\affil[1]{Institute for Theoretical Physics, ETH Zurich, 8093 Zurich, Switzerland.}
\affil[2]{Department of Computer Science, ETH Zurich, 8092 Zurich, Switzerland.}

\date{\today}

\newcommand{\C}{\complex}
\newcommand{\ball}[2]{\left(#2\right)^{#1}}
\newcommand{\BB}{\text{BB84}} 
\newcommand{\qkd}{\text{qkd}}
\newcommand{\dis}{\text{dis}}
\newcommand{\ec}{\text{ec}}
\newcommand{\pa}{\text{pa}}
\newcommand{\noise}{\text{noise}}
\newcommand{\adv}{\text{adv}}
\DeclareMathOperator{\synd}{synd}
\DeclareMathOperator{\corr}{corr}
\newcommand{\textcorr}{\text{corr}}
\newcommand{\verif}{\text{verif}}
\newcommand{\mac}{\mathsf{mac}}
\newcommand{\sss}{\mathsf{ss}}
\newcommand{\simE}{\mathsf{sim}_E}
\newcommand{\cj}{Choi\-/Ja\-mi\-o\l{}\-kow\-ski\xspace}
\newcommand{\cjrep}{Choi\-/Ja\-mi\-o\l{}\-kow\-ski representation\xspace}
\newcommand{\vacuum}{\ket{\Omega}}
\newcommand{\T}{\cT}
\newcommand{\LtwoOp}{\ell^2}
\newcommand{\Ltwo}[2]{#2 \tensor \LtwoOp\brs{#1}}
\newcommand{\LtwoTC}{\Ltwo{\T}{\C^d}}
\newcommand{\fock}[2][@]{\mathcal{F}\ifthenelse{\equal{#1}{@}}{}{_{#1}}\brs{#2}}
\newcommand{\tcop}[1]{\mathfrak{T}\brs{#1}}

\begin{document}

\maketitle

\begin{abstract}
  We model (interactive) resources that provide Alice with a string
  $X$ and a guarantee that any Eve interacting with her interface of
  the resource obtains a (quantum) system $E$ such that the
  conditional (smooth) min\-/entropy of $X$ given $E$ is lower bounded
  by some $k$, $\HminSmooth{\delta}{X|E} \geq k$. This (abstract)
  resource specification encompasses any setting that results in the
  honest players holding such a string (or aborting). For example, it
  could be constructed from, e.g., noisy channels, quantum key
  distribution (QKD), or a violation of Bell inequalities, which all
  may be used to derive bounds on the min\-/entropy of $X$.

  As a first application, we use this min\-/entropy resource to
  modularize key distribution (KD) schemes by dividing them in two
  parts, which may be analyzed separately. In the first part, a KD
  protocol constructs a min\-/entropy resource given the (physical)
  resources available in the specific setting considered. In the
  second, it distills secret key from the min\-/entropy resource \---
  i.e., it constructs a secret key resource. We prove security for a
  generic key distillation protocol that may use any min\-/entropy
  resource. Since the notion of resource construction is composable
  \--- security of a composed protocol follows from the security of
  its parts \--- this reduces proving security of a KD protocol
  (e.g., QKD) to proving that it constructs a min\-/entropy resource.

  As a second application, we provide a composable security proof for
  the recent Fehr-Salvail protocol [EUROCRYPT 2017] that authenticates
  classical messages with a quantum message authentication code
  (Q-MAC), and recycles all the key upon successfully verifying the
  authenticity of the message. This protocol uses (and recycles) a
  non\-/uniform key, which we model as consuming and constructing a
  min\-/entropy resource.
\end{abstract}
  % We use this min\-/entropy resource in conjunction with the Abstract
  % Cryptography (AC) framework to prove several new
  % results. Information\-/theoretic key distribution protocols may
  % generally be decomposed in two parts. The first constructs a
  % min\-/entropy resource given the (physical) resources available in
  % the specific setting considered. The second distills secret key from
  % the min\-/entropy resource. We perform this key distillation in a
  % generic way for any min\-/entropy resource. As a direct consequence,
  % it is sufficient to prove that a key distribution protocol (e.g.,
  % QKD) constructs a min\-/entropy resource. It then follows from the
  % composition theorem of AC that composing such a protocol with the
  % key distillation analyzed here generates secret key. As a second
  % application, we provide a composable security proof for the recent
  % Fehr-Salvail protocol [EUROCRYPT 2017] that authenticates classical
  % messages with a quantum message authentication code (Q-MAC), and
  % recycles all the key upon successfully verifying the authenticity of
  % the message. This protocol uses (and recycles) a non\-/uniform key,
  % which we model as consuming and constructing a min\-/entropy
  % resource.

% {\bf Keywords: } Key words?

% \clearpage
% \phantomsection
% \pdfbookmark[1]{\contentsname}{sec:toc}
% \tableofcontents
% \newpage
% \phantomsection
% \pdfbookmark[1]{\listfigurename}{sec:lof}
% \listoffigures
% \clearpage

\section{Introduction}
\label{sec:intro}

When formalizing information processing systems, one models the
information that matters, e.g., the input and output values of the
system. A concrete implementation is much more complex, e.g., values
are output at certain times, encoded in some physical form. If the
exact physical encoding or timing information is not considered
relevant, this is generally not modeled as part of the abstract
representation of the information processing system. Such a model can
be described as a \emph{specification} of a system, since it specifies
certain behaviors that the system must have, but leaves many
(irrelevant) parameters unspecified.

The level of specification of information processing systems found in
the literature is often fixed, e.g., the Universal Composability (UC)
framework~\cite{Can01,Can13} provides a concrete model of systems in
terms of interactive Turing machines, and all objects in the framework
must be specified in the given language (or something equivalent) with
the given level of specificity. Using such a framework, one can
neither model more specific objects (e.g., a protocol run at a
specific location in time and space), nor less specific objects (e.g.,
a non\-/deterministic behavior), and even less compose such objects
with each other.

The Abstract Cryptography (AC) framework~\cite{MR11} uses a top\-/down
approach to model cryptographic security, which does not have this
limitation. The more traditional bottom\-/up approach starts with a
concrete model of information processing systems, then defines how
these systems communicate, how scheduling is performed, how parties
are corrupted, etc. Instead of this, AC starts on the most abstract
level. It models abstract objects satisfying certain axioms that are
needed to make cryptographic statements. These objects can then be
instantiated with, e.g., different models of computation,
communication, scheduling, etc. In particular, as done in \cite{MR16},
AC can be instantiated with a model of \emph{resource specifications},
where two resources satisfy a relation $\aR \subset \aS,$ if $\aR$ is
more specific than $\aS$, i.e., any concrete system satisfying the
specifications of $\aR$ also satisfies those of $\aS$. Such resource
specifications can be composed and a security statements about the
composed system may be derived from the framework.

\subsection{Contributions}

% Our first contribution is to further develop the model of resource
% specifications from \cite{MR11,MR16}. We generalize the notion of
% specification construction to incorporate errors $\eps$, and then
% prove that sequential and parallel composition of two constructions
% with errors $\eps$ and $\delta$ results in a new construction with
% error at most $\eps+\delta$.

Our main contribution is to use the concept of resource
specifications~\cite{MR11,MR16,del15,dRKR15} to model min\-/entropy
resources, i.e., systems that provide honest player(s) with a random
string $X$ and an adversary with a (possibly quantum) system $E$ such
that the joint state $\rho_{XE}$ has bounded conditional (smooth)
min\-/entropy,\footnote{A formal definition of smooth min\-/entropy is
  provided in \appendixref{app:basic.min-entropy}.}
$\HminSmooth[\rho]{\delta}{X|E} \geq k$. The same techniques could be
used for other entropy measures. The reason we choose the smooth
min\-/entropy, is that it has many applications in cryptography, e.g.,
it caracterizes the maximum amount of secret key that can be extracted
from a source~\cite{Ren05}.

Specifications are essential in modeling min\-/entropy resources,
since we generally do not know (and do not care) how the joint state
$\rho_{XE}$ was generated \--- whether the adversary can influence it
or not, whether this is done with one round or multiple rounds of
interaction. Any system which interacts arbitrarily with the
adversary, but provides a random string $X$ to the honest player(s)
with a guaranteed lower bound on its conditional min\-/entropy
satisfies the specification, and can be used by protocols requiring
such a bound.

We then present two applications in which these min\-/entropy
specifications are needed to model the security of cryptographic
schemes. In the first, we show how to distill secret key from any
min\-/entropy resource shared between two players. The protocol is
standard, it uses error correction and privacy amplification to
construct a secret key resource. But the proof is not restricted to a
certain context or specific source of randomness, it is valid for any
min\-/entropy resource, e.g., whether it is constructed by noisy
classical channels~\cite{Mau93,AC93,RW05}, standard quantum key
distribution (QKD)~\cite{BB84,Ren05,TLGR12,HT12,TL15} or even
untrusted (but non\-/communicating)
devices~\cite{BHK05,PABGMS09,VV14,AFRV16}. We discuss this application
further in \secref{sec:intro.distil}.

Our second application is to provide a composable security proof for a
recent authentication protocol by Fehr and Salvail~\cite{FS17a,FS17b},
that encodes classical messages in quantum states. Due to the quantum
properties of these states, all of the secret key used in the protocol
can be recycled once the recipient has confirmed the authenticity of
the message. But a small loss of entropy occurs in the key, so it
cannot be modeled as a standard (uniform) key
resource~\cite{PR14}. Instead, we model the shared key as a
min\-/entropy specification, and prove that a slightly modified
version~\cite{FS17b} of the original Fehr-Salvail
protocol~\cite{FS17a} uses a min\-/entropy resource and an insecure
channel to construct an authentic channel and a new min\-/entropy
resource (the recycled key). This result is discussed further in
\secref{sec:intro.FS}.

\subsection{Distilling Secret Key}
\label{sec:intro.distil}

Secret key distillation from a joint probability distribution
$P_{XYE}$ between three players, Alice, Bob, and Eve, was proposed
independently by Maurer~\cite{Mau93} and by Ahlswede and
Csisz\'ar~\cite{AC93}. This problem has since been generalized to many
different contexts, e.g., in a finite setting~\cite{RW05}, when $E$ is
quantum and the tripartite state $\rho_{XYE}$ has been obtained from a
QKD protocol~\cite{Ren05,TLGR12,HT12,TL15}, or in a
device\-/independent setting, where the devices generating
$\rho_{XYE}$ are untrusted~\cite{VV14,AFRV16}. The protocols always
follow the same pattern: if needed, one performs some \emph{advantage
  distillation} to increase the correlations between the honest
players, then one performs \emph{error correction} so that they share
the same strings, and finally \emph{privacy amplification} to extract
a secret key. The different security proofs found in the literature
that one can obtain secret key given a bound on the correlations of
the state $\rho_{XYE}$ are very similar, with (small) variations to
account for the changes in
context. % This is particularly obvious when one looks at proofs in the
% literature for one kind of protocol, e.g., there is typically a large
% overlap between different QKD security proofs; they all perform
% similar error correction and privacy amplification steps with minor
% changes to adapt to each setting.

This overlap between different works can be avoided by employing a
modular (composable) approach: a task is divided in different parts,
and the security of each part is proven separately. Thus, if one part
is changed, only that piece requires a new security proof; the new
piece can then be seamlessly plugged into the other parts. The AC
framework~\cite{MR11,MR16} formalizes this as a resource theory: a
cryptographic protocol uses some resource specification $\aR$ to
construct some other resource specification $\aS$. Different protocols
may be used to construct $\aS$ in different ways. The next part of a
cryptosystem may use $\aS$ to construct $\aT$, and is oblivious to how
$\aS$ was constructed. The piece constructing $\aS$ may be changed at
will, without affecting anything else.

More concretely, we prove in this work that a min\-/entropy resource
may be used to construct a secret key resource \--- using the
aforementioned steps of error correction and privacy amplification. As
a consequence, information\-/theoretic key distribution protocols
(e.g., QKD) do not need to show that they produce secret key (and
needlessly repeat many of the steps), it is sufficient to prove that
they construct a min\-/entropy resource, i.e., that they generate a
raw key with a bound on its min\-/entropy. Bounds on the security of
the final key then follow generically from the composability of the AC
framework by plugging this into our work.

This also provides clear conditions on the min\-/entropy of the raw
key that are sufficient for the error correction and privacy
amplification to go through (namely, that the constructed resource
corresponds to a min\-/entropy specification, which we define in
\secref{sec:min-entropy}).

\subsection{Quantum Authentication of Classical Messages}
\label{sec:intro.FS}

As far back as 1982, Bennett, Brassard, and Breidbart~\cite{BBB82}
discussed how one could authenticate a classical message in a quantum
state, so that after confirming reception of the original message, the
secret key used by the protocol can be reused. A recent breakthrough
result by Fehr and Salvail~\cite{FS17a,FS17b} showed how to do this
with a \emph{prepare\-/and\-/measure} protocol, i.e., one which
involves only preparing and measuring single qubit states, and which
could already be implemented with today's
technology~\cite{SBCDLP09}. This has many interesting applications,
e.g., it could be used as a subroutine for authenticating messages in
QKD.\footnote{Whether this application is practical depends on the
  noise tolerance of the Fehr-Salvail protocol, which was has not been
  worked out~\cite{FS17a}.} One would then not need to sacrifice any
key bits for authentication. % What is more, \cite{FS17a}
% also encrypts the messages, so the loss of randomness due to the error
% correction subroutine leaking the syndrome does not occur.

Fehr and Salvail use a trace\-/distance\-/type security criterion that
is tailored for substitution attacks \--- Eve changes the message
being sent. They then show that this criterion is sufficient to prove
that their protocol can be composed sequentially with
itself. Different contexts, such as sequential composition with other
protocols,\footnoteremember{fn:intro.composition}{Ad hoc security
  definitions for individual protocols \--- e.g.,
  trace\-/distance\-/type criteria \--- do not necessarily guarantee
  that a composed protocol is secure. When using such definitions, one
  needs to additionally prove (for every different context) that the
  protocol can be safely used, and work out the corresponding
  error. For example, Fehr and Salvail do this for sequential
  composition of their protocol with itself~\cite{FS17a}. This can be
  avoided by using composable security definitions \--- e.g., the
  notion of resource construction from AC \--- which guarantee
  security in any environment. The error of the composed protocol is
  then the sum of the erros of the individual components.} parallel
composition,\footnoterecall{fn:intro.composition} or impersonation
attacks \--- Eve sends a forged cipher to Bob without knowledge of a
valid message\-/cipher pair \--- were not explicitly considered in
\cite{FS17a}.\footnote{Following an initial draft of the current work
  pointing out the issue with impersonation attacks and proposing a
  solution that recycles less key in the case of a reject, an extended
  version of the Fehr-Salvail paper was made available~\cite{FS17b},
  which sketches how the protocol can be modified to resist
  impersonation attacks without any extra loss of key.} When key
recycling is involved, impersonation attacks are particularly
powerful, because they allow Eve to obtain part of the recycled key
after Bob receives the forged cipher (e.g., he uses it in a protocol
which leaks it, like a one-time pad), which she can use to generate a
message correlated to the key and provide it to Alice for
authentication. The resulting cipher could then potentially leak more
information about the remaining secret key than a cipher prepared by
Alice under a substitution attack, where the message and key are
independent. Note that even if the protocol is only composed in a
restricted setting in which the recycled key is not leaked, the
accept/reject bit of the receiver is generally correlated to the
recycled key and cannot be hidden, which leads to the same type of
attacks.

This raises the question of whether there are other attacks or
vulnerabilities that have not been considered, and what security
criteria must be satisfied for the protocol to be usable, e.g., as a
subroutine in QKD. This is answered in a generic way by composable
frameworks. In the AC language, a QKD protocol uses an authentic
channel resource specification $\aA$ for the players to
communicate~\cite{PR14}. A standard authentication scheme that appends
a message authentication code (MAC) to the message constructs such an
authentic channel $\aA$~\cite{Por14}. And hence, by the composition
theorem of the framework, the two may be composed, and the total error
is the sum of the errors of the individual protocols.

In this work we perform such a composable analysis for a slightly
modified version~\cite{FS17b} of the original Fehr-Salvail
protocol~\cite{FS17a}. We show that this protocol uses a min\-/entropy
resource $\aH^{k}_{\min}$, a uniform key resource $\aK$, and an
insecure channel $\aQ$ to construct a new min\-/entropy resource
$\aH^{k}_{\min}$, a new uniform key $\aK$, and an authentic channel
$\aA$. The constructed channel $\aA$ may then be used by any protocol
requiring such a resource, while the constructed key resources
$\aH^{k}_{\min}$ and $\aK$ (the recycled keys) may be plugged into the
next round of authentication \--- or any other protocol requiring such keys.

% As a consequence, security of the Fehr-Salvail protocol when composed
% with QKD as suggested above does not follow from their
% work.\footnote{Using a composable framework along with the
%   min\-/entropy resources developed in this work, one could prove that
%   the security criterion used in \cite{FS17a,FS17b} is sufficient to
%   obtain composability in a restricted setting, where the honest
%   players control the scheduling and thus prevent impersonation
%   attacks. This is however not an assumption generally made in QKD.}
% We address this by providing a composable security proof

\subsection{Structure of the paper}
\label{sec:intro.structure}

We start by introducing the AC framework in \secref{sec:cc}, where we
instantiate the systems with resource specifications. In
\secref{sec:min-entropy} we then define the min\-/entropy
specifications that are used throughout this work. In
\secref{sec:distil} we show how to distill secret key from such a
min\-/entropy resource. And in \secref{sec:FS} we provide a composable
security proof for the Fehr-Salvail authentication protocol, in which
the non\-/uniform keys are modeled as min\-/entropy resources.

\section{Constructive Cryptography}
\label{sec:cc}

The AC framework~\cite{MR11,MR16} models cryptography as a resource
theory, e.g., a QKD protocol constructs a secret key resource from an
authentic channel and an insecure quantum channel. More generally, a
security statement is a constructive statement of the form ``$\gamma$
constructs $\aS$ from $\aR$ with error $\eps$,'' which we denote
\begin{equation} \label{eq:construction} 
\aR \xrightarrow{\gamma,\eps} \aS.
\end{equation}
For this reason the framework is also called \emph{constructive
  cryptography} in the literature~\cite{Mau12,MR16}.

In \eqref{eq:construction}, $\aR$ and $\aS$ are resource
specifications denoting the resources that are used by the
construction and the ones that are achieved by the construction
$\gamma$, respectively. Naturally, if $\gamma$ constructs $\aS$ from
$\aR$ and $\pi$ constructs $\aT$ from $\aS$, then one expects that
applying both constructions in sequence constructs $\aT$ from $\aR$
with an error which is the sum of the individual errors,
i.e.,
\begin{equation} \label{eq:serial}
\aR \xrightarrow{\gamma,\eps} \aS  \quad \text{and} \quad
\aS \xrightarrow{\pi,\delta} \aT
\implies \aR \xrightarrow{\pi \circ \gamma,\eps+\delta} \aT.
\end{equation}
Similarly, if $\gamma_1$ constructs $\aS_1$ from
$\aR_1$ and $\gamma_2$ constructs $\aS_2$ from $\aR_2$, then one expects that
if resources $\aR_1$ and $\aR_2$ are both available simultaneously,
and one applies both constructions, this should result in the
resources $\aS_1$ and $\aS_2$ both being constructed,
i.e.,
\begin{equation} \label{eq:parallel}
\aR_1 \xrightarrow{\gamma_1,\eps_1} \aS_1  \quad \text{and} \quad
\aR_2 \xrightarrow{\gamma_2,\eps_2} \aS_2 
\implies \aR_1 \| \aR_2 \xrightarrow{\gamma_1 |
\gamma_2,\eps_1+\eps_2} \aS_1 \| \aS_2.
\end{equation}

In the rest of this section we formalize the notions of resource
specifications and construction, which allows us to define
constructive security statement such as \eqnref{eq:construction} and
prove that \eqnsref{eq:serial} and \eqref{eq:parallel} hold with this
notion of construction. We do this following the AC top\-/down
approach, i.e., we only define the properties of the objects that are
needed, allowing them to be instantiated with any concrete model that
satisfies these properties.

\subsection{Specifications}
\label{sec:cc.spec}

As mentioned in \secref{sec:intro}, a specification can be thought of
as a (partial) description of an object to the accuracy needed for the
task at hand. This concept is much more general than modeling
cryptographic systems, and has been used to model resource theories in
physics~\cite{del15,dRKR15}. For example, let $\aS$ denote a chair. A
chair with wheels $\aR$ is more specific than a chair. We thus have
$\aR \subset \aS$, where $\subset$ is a transitive relation meaning
``more specific than''.  A piece of furniture $\aT$ is less specific,
hence $\aT \supset \aS$. As another example, consider instructions for
building a model car, which state that a piece should be painted. Let
$\aR$ denote the pair of the piece from the model car and some paint,
and let $\gamma$ denote the action of painting the piece. Then given
$\aR$, $\gamma$ constructs a painted piece $\aS$. When actually
building the model car, the paint chosen will have a certain color,
e.g., red. Let $\aR'$ denote the pair of the piece from the model car
and red paint. We then have $\aR' \subset \aR$, because $\aR'$ is more
specific than $\aR$. And applying $\gamma$ to $\aR'$ we obtain
$\aS' \subset \aS$, where $\aS'$ is the piece of the car painted
red. The instructions apply to anything that satisfies the
specifications, in particular, to objects that are more specific.

We define specifications as in \cite{MR16}: a specification is a
subset of a universe $\Phi$ of objects, namely those that satisfy the
specification. The relation $\aR \subset \aS$ for
$\aR,\aS \in \cP(\Phi)$ is then simply the subset relation, where
$\cP(\Phi)$ denotes the power set of $\Phi$.\footnote{Specifications
  are more general than this, and may be instantiated in other ways
  than with sets.} We call the elements $\sR \in \Phi$ \emph{atomic
  resources}, and the specifications $\aR \in \cP(\Phi)$ are
\emph{resource specifications}.

As two examples, we have illustrated in \figref{fig:channel}
specifications for noiseless and noisy channels. Here the atomic
resources correspond to individual channels defined by their input and
output behavior, and the specification denotes any set of channels
satisfying the specification, e.g., depolarizing channels with noise
$\lambda \leq p$ as in \figref{fig:channel.noisy}.

\begin{figure}[tb]
  \subcaptionbox[Noiseless]{\label{fig:channel.noiseless}A noiseless channel.}[.5\textwidth][c]{
\begin{tikzpicture}[
resource/.style={draw,thick,minimum width=3.2cm,minimum height=1cm},
sArrow/.style={-{Stealth[scale=.75]},thick}]

\small

\def\t{2.3} %1.6+.7

\node[resource] (ch) at (0,0) {};
\node[yshift=-1.5,above] at (ch.north) {Noiseless channel $\aS_{\id}$};
\node (alice) at (-\t,0) {};
\node (bob) at (\t,0) {};

\draw[sArrow] (alice.center) to node[auto] {$\rho$} (ch);
\draw[sArrow] (ch) to node[auto] {$\rho$} (bob.center);
\end{tikzpicture}}
\subcaptionbox[Noisy channel]{\label{fig:channel.noisy}Any depolarizing
  channel given by
  $\cE_\lambda(\rho) = (1-\lambda) \rho + \frac{\lambda}{d}
  I$, where $\lambda \leq p$.}[.5\textwidth][c]{
\begin{tikzpicture}[
resource/.style={draw,thick,minimum width=3.2cm,minimum height=1cm},
sArrow/.style={-{Stealth[scale=.75]},thick}]

\small

\def\t{2.3} %1.6+.7

\node[resource] (ch) at (0,0) {};
\node[yshift=-1.5,above] at (ch.north) {Depolarizing channels $\aR_p$};
\node (alice) at (-\t,0) {};
\node (bob) at (\t,0) {};

\draw[sArrow] (alice.center) to node[auto] {$\rho$} (ch);
\draw[sArrow] (ch) to node[auto,pos=.7] {$\cE_\lambda(\rho)$} (bob.center);
\end{tikzpicture}}

\caption[Channels]{\label{fig:channel}Two examples of resource
  specifications. \subref{fig:channel.noiseless} is a specification
  that contains one channel, the identity. \subref{fig:channel.noisy}
  is a specification for any depolarization channel with
  $\lambda \leq p$.}
\end{figure}
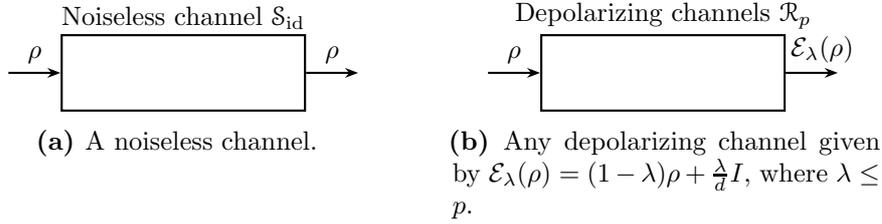

\subsection{Converters}
\label{sec:cc.converters}

To model cryptographic protocols, we consider resources that have
\emph{interfaces}, where each interface can be accessed by one
user. For example, the channel specifications from
\figref{fig:channel} have an input interface and an output interface,
with which the sender (Alice) and receiver (Bob) may interact \--- to
send and receive messages, respectively. We usually label interfaces
so as to make explicit who controls the interface, e.g., $A$, $B$,
$E$, for Alice, Bob, and Eve.

A local operation performed at one interface is called a
\emph{converter}. For example, Alice may encode her message with a
converter $\alpha$ and Bob may decode it with a converter $\beta$, as
drawn in \figref{fig:ecc}. Formally, we consider a set of objects
$\Sigma$, called converters, which are functions mapping a resource
specification to another specification, i.e.,
$\alpha : \cP(\Phi) \to \cP(\Phi)$ for $\alpha \in \Sigma$. For
example, the resource constructed in \figref{fig:ecc} is given by
$\beta(\alpha(\aR_p))$ \--- or, equivalently,
$\alpha(\beta(\aR_p))$\--- which we simply write $\beta \alpha \aR_p$
(or $\alpha \beta \aR_p$). We usually draw converters with rounded
corners.

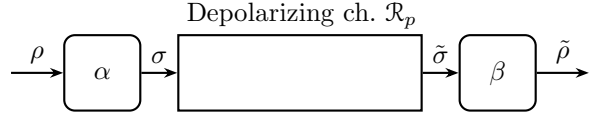
\begin{figure}[tb]
\centering
\begin{tikzpicture}[
resource/.style={draw,thick,minimum width=3.2cm,minimum height=1cm},
converter/.style={draw,thick,rounded corners,minimum width=1cm,minimum height=1cm},
sArrow/.style={-{Stealth[scale=.75]},thick}]

\small

\def\t{3.8} %1.6+.5+1+.7
\def\x{2.6} %1.6+1

\node[resource] (ch) at (0,0) {};
\node[yshift=-1.5,above right] at (ch.north west) {Depolarizing ch.\ $\aR_p$};
\node[converter] (alpha) at (-\x,0) {$\alpha$};
\node[converter] (beta) at (\x,0) {$\beta$};
\node (alice) at (-\t,0) {};
\node (bob) at (\t,0) {};

\draw[sArrow] (alice.center) to node[auto] {$\rho$} (alpha);
\draw[sArrow] (alpha) to node[auto] {$\sigma$} (ch);
\draw[sArrow] (ch) to node[auto] {$\tilde{\sigma}$} (beta);
\draw[sArrow] (beta) to node[auto] {$\tilde{\rho}$} (bob.center);
\end{tikzpicture}

\caption[Error correction]{\label{fig:ecc}Error correction converters
  $\alpha$ and $\beta$ are applied to Alice and Bob's interfaces of the
  channel specification $\aR_p$.}
\end{figure}

We often use subscripts to denote the interface to which a converter
applies, e.g., $\beta_B \alpha_A \aR_p$.  Converters applied at
different interfaces must commute (as in \figref{fig:ecc}), i.e.,
\[ \beta_B \alpha_A \aR = \alpha_A\beta_B \aR.\]
Composition of converters, $\beta \circ \alpha$, is defined by
\[ (\beta \circ \alpha) \aR \coloneqq \beta (\alpha \aR).\]
Converters must conserve the specificity relation, i.e.,
\[ \aR \subset \aS \implies \alpha \aR \subset \alpha \aS.\]
The set $\Sigma$ must also be closed under composition and contain an
identity element $\id \in \Sigma$ satisfying
\[\id \circ \alpha = \alpha \circ \id = \alpha.\]
In this work we usually write converters for honest players on the
left of the resources, and converters for dishonest players on the
right, i.e., we write $\alpha_A \aR \sigma_E$ instead of
$\sigma_E \alpha_A \aR$, where $E$ denotes Eve's interface.

Note that a set of converters $\Sigma$ that are maps $\alpha : \Phi \to
\Phi$ on the set of atomic resources satisfying the properties above
immediately induces a set of converters on resource specifications
with \[\alpha \aR \coloneqq \left\{ \alpha \sR : \sR \in \aR \right\}.\]

\subsection{Approximations}
\label{sec:cc.approximations}

After constructing the resource $\beta \alpha \aR_p$ illustrated in
\figref{fig:ecc}, one typically wants to state that it approximately
corresponds to a noiseless channel $\aS_{\id}$, i.e., $\beta \alpha
\aR_p \subset \aS^\eps_{\id}$, where $\aS^\eps_{\id}$
is an $\eps$\=/ball around $\aS_{\id}$ containing any channel that is
$\eps$\=/close to $\aS_{\id}$. Such an $\eps$\=/ball is defined with
respect to a pseudo\-/metric on atomic resources, $d : \Phi \times
\Phi \to \reals^+$,
\[ \aR^\eps \coloneqq \left\{\sS \in \Phi : \exists \sR \in \aR, d(\sR,\sS)
  \leq \eps\right\}.\]

It follows from the triangle inequality of the pseudo\-/metric that
\[ \ball{\delta}{\aR^\eps} \subset \aR^{\eps+\delta}.\]
Furthermore, if the pseudo\-/metric is contractive, i.e., for any
$\alpha \in \Sigma$ and any $\sR,\sS \in \Phi$,
$d(\alpha \sR,\alpha \sS) \leq d(\sR,\sS)$, then
\[ \alpha \aR^\eps \subset \ball{\eps}{\alpha \aR}.\]

\subsection{Resource Composition}
\label{sec:cc.composition}

If two resources $\aR$ and $\aS$ are simultaneously accessible to the
players, we wish to write this as one resource corresponding to the
\emph{(parallel) composition} of both, i.e., they are merged into one
resource such that the interfaces of each resource are simultaneously
available to each player. To do this, we define a parallel composition
operation on resource specifications
$ \| : \cP(\Phi) \times \cP(\Phi) \to \cP(\Phi)$ and write
$\aR \| \aS$ for the resulting specification. For example, a quantum
key distribution (QKD) protocol usually requires two resources to be
available to the players: an insecure quantum channel $\aQ$ and a
(multiple use, two\-/way) authentic classical channel
$\aA^c$.\footnote{The superscript $c$ in $\aA^c$ represents the number
  of uses of the channel, which we usually denote by an unspecified
  constant $c$.} As drawn in \figref{fig:qkd.real}, the channel $\aQ$
just delivers the (quantum) message to Eve and allows her to replace
it with an arbitrary message that is sent to Bob. The channel $\aA^c$
faithfully delivers (classical) messages between Alice and Bob, but
provides Eve with copies.\footnoteremember{fn:auth}{In order to
  satisfy the requirement of sequential scheduling for the specific
  instantiation of atomic resources with quantum combs
  (see \appendixref{app:sys.combs}, in particular,
  \remref{rem:scheduling}), $\aA^c$ does not immediately output two
  messages at Eve's and the receiver's interfaces. $\aA^c$ outputs a
  single message at the receiver's interace. Upon a request from Eve
  at her interface, it then gives her a copy of the messages that were
  sent. To simplify \figref{fig:qkd.real}, we do not draw the extra
  ``request arrows'', but only the information that is output by the
  systems. These requests are not needed for an instantiation of
  systems such as causal boxes that supports more complex scheduling
  (see \appendixref{app:sys.boxes}).} The players thus have access to
the resource $\aA^c \| \aQ$ and run their protocol corresponding to a
pair of converters $(\pi^\qkd_A,\pi^\qkd_B)$.\footnote{Similarly to
  the authentic channel $\aA^c$, $\pi^\qkd_A$ and $\pi^\qkd_B$ do not
  spontaneously output the key generated, they wait for a request to
  output it. Furthermore, the protocol starts with $\pi^\qkd_A$ being
  activated at its outer interface, which is not drawn in
  \figref{fig:qkd.real} either.}  The constructed system is given by
the specification $\pi^\qkd_B\pi^\qkd_A (\aA^c\|\aQ)$.

\begin{figure}[tb]
\centering
\begin{tikzpicture}[
sArrow/.style={->,>=stealth,thick},
thinResource/.style={draw,thick,minimum width=1.618*2cm,minimum height=1cm},
protocol/.style={draw,thick,rounded corners,minimum width=1.545cm,minimum height=2.5cm},
pnode/.style={minimum width=1cm,minimum height=.5cm}]

\small

\def\t{5.222} %.75+1.545+.5+1.618*3/2
\def\u{3.7} %1.545/2+.5+1.618*3/2
\def\v{.75}
\def\w{.809}

\node[pnode] (a1) at (-\u,\v) {};
\node[pnode] (a2) at (-\u,0) {};
\node[pnode] (a3) at (-\u,-\v) {};
\node[protocol] (a) at (-\u,0) {};
\node[yshift=-2,above right] at (a.north west) {$\pi^{\qkd}_A$};
\node (alice) at (-\t,0) {};
\node[below] at (alice) {Alice};

\node[pnode] (b1) at (\u,\v) {};
\node[pnode] (b2) at (\u,0) {};
\node[pnode] (b3) at (\u,-\v) {};
\node[protocol] (b) at (\u,0) {};
\node[yshift=-2,above right] at (b.north west) {$\pi^{\qkd}_B$};
\node (bob) at (\t,0) {};
\node[below] at (bob) {Bob};

\node (eve) at (0,-2) {Eve};

\node[thinResource] (cch) at (\w,\v) {};
\node[yshift=-2,above] at (cch.north) {Authentic channel  $\aA^c$};
\node[thinResource] (qch) at (-\w,-\v) {};
\node[yshift=-1.5,above right] at (qch.north west) {Insecure channel $\aQ$};
\node (eveq1) at (-\w-.4,-1.75) {};
\node (junc1) at (eveq1 |- a3) {};
\node (eveq2) at (-\w+.4,-1.75) {};
\node (junc2) at (eveq2 |- a3) {};
\node (evec) at (\w+\w,-1.75) {};
\node (junc3) at (evec |- b1) {};

\draw[sArrow,<->] (a1) to node[auto,pos=.08] {$t$} node[auto,pos=.92] {$t$}  (b1);
\draw[sArrow] (junc3.center) to node[auto,pos=.9] {$t$} (evec.center);

\draw[sArrow] (a2) to node[auto,pos=.75,swap] {$k_{A},\bot$} (alice.center);
\draw[sArrow] (b2) to node[auto,pos=.75] {$k_{B},\bot$} (bob.center);

\draw[sArrow] (a3) to (junc1.center) to node[pos=.8,auto,swap] {$\rho$} (eveq1.center);
\draw[sArrow] (eveq2.center) to node[pos=.264,auto,swap] {$\rho'$} (junc2.center) to (b3);

\end{tikzpicture}
\caption[Real QKD system]{\label{fig:qkd.real}The real QKD setting,
  with Alice on the left, Bob on the right and Eve below. Resources
  $\aA^c$ and $\aQ$ are available to the honest players, who run their
  protocol ($\pi^\qkd_A,\pi^\qkd_B$). At the end of this they either
  hold keys $(k_A,k_B)$ or produce an error $\bot$.}
\end{figure}
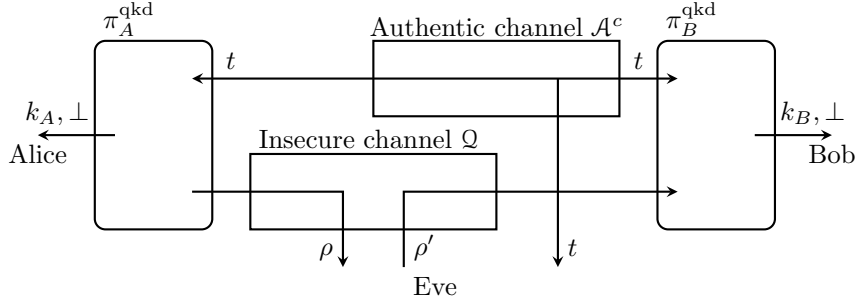

This notion of parallel composition of resources induces a notion of
parallel composition of converters $\alpha | \beta$, defined as
\[ (\alpha | \beta) (\aR \| \aS) \coloneqq (\alpha \aR) \| (\beta
\aS).\]
We require $\|$ to be associative, and when multiple resources are in
parallel, we add subscripts to the resources and converters to clarify
how they connect, e.g,
$(\alpha_1 | \beta_{23}) (\aR_1 \| \aS_2 \| \aT_3)$ \--- or
$(\alpha_{A_1} | \beta_{B_{23}}) (\aR_1 \| \aS_2 \| \aT_3)$ when we
additionally need to denote the interface to which the converters
connect. Naturally, we also require that the specificity relation be
conserved, i.e.,
\[ \aR \subset \aS \implies \aR \| \aT \subset \aS \| \aT \quad
\text{and} \quad \aT \| \aR \subset \aT \| \aS.\]
As for converters, if an
operation $\|$ with the same properties is defined on atomic
resources, then this yields a parallel composition operation on
specifications, defined as
\[ \aR \| \aS \coloneqq \left\{\sR \| \sS : \sR \in \aR, \sS \in
  \aS\right\}.\]
And if the pseudo\-/metric is context\-/insensitive, i.e., for any
atomic resources $\sR,\sS,\sT \in \Phi$,
$d(\sR \| \sT, \sS \| \sT) \leq d(\sR,\sS)$ and
$d(\sT \| \sR, \sT \| \sS) \leq d(\sR,\sS)$, then
\[ \aR^\eps \| \aS \subset \ball{\eps}{\aR \| \aS} \quad \text{and}
\quad \aR \| \aS^\eps \subset \ball{\eps}{\aR \| \aS}.\]

\subsection{Cryptographic Security}
\label{sec:cc.security}

% In the example of constructing a noiseless channel from a noisy one
% given in \secref{sec:cc.approximations}, all players are honest, and
% we have an exact description of what the interfaces of the constructed
% resource $\aS_{\id}$ should look like \--- namely, an input is
% provided at Alice's interface and an output is produced at Bob's
% interface.
In a context where some players are dishonest, one typically does not
have an exact description of how they may influence a protocol
outcome, but one has an upper bound on what they can do. For example,
in the case of QKD, one would ideally construct a secret key resource
$\aK$ that at most allows the adversary, Eve, to decide if the players
get a key or not,\footnote{Ideally one would want Eve to not even have
  the possibility of preventing the players from getting the key. But
  this cannot be achieved, since Eve can always cut the communication
  in the real system.} but not get any information about the value of
the key, as drawn in
\figref{fig:qkd.ideal}.\footnoteremember{fn:key}{To preserve a
  sequential execution of the systems, the key resource is defined so
  that upon inputing her bit, Eve receives a confirmation from
  $\aK$. The players may then individually send a request to the
  resource, and get either their key or an error. To simplify the
  picture, only the actual information transmitted is drawn in
  \figref{fig:qkd.ideal}. Request and confirmation arrows have been
  omitted.}

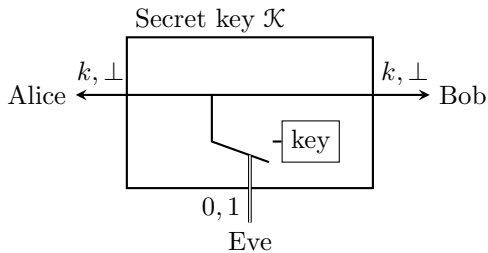
\begin{figure}[tb]
\centering
\begin{tikzpicture}[
sArrow/.style={->,>=stealth,thick},
largeResource/.style={draw,thick,minimum width=1.618*2cm,minimum height=2cm}]

\small

\def\u{.236} %2/1.618-1

\node[largeResource] (keyBox) at (0,0) {};
\node[yshift=-2,above right] at (keyBox.north west) {Secret key $\aK$};

\node (alice) at (-2.8,\u) {Alice};
\node (bob) at (2.8,\u) {Bob};
\node (eve) at (0,-1.7) {Eve};
\node[draw] (key) at (.8,\u/2-.5) {key};
\node (junc) at (-.5,0 |- key.center) {};

\draw[sArrow,<->] (alice) to node[pos=.07,auto] {$k,\bot$} node[pos=.93,auto] {$k,\bot$} (bob);
\draw[thick] (junc.center |- 0,\u) to (junc.center) to node[pos=.666] (handle) {} +(160:-.8);
\draw[thick] (.3,0 |- junc.center) to (key);
\draw[double] (eve) to node[pos=.2,auto] {$0,1$} (handle.center);

\end{tikzpicture}
\caption[A secret key resource]{\label{fig:qkd.ideal}A secret key resource
  that allows Eve to decide if Alice and Bob get a uniform key or an
  error message.\footnoterecall{fn:key}}
\end{figure}

In reality, Eve might have less power than this, e.g., she
decides that the players should get the key, but the protocol aborts
nonetheless. This, of course, is also an acceptable resource to
construct, since it is stronger than $\aK$ \--- it is stronger (for
the honest players), because Eve has less influence over the
outputs. More precisely, any resource $\aK \sigma_E$ is an acceptable
ideal resource for QKD, since anything achieved by it can be achieved
by $\aK$ if Eve decides to run $\sigma_E$ at her interface. A
specification of an ideal key resource that (at most) allows Eve to
provoke an abort is thus given by
\[ \aK^* \coloneqq \left\{ \aK \sigma_E : \sigma \in \Sigma\right\}.\]

We are now ready to define the notion of construction in the special
case of three party protocols with honest Alice and Bob and dishonest Eve.

\begin{deff}[Cryptographic security~\cite{MR11,MR16}]
\label{def:security}
We say that a protocol $\pi = (\pi_A,\pi_B)$ constructs a resource
specification $\aS$ from the specification $\aR$ with error $\eps$,
which we denote
\[ \aR \xrightarrow{\pi,\eps} \aS,\] if
\begin{equation}
\label{eq:security}
\pi_B \pi_A \aR \subset \ball{\eps}{\aS^*}.
\end{equation}
\end{deff}

One will often want to make several statements about a protocol, e.g.,
$ \aR \xrightarrow{\pi,\eps} \aS$ and
$\underline{\aR} \xrightarrow{\pi,\eps'} \underline{\aS}$, where $\aR$
and $\aS$ might be the resources in case of an active adversary (e.g.,
Eve can arbitrarily modify the communication between Alice and Bob),
and $\underline{\aR}$ and $\underline{\aS}$ are resources in case no
adversary is present (e.g., there is only non\-/malicious noise on the
channels). One then typically has $\underline{\aR} \subset \aR^*$ and
$\underline{\aS} \subset \aS^*$, i.e., if we can say something more
specific about the resources available (e.g., specific noise model),
then we can make strong statements about the constructed resource
(e.g., a key is produced with probability $1-\eps$).

Note that in order to prove that \defref{def:security} is satisfied by
a protocol $\pi$, it is sufficient to prove that for every $\sR \in
\aR$ there exists an $\sS \in \aS$ and $\sigma \in \Sigma$ such that
\begin{equation} \label{eq:security.atomic}
d\left(\pi_B \pi_A \sR , \sS \sigma_E\right) \leq \eps,
\end{equation} since this implies
that \eqnref{eq:security} holds. The exact $\sigma_E$ used in the
security proof will be called \emph{simulator}. In the following, we
often write $\sR \close{\eps} \sS$ instead of $d(\sR,\sS) \leq \eps$,
e.g., \eqnref{eq:security.atomic} becomes
\[\pi_B \pi_A \sR \close{\eps} \sS \sigma_E. \]

We now prove that this definition of construction is composable \---
i.e., \eqnsref{eq:serial} and \eqref{eq:parallel} are satisfied \---
with the error of the composed construction that is the sum of the
errors of the parts.

\begin{thm}
\label{thm:security}
If the pseudo\-/metric $d : \Phi \times \Phi \to \reals^+$ is
contractive, then
\begin{equation*} %\label{eq:thm.serial}
\aR \xrightarrow{\gamma,\eps} \aS \quad \text{and} \quad
\aS \xrightarrow{\pi,\delta} \aT \implies \aR \xrightarrow{\pi \circ
  \gamma,\eps+\delta} \aT. 
\end{equation*}
If, additionally, $d$ is context\-/insensitive, then
\begin{equation*} %\label{eq:thm.parallel}
  \aR_1 \xrightarrow{\gamma_1,\eps_1} \aS_1  \quad \text{and} \quad
  \aR_2 \xrightarrow{\gamma_2,\eps_2} \aS_2 
  \implies \aR_1 \| \aR_2 \xrightarrow{\gamma_1 |
    \gamma_2,\eps_1+\eps_2} \aS_1 \| \aS_2.
\end{equation*}
\end{thm}

\begin{proof}
  Using the properties of converters and $\eps$\=/balls from
  \secsref{sec:cc.converters} and \ref{sec:cc.approximations}, one has
\[ \pi \gamma \aR \subset \pi \ball{\eps}{\aS^*} \subset
\ball{\eps}{\pi\aS^*} \subset
\ball{\eps}{\left(\ball{\delta}{\aT^*}\right)^*} \subset
\ball{\eps}{\ball{\delta}{\left(\aT^*\right)^*}} \subset
\ball{\eps+\delta}{\aT^*}.\]
Additionally using the properties of parallel composition
from \secref{sec:cc.composition}, one has
\begin{multline*} \left(\gamma_1 | \gamma_2\right) \left(\aR_1 \| \aR_2 \right) =
\left(\gamma_1 \aR_1 \right) \| \left( \gamma_2 \aR_2\right) \subset
  \ball{\eps_1}{\aS_1^*} \| \ball{\eps_2}{\aS_2^*} \\ \subset
  \ball{\eps_1+\eps_2}{\aS_1^* \| \aS_2^*} \subset
  \ball{\eps_1+\eps_2}{\left(\aS_1 \| \aS_2\right)^*} . \qedhere\end{multline*}
\end{proof}

\subsection{Atomic resources}
\label{sec:cc.instants}

So far we have instantiated the AC framework with specifications
defined as sets of atomic resources. To model concrete systems (such
as those from Figures~\ref{fig:channel}, \ref{fig:ecc},
\ref{fig:qkd.real}, \ref{fig:qkd.ideal}), we need to instantiate these
atomic resources with a model of interactive systems that captures
quantum information\-/processing. Such a model has been given in
\cite{PMMRT17}, where the atomic resources are called \emph{causal
  boxes}, and proven to satisfy the properties required by the AC
framework.

The causal box framework~\cite{PMMRT17} allows superpositions of
causal structures to be modeled, e.g., messages can be sent in
superpositions of different orders to superpositions of different
players. In order for this to be possible, one has to model messages
as pairs $\ket{v,t} \in \hilbert_V \tensor \hilbert_T$, where $v$
denotes the value being sent and $t$ can be understood as a time tag
that controls the ordering of messages. This then allows
superpositions of different values at different times (or different
orders), e.g.,
\[\ket{\psi} = \ket{v_1,t_1} + \ket{v_2,t_2}.\]

Most of the results in the current paper consider simpler resources,
that use (classical) sequential scheduling, i.e, they receive a
message, send a message, receive another message, etc. This is the
case of all the examples seen so far in this section. These can be
modeled as memory channels or \emph{quantum
  combs}~\cite{GW07,Gut12,CDP09} (see also \cite{Har11,Har12,Har15}),
i.e., systems with some internal memory $M$, and which upon receiving
an input in some register $X_i$ produce a single output in some
register $Y_i$ by applying a map $\cE_i : \lo{X_iM} \to \lo{MY_i}$
which updates the internal memory and computes the output. Such
systems can be seen as special cases of causal boxes, in which the
(unnecessary) time tag is not written out explicitly.

These two models of quantum information\-/processing systems \---
quantum combs and causal boxes \--- are introduced in detail
in \appendixref{app:sys}. In the following we model all systems as
quantum combs, unless stated otherwise.

Regardless of the instantiation, the distance between two atomic
resources $\sR$ and $\sS$ is taken to be the \emph{distinguishing
  advantage}, i.e., a distinguisher $\sD$ interacts with one of the two and
has to guess which one it is given. The distinguishing advantage is
then defined as
\[ d(\sR,\sS) \coloneqq \sup_{\sD} \left| \Pr \left[ \sD(\sR) = 1
  \right] - \Pr \left[ \sD(\sS) = 1\right] \right|,\]
where $\sD(\sR)$ is the binary random variable corresponding to the
distinguisher's guess when interacting with $\sR$. In the case of
quantum combs the distinguishing advantage corresponds to the trace
distance between the two states held in memory by the distinguisher
after interacting with the corresponding systems~\cite{CDP09,Gut12}. A
similar result is derived for causal boxes~\cite{PMMRT17}.

\section{Min-entropy resources}
\label{sec:min-entropy}

Before modeling min\-/entropy specifications, we first consider
min\-/entropy atomic resources using both the quantum
combs~\cite{GW07,Gut12,CDP09} and causal boxes~\cite{PMMRT17} models,
in \secsref{sec:min-entropy.comb} and \ref{sec:min-entropy.box},
respectively. We then define min\-/entropy specifications in
\secref{sec:min-entropy.spec}, and discuss $\eps$\-/balls around
min\-/entropy specifications in \secref{sec:min-entropy.approx}.

The quantum notation and concepts used in this section and the
following are introduced in \appendixref{app:basic}. The models of
quantum combs and causal boxes are explained
in \appendixref{app:sys.combs} and \appendixref{app:sys.boxes},
respectively.

\subsection{Finite dimensional systems}
\label{sec:min-entropy.comb}

Quantum combs~\cite{GW07,Gut12,CDP09}, by definition, can only
interact a finite number of times with their environment, since they
are modeled by (the \cjrep of) a finite dimensional completely
positive, trace\-/preserving (CPTP) map
(see \appendixref{app:sys.combs}). This is sufficient to model any
min\-/entropy resource that generates the output $X$ after a fixed
number of rounds of interaction (as is the case for the applications
to QKD and authentication considered in this work), or if the systems
have a timeout after which they abort \--- since in a finite amount of
time a causal system can only send a finite number of
messages~\cite{PMMRT17}.

To model min\-/entropy atomic resources with combs, we consider a
system $\sH^n$ that is activated by receiving a fixed symbol
``start''. It then outputs a quantum message, receives a quantum
message as input, outputs another message, etc., until a total of $n$
messages have been exchanged with its environment. It then finally
outputs a classical (random) value $X$ from an alphabet
$\cX \cup \{\bot\}$, where $\bot$ is an error symbol denoting that the
system failed to generate an output (with enough entropy). Such a
system is illustrated in \figref{fig:min-entropy.comb}, and is a
called a min\-/entropy comb if it satisfies the following conditions.

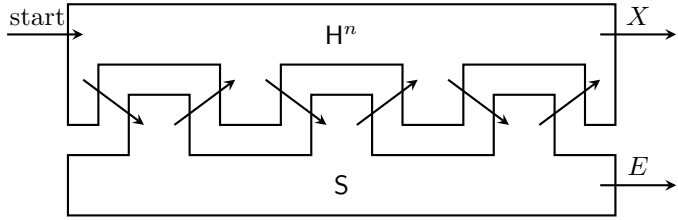
\begin{figure}[tb]
\centering
\begin{tikzpicture}[scale=.8,
sArrow/.style={->,>=stealth,thick}]

\small

\draw[thick] (0,0) to ++(0,1) to ++(1,0) to ++(0,1) to ++(1,0) to
++(0,-1) to ++(2,0) to ++(0,1) to ++(1,0) to
++(0,-1) to ++(2,0) to ++(0,1) to ++(1,0) to
++(0,-1) to ++(1,0) to ++(0,-1) -- cycle;

\draw[thick] (0,3.5) to ++(0,-2) to ++(.5,0) to ++(0,1) to ++(2,0) to
++(0,-1) to ++(1,0) to ++(0,1) to ++(2,0) to
++(0,-1) to ++(1,0) to ++(0,1) to ++(2,0) to
++(0,-1) to ++(.5,0) to ++(0,2) -- cycle;

\node (R) at (4.5,3) {$\sH^n$};
\node (S) at (4.5,.5) {$\sS$};

\draw[sArrow] (.25,2.25) to (1.25,1.5);
\draw[sArrow] (1.75,1.5) to (2.75,2.25);
\draw[sArrow] (3.25,2.25) to (4.25,1.5);
\draw[sArrow] (4.75,1.5) to (5.75,2.25);
\draw[sArrow] (6.25,2.25) to (7.25,1.5);
\draw[sArrow] (7.75,1.5) to (8.75,2.25);
\draw[sArrow] (-1,3) to node[auto,pos=.4] {start}  (.25,3);
\draw[sArrow] (8.75,3) to node[auto] {$X$}  (10,3);
\draw[sArrow] (8.75,.5) to node[auto] {$E$} (10,.5);

\end{tikzpicture}
\caption[Min-entropy comb]{\label{fig:min-entropy.comb}$\sH^n$ is a
  $(k,\delta)$\-/min\-/entropy comb if for all $\sS$ the state
  $\rho_{XE}$ satisfies the conditions from
  \defref{def:min-entropy.comb}.}
\end{figure}

\begin{deff}[Min-entropy comb]
\label{def:min-entropy.comb}
Let $\sH^n$ be a quantum comb with a structure as described above, and
let $\sS$ be any comb that mirrors it: $\sS$ first receives a quantum
message, then outputs one, etc, for a total of $n$ interactions, and
finally outputs a quantum system $E$. We say that $\sH^n$ is a
\emph{$(k,\delta)$\-/min\-/entropy comb} if for all $\sS$, the final
joint state $\rho_{XE}$ is given by
\begin{equation} \label{eq:min-entropy.1} 
\rho_{XE} = \proj{\bot} \tensor \tau_E + \sigma_{XE}
\end{equation}
for an arbitrary state $\tau_E$, and a state $\sigma_{XE}$ satisfying
\begin{equation} \label{eq:min-entropy.2}
  \HminSmooth[\sigma]{\delta}{X|E} \geq k,
\end{equation}
where $\HminSmooth[\sigma]{\delta}{X|E}$ denotes the $\delta$\-/smooth
conditional min\-/entropy of $\sigma_{XE}$ and is defined
in \appendixref{app:basic.min-entropy},
\defref{def:smooth-min-entropy}.
\end{deff}

Note that the min\-/entropy condition in \eqnref{eq:min-entropy.2}, is
defined on the \emph{subnormalized} state $\sigma_{XE}$. This is
because the min\-/entropy of the renormalized state
$\sigma_{XE} / \trace{\sigma_{XE}}$ is not meaningful when the
probability of not aborting, $\trace{\sigma_{XE}}$, is very small \---
one does not care what happens conditioned on an unlikely event. For a
state $\sigma_{XE}$ with $\trace{\sigma_{XE}} \leq 2^{-k}+\delta^2/2$,
\eqnref{eq:min-entropy.2} is trivially satisfied.

\subsection{Unbounded interactions}
\label{sec:min-entropy.box}

The systems defined in \secref{sec:min-entropy.comb} have a fixed
number of interactions (namely $n$) with the environment. This is
necessary if one wants to keep the joint input and output Hilbert
spaces finite dimensional. One can however imagine more general
systems where the number of interactions is a priori unbounded and
depends on the values that are input \--- or where the number of
interactions is in a superposition of different values. We model this
using causal boxes~\cite{PMMRT17}.

As explained in \appendixref{app:sys.boxes}, in the causal box model a
message can be thought of as a pair of a value $v \in \cV$ and a
position $t \in \T$, where $\T$ is a discrete partially ordered set
allowing messages to be ordered with respect to each other (e.g.,
$\T = \rationals$). Let
$\hilbert_X = \hilbert_\cV \tensor \hilbert_\T$ be the corresponding
Hilbert space. An output of a causal box may consist in one message,
multiple messages, no messages at all, or any superposition
thereof. The corresponding message space is a Fock space given by
\begin{equation*} %\label{eq:generalfockspace}
  \fock{\hilbert_X} \coloneqq \bigoplus_{n = 0}^\infty \vee^n
  \hilbert_X\,, \end{equation*} where $\vee^n \hilbert_X$ denotes the
symmetric subspace of $\hilbert_X^{\tensor n}$, and
$\hilbert_X^{\tensor 0}$ is the one dimensional space containing the
vacuum state $\vacuum$.

Thus, a figure drawing a causal box does not have one arrow for every
output, but instead one arrow for an output wire that many produce an
unbounded number of messages. For example, in
\figref{fig:min-entropy.box}, the two causal boxes $\sH$ and $\sS$ are
connected by two wires, and may be repeatedly exchanging messages on
these wires \--- these systems are not required to
terminate.\footnote{Formally, causal boxes are modeled as sequences of
  maps over ever increasing intervals of $\T$,
  see \appendixref{app:sys.boxes}.}

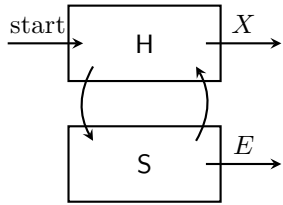
\begin{figure}[tb]
\centering
\begin{tikzpicture}[
sArrow/.style={->,>=stealth,thick},
sys/.style={draw,thick,minimum width=2cm,minimum height=1cm},
sysInternal/.style={minimum width=.6cm,minimum height=.6cm}]

\small

\def\t{1.8}%1+.8
\def\v{.8}%.5+.3
\def\x{.5}

\node[sys] (H) at (0,\v) {$\sH$};
\node[sysInternal] (HL) at (-\x,\v) {};
\node[sysInternal] (HR) at (\x,\v) {};
\node[sys] (S) at (0,-\v) {$\sS$};
\node[sysInternal] (SL) at (-\x,-\v) {};
\node[sysInternal] (SR) at (\x,-\v) {};

\draw[sArrow,bend right] (HL) to (SL);
\draw[sArrow,bend right] (SR) to (HR);

\draw[sArrow] (-\t,\v) to node[auto,pos=.4] {start} (HL);
\draw[sArrow] (HR) to node[auto] {$X$}  (\t,\v);
\draw[sArrow] (SR) to node[auto] {$E$}  (\t,-\v);

\end{tikzpicture}
\caption[Min-entropy causal box]{\label{fig:min-entropy.box}In
  the model of causal boxes, each wire may be used an unbounded number
  of times to transmit messages, e.g., there is no bound on the number
  of rounds of communication between $\sH$ and $\sS$. The system $\sH$ is a
  $(k,\delta)$\-/min\-/entropy causal box if for all $\sS$ and all $t$ the
  state $\rho^{\leq t}_{XE}$ satisfies the conditions from
  \defref{def:min-entropy.box}, where $\rho^{\leq t}_{XE}$ is the
  state of the $XE$ systems truncated at position $t$.}
\end{figure}

Similar to the model of min-entropy combs, we model a min-entropy
causal box as being started by receiving a corresponding
message,\footnote{Technically this is not necessary, a causal box may
  spontaneously output a message. But it simplifies the scheduling
  with other systems to consider only causal boxes that are started by
  an external message.} after which it interacts with an environment
(e.g., the system $\sS$ drawn in \figref{fig:min-entropy.box}). At any
point during this interaction, it may output a value on the wire $X$
from the alphabet $\cX \cup \{\bot\}$, where, as above, $\bot$ is an
error symbol denoting that the system failed to generate an output
(with enough entropy). But we restrict the output on $X$ to the
subspace $\hilbert_X^{\tensor 0} \oplus \hilbert_X^{\tensor 1}$, i.e.,
at most one value is output. Unlike for combs, we cannot place a
requirement on the entropy of $X$ once a value is output, since this
condition might not be well\-/defined (e.g., for $t \in \naturals$, a
value is produced in position $t$ with probability $2^{-t}$). Instead,
we bound the entropy of the output up to position $t$ for all
$t \in \T$.

\begin{deff}[Min-entropy causal box]
\label{def:min-entropy.box}
Let $\sH$ be a causal box with a structure as described above, and let
$\sS$ be any causal box with an input and output wire that match the
dimensions of the output and input wires of $\sH$, and with an
additional output wire $E$, and let them be connected by these
wires. Furthermore, for $VT=X$, let $\rho^{\leq t}_{VTE}$ be the
subnormalized output on $XE$ up to position $t$ projected on the
non\-/vacuum subspace
$\hilbert_V \tensor \hilbert_T \tensor \fock{\hilbert_E}$.  We say
that $\sH$ is a \emph{$(k,\delta)$\-/min\-/entropy causal box} if for
all $\sS$ and all $t \in \T$, the state $\rho^{\leq t}_{VTE}$ is given
by
\begin{equation*} %\label{eq:min-entropy.3} 
\rho^{\leq t}_{VTE} = \proj{\bot} \tensor \tau_{TE} + \sigma_{VTE}
\end{equation*}
for an arbitrary state $\tau_{TE}$, and a state $\sigma_{VTE}$ satisfying
\begin{equation*} %\label{eq:min-entropy.4}
  \HminSmooth[\sigma]{\delta}{V|TE} \geq k.
\end{equation*}
\end{deff}

Note that \defref{def:min-entropy.box} conditions the entropy of the
output on the register $T$ as well as $E$, since this ordering
information could be learned by an adversary. Though if one considers
only systems where timing is independent of the values of messages,
then it would be sufficient to bound
$\HminSmooth[\sigma]{\delta}{V|E}$.

\subsection{Min-entropy specification}
\label{sec:min-entropy.spec}

For min\-/entropy combs and causal boxes (defined in the previous two
sections) to be valid atomic resources, one still has to assign
interfaces to the various inputs and outputs. In a model with two
interfaces $A$ and $E$ \--- accessed by Alice and Eve, respectively
\--- the final output $X$ occurs at Alice's interface, and the
interaction with the environment occurs at Eve's interface. The input
``start'' may be provided at either Alice or Eve's interface,
depending on the protocol considered. A min\-/entropy specification is
then simply the set of all min\-/entropy atomic resources satisfying a
certain bound on the conditional min\-/entropy.

\begin{deff}[Min-entropy specification]
\label{def:min-entropy.spec}
  A $(k,\delta)$\=/min\-/entropy specification $\aH^{k,\delta}_{\min}$
  is the set of all $(k,\delta)$\=/min\-/entropy atomic resources.
\end{deff}

\defref{def:min-entropy.spec} is typically the min\-/entropy
specification one would construct in an adversarial setting: it
provides a bound on the entropy, but not on the actual probability of
terminating which such a random string, since usually Eve can force
honest players to abort. If one considers a context where Eve is not
active, then one generally wishes to prove something stronger, e.g.,
with high probability the protocol generates a random string. Such a
statement can easily be achieved by restricting the set
$\aH^{k,\delta}_{\min}$ to contain only those systems such that for
all $\sS$, the output
$\rho_{XE} = \proj{\bot} \tensor \tau_E + \sigma_{XE}$ has
$\trace{\sigma_{XE}} \geq 1-\eps$.\footnote{In the model with causal
  boxes, one would require that by a certain position $t_0$,
  $\trace{\sigma^{\leq t_0}_{XE}} \geq 1-\eps$.}

In the two applications considered in this work, we need min\-/entropy
resources with three interfaces for Alice, Bob, and Eve. These are
defined identically to the two interface versions described so far,
except that a string $Y$ is output at Bob's interface as
well.\footnote{As in the examples of systems in \secref{sec:cc}, to
  preserve sequential scheduling when the atomic resources are
  instantiated with quantum combs, the outputs $X$ and $Y$ are not
  produced spontaneously, but only after having received a request
  from Alice and Bob, respectively.} Here, the conditional
min\-/entropy is always defined on the $XE$ system. The allowed
correlations between $Y$ and $XE$ will depend on the context, e.g.,
the min\-/entropy resource constructed by QKD allows $Y$ to be an
arbitrary random string, whereas the min\-/entropy resource used as a
secret key by the Fehr-Salvail protocol requires $Y$ to be a copy of
$X$.

\subsection{Approximate min-entropy specifications}
\label{sec:min-entropy.approx}

In \secref{sec:cc.approximations} we defined a notion of an
approximate specification, directly in the AC language, by taking an
$\eps$\=/ball around specifications. The min\-/entropy resources
defined in \secref{sec:min-entropy.spec} have their own notion of
approximation built-in, namely the smoothing parameter $\delta$. This
is useful, because (as we prove in the next lemma)
$\delta$\=/smoothing is a less specific condition that results in a
larger set of atomic resources than $\eps$\=/balls.\footnote{It is
  currently unknown if it is possible to prove that a QKD protocol
  constructs the smaller set of resources without the smoothing
  parameter.} It remains open to find bounds in the other direction
\--- i.e., whether
$\aH^{k,\delta+\eps}_{\min} \subset
\ball{\eps'}{\aH^{k,\delta}_{\min}}$
for some reasonable $\eps'$ \--- or examples showing that such bounds
do not exist.

\begin{lem}
\label{lem:smoothnessandballs}
For min\-/entropy specifications one has
\[ \ball{\eps}{\aH^{k,\delta}_{\min}} \subset
\aH^{k,\delta+\sqrt{2\eps}}_{\min}.\]
\end{lem}

\begin{proof}
By contradiction, let us suppose there exists an $\sH$ such that
\begin{align} 
\sH & \notin \aH^{k,\delta+\sqrt{2\eps}}_{\min}, \label{eq:smoothnessandballs.assumption.1}
\intertext{but} \sH & \in \ball{\eps}{\aH^{k,\delta}_{\min}}. \label{eq:smoothnessandballs.assumption.2}
\end{align}
From \eqnref{eq:smoothnessandballs.assumption.1} we find that there
must exist an $\sS$ such that the state
$\rho_{XE} = \proj{\bot} \tensor \tau_{E} + \sigma_{XE}$ resulting
from $\sS$ and $\sH$ interacting has
\begin{equation} \label{eq:smoothnessandballs.contradiction} 
\HminSmooth[\sigma]{\delta+\sqrt{2\eps}}{X|E} < k.
\end{equation}
From \eqnref{eq:smoothnessandballs.assumption.2} there must exist an
$\sH' \close{\eps} \sH$ such that when interacting with the same
$\sS$, the resulting state $\rho'_{XE} = \proj{\bot}
\tensor \tau'_{XE} + \sigma'_{XE}$ has
\begin{equation} \label{eq:smoothnessandballs.step} 
\HminSmooth[\sigma']{\delta}{X|E} \geq k.
\end{equation}

Since $\sH' \close{\eps} \sH$, it follows that
$D(\rho'_{XE},\rho_{XE}) \leq \eps$, where $D(\cdot,\cdot)$ is the
trace distance. One then obtains
\begin{align*} D(\rho'_{XE},\rho_{XE}) & = \frac{1}{2}
  \trnorm{\sigma'_{XE} - \sigma_{XE}} + \frac{1}{2} \trnorm{\tau'_{E}
    - \tau_{E}} \\ & \geq \frac{1}{2} \trnorm{\sigma'_{XE} -
    \sigma_{XE}} + \frac{1}{2} \left| \tr \tau'_{E} - \tr \tau_{E}
  \right| \\ & = \bar{D}(\sigma'_{XE},\sigma_{XE}) \\ & \geq
  P(\sigma'_{XE},\sigma_{XE})^2/2,\end{align*} where
$\bar{D}(\cdot,\cdot)$ is the generalized trace distance and
$P(\cdot,\cdot)$ is the purified distance (see
\appendixref{app:basic.distance} and \cite{TCR10}).

Putting this together with \eqnref{eq:smoothnessandballs.step} and the
triangle inequality, there must exist a $\tilde{\sigma}_{XE}$ such that
$\Hmin[\tilde{\sigma}]{X|E} \geq k$ and
$P(\tilde{\sigma}_{XE},\sigma_{XE}) \leq \delta + \sqrt{2\eps}$, which
contradicts \eqnref{eq:smoothnessandballs.contradiction}.
\end{proof}

\section{Distilling Secret Key}
\label{sec:distil}

In this section we show how to extract secret key from a tripartite
min\-/entropy resource, where Alice and Bob obtain strings $X$ and
$Y$, and Eve gets side information. In \secref{sec:distil.ec} we show
how error correction and verification construct a min\-/entropy
resource which provides Alice and Bob with identical strings
$X' = Y'$, given a resource that has no guarantee on the correlations
between $X$ and $Y$. Then in \secref{sec:distil.pa} the players use
privacy amplification to construct a secret key resource (see
\figref{fig:qkd.ideal} for an illustration of a secret key
resource). The two steps are composed in \secref{sec:distil.comp}.

These steps of error correction and privacy amplification are
performed in a generic way for any resource that satisfies the
min\-/entropy specification. We show in \appendixref{app:distil.qkd}
that QKD is a special case of this, which constructs a (more specific)
min\-/entropy resource.

\subsection{Error Correction}
\label{sec:distil.ec}

In this section we analyze a generic class of error correction
schemes, which are now standard in QKD (see, e.g., \cite{TLGR12,
  TL15}). The error correction procedure has two steps. In the first,
Alice applies a function $Z = \synd(X)$ which computes a syndrome for
$X$. She then sends this on an authentic channel to Bob, who applies
the second function $\hat{X} = \corr(Y,Z)$ to get an estimate
$\hat{X}$ of Alice's string $X$. Such pairs of functions
$(\synd,\corr)$ are parametrized by two values: by a subset
$\cS \subset \cX \times \cX$ of strings that they can correct, i.e,
for all $(x,y) \in \cS$, $\corr(y,\synd(x)) = x$, and by the length of
the string $|Z| = r$, which bounds how much information is leaked to
the adversary about $X$. Let $\pi^\textcorr$ denote the first part of the
procedure.

In the second step, the players verify whether the error correction
was successful. They do this by first choosing a function $f$
uniformly at random from a family of almost universal hash functions, i.e., a
set $\cF = \{f : \cX \to \cY\}$ such that for all $x,x' \in \cX$,
$x \neq x'$,
\begin{equation}
\label{eq:universalhashing}
\Pr_{f \in \cF} \left[ f(x) = f(x') \right] \leq \eps_\verif.
\end{equation}
Alice sends $(f,f(X))$ to Bob, who verifies that $f(X) = f(\hat{X})$,
and tells Alice whether this is successful or whether they should
abort. Let the length of the hash be $|f(x)| = t$, i.e.,
$|\cY| = 2^t$. Let $\pi^{\verif}$ denote the second part of this
procedure. The complete error correction scheme is then
$\pi^{\ec} = \pi^\verif \circ \pi^\textcorr$.

There are now two cases to consider. The first is when the players
have absolutely no guarantee about the correlations between $X$ and
$Y$ \--- in QKD, this is the case if an adversary is eavesdropping on
the channel and can arbitrarily change the quantum messages. In this
case, we wish to bound the information Eve has as well as bound the
probability that the honest players end up with different strings. The
second case is when some bounds on the correlations between $X$ and
$Y$ are known (e.g., the number of bit flips between the two) \--- in
QKD, this is the case if only natural noise is present on the
channel. If such a more specific resource is present, we then wish to
prove that a more specific (stronger) resource is constructed, namely,
we additionally want a bound on the probability that the players abort
\--- which is known as the \emph{robustness} (to this noise model) of
a protocol.

\subsubsection{Security}
\label{sec:distil.ec.security}

We start with the first case, when Alice and Bob share a resource
$\overline{\aH}^{k,\delta}_{\min}$, where Bob's string $Y$ may be
arbitrary (but is always $\bot$ if Alice's is $\bot$). Furthermore,
they share an authentic channel $\aA^c$.\footnote{See the description
  of the multiple use authentic channel $\aA^c$ in
  \secref{sec:cc.composition} and \footnoteref{fn:auth} (though in
  this case, a single use channel is sufficient).} It is
straightforward that without any further information about $Y$ one has
no hope of guaranteeing that $\pi^\textcorr$ will work. Instead, we just
bound the information leaked to the adversary by this procedure.
\begin{lem}
\label{lem:distil.ec.corr.adv}
$\pi^\textcorr$ leaks at most $r$ bits of information about Alice's string
$X$ to Eve, i.e.,
\begin{equation*}
%\label{eq:distil.ec.corr.adv}
 \overline{\aH}^{k,\delta}_{\min} \| \aA^c \xrightarrow{\pi^\textcorr,0}
\overline{\aH}^{k-r,\delta}_{\min}.
\end{equation*}
\end{lem}

\begin{proof} Immediate from
\lemref{lem:tech.min-entropy-inequalities.chain}.\end{proof}

In this case, the crucial part of the error correction procedure is
the error verification, which constructs a min\-/entropy resource that
provides Bob with a perfect copy of Alice's string (or aborts).

\begin{thm}
\label{thm:distil.ec}
  Let $\overline{\aH}^{k,\delta}_{\min}$ denote a
  $(k,\delta)$\-/min\-/entropy resource, where Bob's string $Y$ is
  arbitrary (but is always $\bot$ if Alice's is $\bot$), and
  $\aH^{k,\delta}_{\min}$ denote a $(k,\delta)$\-/min\-/entropy
  resource in which Bob has an exact copy of Alice's string. Then
\[ \overline{\aH}^{k,\delta}_{\min} \| \aA^c \xrightarrow{\pi^\verif,\eps_\verif}
 \aH^{k-t,\delta}_{\min}.\]
\end{thm}

Note that one has
$\aH^{k,\delta}_{\min} \subset \overline{\aH}^{k,\delta}_{\min}$. What
$\pi^{\verif}$ does is check if the resource shared by Alice and Bob
belongs to $\aH^{k,\delta}_{\min}$, and aborts if not. It leaks $t$
bits in the process and fails with probability $\eps_\verif$.

\begin{proof}
  Let $\sH \in \overline{\aH}^{k,\delta}_{\min}$ and define $\sH'$ to
  work as follows. $\sH'$ internally runs $\sH$. If it obtains
  $X = \bot$, it simply outputs $\bot$ at both Alice and Bob's
  interfaces when requested. If $X \neq \bot$, it prepares $(f,f(X))$
  for a uniformly chosen $f$ to be output at Eve's interface as
  information sent on the authentic channel if requested. It checks
  whether $f(X) = f(Y)$, and prepares the corresponding bit of
  backwards information to be output at Eve's interface if
  requested. If these strings are equal, $\sH'$ now outputs
  $X' = Y' = X$ at Alice's and Bob's interfaces when requested. If
  they are different, it outputs $X' = Y' = \bot$.

  To prove that this theorem holds, we will show that
  \[\pi^\verif \overline{\aH}^{k,\delta}_{\min} \subset
  \ball{\eps_\verif}{\aH^{k-t,\delta}_{\min}}.\]
  This statement follows if we can show that
  $\sH' \in \aH^{k-t,\delta}_{\min}$ and that
  $\pi^\verif \sH \close{\eps_\verif} \sH'$. Note that
  $\pi^\verif \sH$ and $\sH'$ behave identically, except for the
  string $Y'$ output at Bob's interface, which, in the case of $\sH'$
  is always $Y' = X'$, whereas in the real setting one might have
  $Y' \neq X'$. A bound on the probability that $Y' \neq X'$ is thus a
  bound on the distinguishability between the two systems; and it
  follows from the almost universal hashing property
  (\eqnref{eq:universalhashing}) that
  $\Pr[Y' \neq X'] \leq \eps_\verif$.

  To prove that $\sH' \in \aH^{k-t,\delta}_{\min}$ we need to show
  that for all $\sS$ interacting with $\sH'$, the resulting state is
  of the from
  $\rho_{X'Y'E'} = \proj{\bot,\bot} \tensor \tau_{E'} +
  \sigma_{X'Y'E'}$
  with $X' = Y'$ and $\HminSmooth[\sigma]{\delta}{X'|E'} \geq k-t$. By
  construction we always have $X' = Y'$, so the only condition to
  verify is the min\-/entropy bound. The last messages output by
  $\sH'$ at Eve's interface are $f$, $f(X)$, and the decision bit
  $f(X) = f(Y)$. We denote these registers by $F$, $Z$, and
  $\Omega$. It is sufficient to consider a system $\sS$ that interacts
  with $\sH$, and together with the output in register $E$ also
  outputs $FZ\Omega$, hence $E'=EFZ\Omega$. By definition of $\sH$,
  the state $\rho_{XYEFZ\Omega}$ resulting from $\sS$ interacting with
  $\sH$ and adding the transcript $FZ\Omega$ that is produced at Eve's
  interface has the form
  $\rho_{XYEFZ\Omega} = \proj{\bot,\bot} \tensor \tau_{EFZ\Omega} +
  \gamma_{XYEFZ\Omega}$
  with $\HminSmooth[\gamma]{\delta}{X|E} \geq k$. From
  \lemref{lem:tech.min-entropy-inequalities.chain} one has
  $\HminSmooth[\gamma]{\delta}{X|EFZ} \geq k - t$. Furthermore,
  $\gamma_{XYEFZ\Omega} = \gamma^0_{XYEFZ} \tensor \proj{0} +
  \gamma^1_{XYEFZ} \tensor \proj{1}$,
  and $\sigma_{X'Y'E'} = \gamma^1_{XYEFZ} \tensor \proj{1}$. It then
  follows from \lemref{lem:tech.min-entropy-inequalities.event} that
  $\HminSmooth[\sigma]{\delta}{X'|E'} \geq k-t$.
\end{proof}

\begin{cor}
\label{cor:distil.ec}
$\pi^\ec = \pi^\verif \circ \pi^\textcorr$ constructs
$\aH^{k-r-t,\delta}_{\min}$ from $\overline{\aH}^{k,\delta}_{\min}$
and $\aA^c$ with error $\eps_\verif$,
\begin{equation*}
%  \label{eq:distil.ec}
\overline{\aH}^{k,\delta}_{\min} \| \aA^c
  \xrightarrow{\pi^\ec,\eps_\verif} \aH^{k-r-t,\delta}_{\min},\end{equation*}
where $\aA^c =
\aA^{c_1} \| \aA^{c_2}$, and $\aA^{c_1}$ and $\aA^{c_2}$ are used by
$\pi^\textcorr$ and $\pi^\verif$, respectively.
\end{cor}

\begin{proof}
  Follows from \lemref{lem:distil.ec.corr.adv},
  \thmref{thm:distil.ec}, and \thmref{thm:security}.
\end{proof}

% \begin{rem}
% \label{rem:distil.ec}
% The analysis performed in this section assumes that the players use
% authentic channels to communicate. If instead the players use secure
% channels $\aS^c$, e.g., by using the modified Fehr-Salvail protocol
% that we analyze in \secref{sec:FS},\footnote{Since this protocol
%   allows all the key to be recycled if the messages pass
%   authentication, it has no net key cost.} then the error correction
% scheme does not leak any information about $X$ to the
% adversary. Instead of \eqnref{eq:distil.ec}, one then gets
% \begin{equation*}
% \overline{\aH}^{k,\delta}_{\min} \| \aS^c
%   \xrightarrow{\pi^\ec,\eps_\verif} \aH^{k,\delta}_{\min}.\end{equation*}
% \end{rem}

\subsubsection{Robustness}
\label{sec:distil.ec.robustness}

We now consider the case where correlations between $X$ and $Y$ are
known. The procedure $\pi^\textcorr$ must be chosen to deal specifically
with these errors. Let
$\overline{\aR}^{k,\delta,p}_{\min} \subset
\overline{\aH}^{k,\delta}_{\min}$
be a min\-/entropy specification that is restricted to resources that
only abort with probability at most $p$ and furthermore, that do not
allow arbitrary outputs $Y$, but only strings such that
$(X,Y) \subset \cS$, where $\cS$ is the set of pairs that can be
corrected. And let
$\aR^{k,\delta,p}_{\min} \subset \aH^{k,\delta}_{\min}$ be a
specification of resources that also abort with probability at most
$p$ and produce identical strings at Alice's and Bob's interfaces. It
is straightforward that in this case the procedure $\pi^\ec$
constructs $\aR^{k,\delta,p}_{\min}$ from
$\overline{\aR}^{k-r-t,\delta,p}_{\min}$, i.e., one can make
statements about the probability that the protocol aborts.

\begin{lem}
\label{lem:distil.ec}
Let $\pi^\ec$, $\overline{\aR}^{k,\delta,p}_{\min}$, and
$\aR^{k,\delta,p}_{\min}$ be as described above. Then
\begin{equation*}
%\label{eq:distil.ec.2}
\overline{\aR}^{k,\delta,p}_{\min} \| \aA^c
\xrightarrow{\pi^\ec,0} \aR^{k-r-t,\delta,p}_{\min}.\end{equation*}
\end{lem}

\begin{proof} Immediate from
  \lemref{lem:tech.min-entropy-inequalities.chain}, and the properties
  of the error correction code chosen.\end{proof}

Note that if $(X,Y) \subset \cS$ only holds with probability
$1 - \eps$, then one does not have the resource
$\overline{\aR}^{k,\delta,p}_{\min}$, but a less specific resource
$\ball{\eps}{\overline{\aR}^{k,\delta,p}_{\min}}$, and using the fact
that for any $\aS$, $\aS^\eps \xrightarrow{\id,\eps} \aS$, and
\thmref{thm:security}, one has
\[ \ball{\eps}{\overline{\aR}^{k,\delta,p}_{\min}} \| \aA^c
\xrightarrow{\pi^\ec,\eps} \aR^{k-r-t,\delta,p}_{\min}.\]

% As pointed out in \remref{rem:distil.ec}, if the players use secure
% channels instead of authentic channels we get,
% \[ \overline{\aR}^{k,\delta,p}_{\min} \| \aS^c
% \xrightarrow{\pi^\ec,0} \aR^{k,\delta,p}_{\min}\]
% instead of \eqnref{eq:distil.ec.2}.

\subsection{Privacy Amplification}
\label{sec:distil.pa}

After the error correction scheme, the players share a string that is
partially known to the adversary. To obtain a secret string, Alice
picks a string $Z$ uniformly at random, a \emph{seed}, and sends it to
Bob on an authentic channel. They then each compute the strings $K =
\Ext(X,Z)$ and $K' = \Ext(Y,Z)$ for some predefined function $\Ext$,
known as an \emph{extractor}. This procedure is called privacy
amplification~\cite{BBCM95,RK05}, and we denote the corresponding
pair of converters as $\pi^\pa$. The function $\Ext$ has to satisfy
the following property.

\begin{deff}
\label{def:extractor}
A function $\Ext : \{0,1\}^n \times \{0,1\}^d \to \{0,1\}^m$ is a
\emph{quantum\-/proof $(k,\eps)$\-/strong extractor for
subnormalized states} if for any subnormalized state $\rho_{XE}$
classical on $X$, with $\Hmin[\rho]{X|E} \geq k$, and a uniform $Z$,
we have
\[ \frac{1}{2} \trnorm{\rho_{\Ext(X,Z)ZE} - \tau_K \tensor \tau_Z \tensor
      \rho_E} \leq \eps,\] where $\tau_K$ is the fully mixed state.
\end{deff}

The more standard definition for quantum\-/proof extractors
only requires them to be defined for normalized states. But as shown
in \lemref{lem:tech.extractors.sub}, any extractor for normalized
states is also an extractor for subnormalized states with slightly
weaker parameters. Efficient constructions of extractors are given in,
e.g., \cite{RK05,FS08,TSSR11,HT16,DPVR12}. Note that some of these, e.g.,
the universal hashing extractors~\cite{RK05,TSSR11,HT16}, directly
satisfy \defref{def:extractor}.

We are now ready to state the main theorem for this section.

\begin{thm}
\label{thm:distil.pa}
Let $\pi^\pa$ be a privacy amplification protocol as described above,
which uses a quantum\-/proof $(k,\eps_\pa)$\-/strong extractor for
subnormalized states. Given a $(k,\delta)$\-/min\-/entropy resource
$\aH^{k,\delta}_{\min}$ which provides Alice and Bob with identical
strings $X = Y$, and an authentic channel $\aA^c$, $\pi^\pa$
constructs am $m$-bit secret key resource $\aK^m$ (see
\figref{fig:key.sim}) with error $\eps_\pa+2\delta$,
\[ \aH^{k,\delta}_{\min} \| \aA^c \xrightarrow{\pi^\pa,\eps_\pa+2\delta} \aK^m.\]
\end{thm}

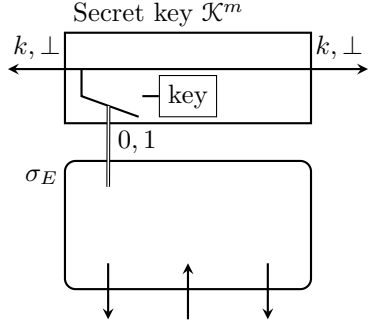
\begin{figure}[tb]
\centering
\begin{tikzpicture}[
sArrow/.style={->,>=stealth,thick},
thinResource/.style={draw,thick,minimum width=1.618*2cm,minimum height=1.2cm},
simulator/.style={draw,thick,rounded corners,minimum width=1.618*2cm,minimum height=1.7cm},
innersnode/.style={minimum width=.4cm,minimum height=.3cm}]

\small

\def\t{2.368} % 1.618+.75
\def\u{-1.95} % .6+.5+1.7/2
\def\w{-3.2} % .5+.5+1.7+.5
\def\v{.118} % 1/1.618-.5
\def\s{1.05}

\node[thinResource] (keyBox) at (0,0) {};
\node[draw] (key) at (0,\v/2-.3) {key};
\node (junc) at (-1.4,0 |- key.center) {};
\node[yshift=-1.5,above right] at (keyBox.north west) {Secret key $\aK^m$};
\node (alice) at (-\t,\v) {};
\node (bob) at (\t,\v) {};

\draw[sArrow,<->] (alice.center) to node[pos=.08,auto] {$k,\bot$} node[pos=.92,auto] {$k,\bot$} (bob.center);
\draw[thick] (junc.center |- 0,\v) to (junc.center) to node[pos=.472] (handle) {} +(160:-.8);
\draw[thick] (-.6,0 |- junc.center) to (key);

\node[simulator] (sim) at (0,\u) {};
\node[xshift=1.5,below left] at (sim.north west) {$\sigma_E$};
\node[innersnode] (a1) at (-\s,\u+.35) {};
\node[innersnode] (a2) at (-\s,\u-.35) {};
\node[innersnode] (b2) at (\s,\u-.35) {};
\node[innersnode] (c2) at (0,\u-.35) {};

\node (evel) at (-\s,\w) {};
\node (evec) at (0,\w) {};
\node (ever) at (\s,\w) {};

\draw[double] (a1) to node[pos=.55,auto,swap] {$0,1$} (handle.center);
\draw[sArrow] (a2) to (evel.center);
\draw[sArrow] (evec.center) to (c2);
\draw[sArrow] (b2) to (ever.center);

\end{tikzpicture}
\caption[Secret key with simulator]{\label{fig:key.sim}An $m$-bit
  secret key resource $\aK^m$ connected to a simulator $\sigma_E$ at
  Eve's interface. As in \secref{sec:cc.security}, $\aK^m$ does not
  spontaneously output the keys at the players' interfaces, but only
  after having received a request, which is omitted from the drawing.}
\end{figure}

Unlike previous proofs, the current one requires a simulator,
$\sigma_E$, which we draw in \figref{fig:key.sim}. Note that the
number of messages exchanged at the outer interface of $\sigma_E$
(i.e., the number of arrows in \figref{fig:key.sim}) will depend on
the atomic resource $\sR \in \pi^\pa(\aH^{k,\delta}_{\min} \| \aA^c)$
from the real system. For example, in the case of a min\-/entropy
resource constructed using a QKD protocol, this would be three
messages as drawn in \figref{fig:key.sim}: the quantum states
generated by Alice are output at Eve's interface, the modified states
are then input by Eve, and finally a transcript of the classical
communication is given to Eve as well (upon request).

\begin{proof}
  To prove that
  $\pi^\pa(\aH^{k,\delta}_{\min} \| \aA^c) \subset
  \ball{\eps_\pa}{\left(\aK^m\right)^*}$,
  for every atomic resource
  $\pi^\pa(\sH \| \sA) \in \pi^\pa(\aH^{k,\delta}_{\min} \| \aA^c)$ we
  will find a $\sK$ and $\sigma_E$ such that
  $\pi^\pa(\sH \| \sA) \close{\eps_\pa} \sK \sigma_E$. Since $\aA^c$
  and $\aK^m$ both have cardinality $1$, we only need to find a
  simulator $\sigma_E$ for every $\sH$. We construct $\sigma_E$ as
  follows. $\sigma_E$ runs $\sH$ internally and generates all the
  communication at the outer interface. If $\sH$ outputs $\bot$,
  $\sigma_E$ notifies the key resource $\sK$ to output an error $\bot$
  as well. If $\sH$ outputs a string $X \neq \bot$, then it picks a
  random seed $Z$, which is made available to be output at its outer
  interface as the transcript on the authentic channel. It notifies
  $\sK$ to generate a uniform key $k$, which is output at Alice's and
  Bob's interfaces when requested.

  Let a distinguisher interact with the real system. After interacting
  with $\sH$ it gets a system $E$, then obtains $Z$ and finally an
  output $K$ from Alice's and Bob's interface. We take $Z = \bot$ in
  case $\sH$ generated a string $X = \bot$. The distinguisher then
  holds the state
  \[\rho_{KZE} = \proj{\bot,\bot} \tensor \rho^\bot_E +
  \rho^{\top}_{\Ext(X,Z)ZE}.\]
  In the case where the distinguisher is interacting with
  $\sK\sigma_E$, it is going to get the state 
  \[\tilde{\rho}_{KZE} = \proj{\bot,\bot} \tensor \rho^\bot_E +
  \tau_K \tensor \tau_{Z} \tensor \rho^{\top}_{E}.\]

  The distinguishability between the two systems is bounded by the
  trace distance between $\rho_{KZE}$ and $\tilde{\rho}_{KZE}$, namely
  \[\frac{1}{2} \trnorm{\rho_{KZE} - \tilde{\rho}_{KZE}} = \frac{1}{2}
  \trnorm{\rho^{\top}_{\Ext(X,Z)ZE} - \tau_K \tensor \tau_{Z} \tensor
    \rho^{\top}_{E}} \leq \eps_\pa+2\delta,\]
  where the last inequality uses the fact that
  $\HminSmooth[\rho^\top]{\delta}{X|E} \geq k$, that $\Ext$ is a
  quantum\-/proof $(k,\eps_\pa)$\-/strong extractor for subnormalized
  states and \lemref{lem:tech.extractors.smooth}, which shows that a
  $(k,\eps)$\-/extractor has error $\eps+2\delta$ if used on states
  with a bound on their smooth min\-/entropy instead of their
  min\-/entropy.
\end{proof}

\subsection{Composed statement}
\label{sec:distil.comp}

By composing \corref{cor:distil.ec}, \lemref{lem:distil.ec} and
\thmref{thm:distil.pa}, we get the following.

\begin{cor}
\label{cor:distil.comp}
Let $\overline{\aH}^{k,\delta}_{\min}$ be a
$(k,\delta)$\-/min\-/entropy resource that produces an arbitrary
output at Bob's interface, let
$\overline{\aR}^{k,\delta}_{\min} \subset
\overline{\aH}^{k,\delta}_{\min}$
be a $(k,\delta)$\-/min\-/entropy resource that never aborts and
furthermore, produces strings $(X,Y) \subset \cS$ that are always
corrected by the functions $(\synd,\corr)$. Let $\underline{\aK}^m$ be
a secret key resource that always provides the players with a uniform
key of length $m$ and let $\aK^m$ be a secret key resource that only
provides a key if the adversary allows it. Finally, let
$\pi = \pi^{\pa} \circ \pi^\ec$ be the
construction described in this section. Then
\begin{gather*}
\overline{\aR}^{k,\delta}_{\min} \| \aA^c
\xrightarrow{\pi,\eps_\pa +2\delta} \underline{\aK}^m \\
\intertext{and}
\overline{\aH}^{k,\delta}_{\min} \| \aA^c
\xrightarrow{\pi,\eps_\verif+\eps_\pa+2\delta} \aK^m,
\end{gather*}
where $\aA^c = \aA^{c_1} \| \aA^{c_2}$, and $\aA^{c_i}$ are the
authentic channels used by the different parts of the
protocol.\end{cor}

\section{Quantum Authentication of Classical Messages}
\label{sec:FS}

Fehr and Salvail~\cite{FS17a,FS17b} propose a prepare\-/and\-/measure
protocol that authenticates a classical message by encoding it into a
quantum cipher. The main feature of their protocol is that it recycles
all the key if the message is accepted by the receiver. In this
section we provide a composable security analysis of a slightly
modified version of their protocol, which is adapted
from~\cite{FS17b}. We first provide a high level view of what the
protocol achieves in the AC framework in
\secref{sec:FS.construction}. Then in \secref{sec:FS.protocol} we give
the details of the protocol. The security proof is provided in
\secref{sec:FS.analysis}. And finally, in \secref{sec:FS.continuing}
we prove bounds on the security if the compromised key is replaced
instead of being discarded.

\subsection{The Construction}
\label{sec:FS.construction}

In this section we model the Fehr-Salvail protocol as a pair of
converters $\pi = (\pi_A,\pi_B)$ that construct some resource $\aS$
given some other resource $\aR$. As in \secref{sec:distil}, there are
two statements we wish to make. The first is when an adversary is
present and the players share an insecure quantum channel $\aQ$ as
drawn in \figref{fig:resource.insecure}, i.e., the adversary can
change the message being sent and insert a message of her own. We
model the real and ideal systems for this case in
\secref{sec:FS.construction.real} and
\secref{sec:FS.construction.ideal}, respectively.  The second setting
is when no adversary is present, but the channel shared has some
natural noise, as in \figref{fig:resource.noisy}. This case is
discussed in \secref{sec:FS.construction.noise}.

\begin{figure}[tb]
\subcaptionbox[Insecure channel]{\label{fig:resource.insecure}An insecure channel
  from Alice (on the left) to Bob (on the right) allows Eve (below) to
  intercept the message and insert a message of her own.}[.5\textwidth][c]{
\begin{tikzpicture}[
resource/.style={draw,thick,minimum width=3.2cm,minimum height=1cm},
filter/.style={draw,thick,minimum width=1.6cm,minimum height=.8cm},
sArrow/.style={->,>=stealth,thick},
sLine/.style={-,thick}]

\small

\def\t{2.2} %.1.6+.7
\def\v{1.4} %.5+.5+.4
\def\w{1.2} %.5+.7
\def\s{.4}

\node[resource] (ch) at (0,0) {};
\node[yshift=-1.5,above right] at (ch.north west) {
  Insecure channel $\aQ$};
\node (alice) at (-\t,0) {};
\node[yshift=-1.5,below] at (alice) {\footnotesize
  Alice};
\node (bob) at (\t,0) {};
\node[yshift=-1.5,below] at (bob) {\footnotesize
  Bob};
\node (eve) at (0,-\w) {\footnotesize
  Eve};

\draw[sArrow] (alice.center) to node[auto,pos=.15] {$\rho$} (-\s,0) to (-\s,0 |- eve.north);
\draw[sArrow] (\s,0 |- eve.north) to (\s,0) to node[auto,pos=.85] {$\rho'$}  (bob.center);

\end{tikzpicture}}
\subcaptionbox[Insecure channel without
adversary]{\label{fig:resource.noisy}When no adversary is present,
  Alice's message is delivered to Bob, but some natural noise is
  inserted.}[.5\textwidth][c]{
\begin{tikzpicture}[
resource/.style={draw,thick,minimum width=3.2cm,minimum height=1cm},
filter/.style={draw,thick,minimum width=1.6cm,minimum height=.8cm,rounded corners},
sArrow/.style={->,>=stealth,thick},
sLine/.style={-,thick}]

\small

\def\t{2.3} %.1.6+.7
\def\v{1.4} %.5+.5+.4
\def\w{1.2} %.5+.7
\def\s{1}

\node[resource] (ch) at (0,0) {};
\node[yshift=-1.5,above right] at (ch.north west) {
  Noisy channel $\aR^\eps_\varphi$};
\node (alice) at (-\t,0) {};
\node[yshift=-1.5,below] at (alice) {\footnotesize
  Alice};
\node (bob) at (\t,0) {};
\node[yshift=-1.5,below] at (bob) {\footnotesize
  Bob};
\node (eve) at (0,-\w) {\footnotesize
  Eve};

\draw[sLine] (alice.center) to node[auto,pos=.15] {$\rho$} (-\s,0);
\draw[sLine,decorate,decoration={snake,amplitude=1mm,segment length=1.8mm}] (-\s,0) to (\s,0);
\draw[sArrow] (\s,0) to node[auto,pos=.85] {$\cE_\varphi(\rho)$} (bob.center);

\end{tikzpicture}}

\vspace{6pt}

\subcaptionbox[Authentic channel]{\label{fig:resource.auth.ideal}An
  idealized authentic channel from Alice to Bob allows Eve to receive
  a copy of the message, but guarantees that Bob always gets
  it.}[.5\textwidth][c]{
\begin{tikzpicture}[
resource/.style={draw,thick,minimum width=3.2cm,minimum height=1cm},
filter/.style={draw,thick,minimum width=1.6cm,minimum height=.8cm},
sArrow/.style={->,>=stealth,thick},
sLine/.style={-,thick}]

\small

\def\t{2.2} %.1.6+.7
\def\v{1.4} %.5+.5+.4
\def\w{1.2} %.5+.7
\def\s{.4}

\node[resource] (ch) at (0,0) {};
\node[yshift=-1.5,above] at (ch.north) {
  Authentic channel $\underline{\aA}$};
\node (alice) at (-\t,0) {};
\node[yshift=-1.5,below] at (alice) {\footnotesize
  Alice};
\node (bob) at (\t,0) {};
\node[yshift=-1.5,below] at (bob) {\footnotesize
  Bob};
\node (eve) at (0,-\w) {\footnotesize
  Eve};

\draw[sArrow] (alice.center) to node[auto,pos=.07] {$x$} node[auto,pos=.93] {$x$}  (bob.center);
\draw[sArrow] (0,0) to (eve.north);

\end{tikzpicture}}
\subcaptionbox[Authentic channel with
switch]{\label{fig:resource.authentic}A realistic authentic channel from
  Alice to Bob allows Eve to receive a copy of the message and choose
  whether Bob receives it or an error symbol.}[.5\textwidth][c]{
\begin{tikzpicture}[
resource/.style={draw,thick,minimum width=3.2cm,minimum height=1cm},
filter/.style={draw,thick,minimum width=1.6cm,minimum height=.8cm},
sArrow/.style={->,>=stealth,thick},
sLine/.style={-,thick}]

\small

\def\t{2.2} %.1.6+.7
\def\v{1.4} %.5+.5+.4
\def\w{1.2} %.5+.7
\def\s{.4}

\node[resource] (ch) at (0,0) {};
\node[yshift=-1.5,above] at (ch.north) {
  Authentic channel $\aA$};
\node (alice) at (-\t,0) {};
\node[yshift=-1.5,below] at (alice) {\footnotesize
  Alice};
\node (bob) at (\t,0) {};
\node[yshift=-1.5,below] at (bob) {\footnotesize
  Bob};
\node (eve) at (0,-\w) {\footnotesize
  Eve};

\draw[sLine] (alice.center) to node[auto,pos=.15] {$x$} (0,0) to node (handle) {} (330:2*\s);
\draw[sArrow] (2*\s,0) to node[auto,pos=.85] {$x,\bot$}  (bob.center);
\draw[sArrow] (-\s,0) to (-\s,0 |- eve.north);
\draw[double,thick] (handle.center |- eve.north) to node[auto,swap,pos=.25] {$0,1$} (handle.center);

\end{tikzpicture}}

\caption[Various resources]{\label{fig:resources}Some
  resources needed in the Fehr-Salvail construction.}
\end{figure}

\subsubsection{Real System}
\label{sec:FS.construction.real}

The protocol $\pi = (\pi_A,\pi_B)$ has a very simple structure. First
$\pi_A$ encodes a classical message in a quantum cipher and sends it
to Bob on the insecure channel $\aQ$. $\pi_B$ decodes, and either
accepts or rejects the message. It also recycles some key depending on
the decision. $\pi_B$ then sends one bit of information back to Alice,
to notify her of whether the message was accepted or rejected, after
which $\pi_A$ recycles Alice's keys.

To send this bit of information to Alice, we assume that the players
share a backward authentic channel $\underline{\aA}^1$. As in
\secref{sec:distil}, this is defined as a channel that always
transmits the message, but provides a copy to Eve if she requests
it. This is illustrated in \figref{fig:resource.auth.ideal} (but to
simplify the drawing, we do not draw the request that Eve must make to
get the transcript of the authentic channel).

The final resources the players are going to need are the key
resources. The protocol uses three keys,
$(\ell_\sss,\ell_\mac,\theta)$. $\ell_\sss$ and $\ell_\mac$ need to be
uniform, so we provide Alice and Bob with a shared uniform key
resource $\aK$, as depicted in \figref{fig:qkd.ideal}, which generates
these keys. Note that if Eve does not allow the players to get a key,
they cannot run the protocol, and it is trivially secure. Thus, in the
following, we assume the switch is set by Eve to provide the players
with a key.

The key $\theta$ does not need to be uniform, it is sufficient if it
has bounded min\-/entropy. So we provide the players with a
$(k,0)$\-/min\-/entropy resource $\aH^{k,0}_{\min}$ as defined in
\secref{sec:min-entropy.spec}. Since we always take the smoothing
parameter to be $0$ in this section, we denote this resource
by $\aH^{k}_{\min}$. Note that to satisfy sequential scheduling, both
the key resources $\aK$ and $\aH^{k}_{\min}$ only provide the keys (or
error messages) to the players when they request them (see
\remref{rem:scheduling}, and the discussion of these resources in
Sections~\ref{sec:cc.security} and \ref{sec:min-entropy.spec}).

Putting this together, we get the real system drawn in
\figref{fig:FS.real}. (The switch on $\aK$ as well as the requests to
get keys and messages have been omitted from the picture.)

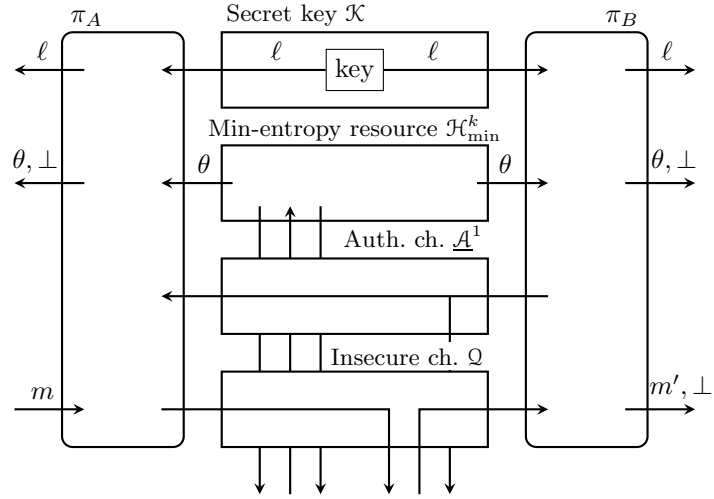
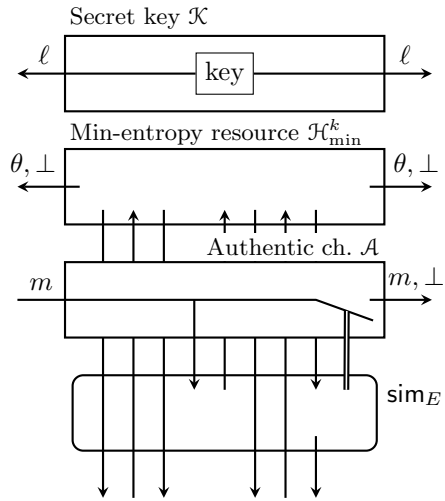
\begin{figure}[tbp]
\begin{subfigure}{\textwidth}
\centering
\begin{tikzpicture}[
resourceLong/.style={draw,thick,minimum width=3.5cm,minimum height=1cm},
resource/.style={draw,thick,minimum width=1.8cm,minimum height=1cm},
sArrow/.style={->,>=stealth,thick},
sLine/.style={-,thick},
protocol/.style={draw,thick,minimum width=1.6cm,minimum height=5.5cm,rounded corners},
pnode/.style={minimum width=1cm,minimum height=.6cm}]

\small

\def\t{4.6} % 1.75+.5+1.6+.75
\def\a{3.05} % 1.75+.5+1.6/2
\def\v{.75}
\def\w{.4}
\def\x{.85}
\def\z{3.5} %5.5/2+.75

\node[protocol] (protA) at (-\a,0) {};
\node[yshift=-1.5,above right] at (protA.north west) {$\pi_A$};
\node[pnode] (a1) at (-\a,3*\v) {};
\node[pnode] (a2) at (-\a,\v) {};
\node[pnode] (a3) at (-\a,-\v) {};
\node[pnode] (a4) at (-\a,-3*\v) {};
\node[protocol] (protB) at (\a,0) {};
\node[yshift=-1.5,above left] at (protB.north east) {$\pi_B$};
\node[pnode] (b1) at (\a,3*\v) {};
\node[pnode] (b2) at (\a,\v) {};
\node[pnode] (b3) at (\a,-\v) {};
\node[pnode] (b4) at (\a,-3*\v) {};

\node (aliceUp) at (-\t,3*\v) {};
\node (aliceMiddle) at (-\t,\v) {};
\node (aliceDown) at (-\t,-3*\v) {};
\node (bobUp) at (\t,3*\v) {};
\node (bobMiddle) at (\t,\v) {};
\node (bobDown) at (\t,-3*\v) {};
\node (eve1) at (-\x-\w,-\z) {};
\node (eve2) at (-\x,-\z) {};
\node (eve3) at (-\x+\w,-\z) {};
\node (eve4) at (\x-\w,-\z) {};
\node (eve5) at (\x,-\z) {};
\node (eve6) at (\x+\w,-\z) {};

\node[resourceLong] (key) at (0,3*\v) {};
\node[above right,inner sep=2] at (key.north west) {\footnotesize
  Secret key $\aK$};

\draw[sArrow] (a1) to node[auto,swap,pos=.6] {$\ell$} (aliceUp);
\draw[sArrow] (b1) to node[auto,pos=.6] {$\ell$} (bobUp);
\node[draw] (key) at (0,3*\v) {key};
\draw[sArrow] (key) to node[auto,swap,pos=.3] {$\ell$} (a1);
\draw[sArrow] (key) to node[auto,pos=.3] {$\ell$} (b1);

\node[resourceLong] (entropy) at (0,\v) {};
\node[yshift=-1,above,inner sep=2] at (entropy.north) {\footnotesize
  Min-entropy resource $\aH^{k}_{\min}$};

\draw[sArrow] (a2) to node[auto,swap,pos=.7] {$\theta,\bot$} (aliceMiddle);
\draw[sArrow] (b2) to node[auto,pos=.7] {$\theta,\bot$} (bobMiddle);
\node[pnode] (kLeft) at (-1.1,\v) {};
\node[pnode] (kRight) at (1.1,\v) {};
\draw[sArrow] (kLeft) to node[auto,pos=.4,swap] {$\theta$} (a2);
\draw[sArrow] (kRight) to node[auto,pos=.4] {$\theta$} (b2);
\draw[sArrow] (eve1 |- kLeft.south) to (eve1);
\draw[sArrow] (eve2) to (eve2 |- kLeft.south);
\draw[sArrow] (eve3 |- kLeft.south) to (eve3);

\node[resourceLong,fill=white] (auth) at (0,-\v) {};
\node[above left,inner sep=2] at (auth.north east) {\footnotesize
  Auth.\ ch.\ $\underline{\aA}^1$};

\draw[sArrow] (b3) to (a3);
\draw[sArrow] (eve6 |- b3) to (eve6);

\node[resourceLong,fill=white] (ch) at (0,-3*\v) {};
\node[above left,fill=white,inner sep=2] at (ch.north east) {\footnotesize
  Insecure ch.\ $\aQ$};

\draw[sArrow] (aliceDown) to node[auto,pos=.4] {$m$} (a4);
\draw[sArrow] (a4) to (eve4 |- aliceDown) to (eve4);
\draw[sArrow] (eve5) to (eve5 |- bobDown) to (b4);
\draw[sArrow] (b4) to node[auto,pos=.8] {$m',\bot$} (bobDown);

\end{tikzpicture}
\caption[Real system]{\label{fig:FS.real}The real system consists
  of the resources consumed and the honest players converters
  $(\pi_A,\pi_B)$ that run the protocol.}
\end{subfigure}

\vspace{9pt}

\begin{subfigure}{\textwidth}
\centering
\begin{tikzpicture}[
resourceLong/.style={draw,thick,minimum width=4.2cm,minimum height=1cm},
resource/.style={draw,thick,minimum width=1.8cm,minimum height=1cm},
sArrow/.style={->,>=stealth,thick},
sLine/.style={-,thick},
simulator/.style={draw,thick,minimum width=4cm,minimum height=1cm,rounded corners},
snode/.style={minimum width=3.8cm,minimum height=.6cm}]

\small

\def\ta{2.85} % 4.2/2+.75
\def\tb{2.85} % 4.2/2+.75
\def\v{.75}
\def\w{.4}
\def\x{1.2}
\def\z{3.5} %5.5/2+.75

\node (aliceUp) at (-\ta,3*\v) {};
\node (aliceMiddle) at (-\ta,\v) {};
\node (aliceDown) at (-\ta,-\v) {};
\node (bobUp) at (\tb,3*\v) {};
\node (bobMiddle) at (\tb,\v) {};
\node (bobDown) at (\tb,-\v) {};
\node (eve1) at (-\x-\w,-\z) {};
\node (eve2) at (-\x,-\z) {};
\node (eve3) at (-\x+\w,-\z) {};
\node (eve4) at (-\w,-\z) {};
\node (eve5) at (0,-\z) {};
\node (eve6) at (\w,-\z) {};
\node (eve7) at (\x-\w,-\z) {};
\node (eve8) at (\x,-\z) {};
\node (eve9) at (\x+\w,-\z) {};

\node[simulator] (sim) at (0,-3*\v) {};
\node[below right] at (sim.north east) {$\simE$};
\node[snode] (innerS) at (0,-3*\v) {};

\node[resourceLong] (key) at (0,3*\v) {};
\node[above right,inner sep=2] at (key.north west) {\footnotesize
  Secret key $\aK$};
\node[draw] (key) at (0,3*\v) {key};

\draw[sArrow] (key) to node[auto,swap,pos=.85] {$\ell$} (aliceUp);
\draw[sArrow] (key) to node[auto,pos=.85] {$\ell$} (bobUp);

\node[resourceLong] (entropy) at (0,\v) {};
\node[inner sep=2,yshift=-1,above right] at (entropy.north west) {\footnotesize
  Min-entropy resource $\aH^{k}_{\min}$};
\node[snode] (innerR) at (0,\v) {};

\draw[sArrow] (innerR) to node[auto,swap,pos=.75] {$\theta,\bot$} (aliceMiddle);
\draw[sArrow] (innerR) to node[auto,pos=.75] {$\theta,\bot$} (bobMiddle);
\draw[sArrow] (eve1 |- innerR.south) to (eve1);
\draw[sArrow] (eve2) to (eve2 |- innerR.south);
\draw[sArrow] (eve3 |- innerR.south) to (eve3);
\draw[sArrow] (eve5 |- innerS.north) to (eve5 |- innerR.south);
\draw[sArrow] (eve6 |- innerR.south) to (eve6);
\draw[sArrow] (eve7) to (eve7 |- innerR.south);
\draw[sArrow] (eve8 |- innerR.south) to (eve8 |- innerS.north);
\draw[sArrow] (eve8 |- innerS.south) to (eve8);

\node[resourceLong,fill=white] (auth) at (0,-\v) {};
\node[above left,fill=white,inner sep=2] at (auth.north east) {\footnotesize
  Authentic ch.\ $\aA$};

\draw[sLine] (aliceDown) to node[auto,pos=.08] {$m$} (eve8 |-
aliceDown) to node[pos=.54] (junc) {} +(340:2*\w);
\draw[sLine,double] (eve9 |- innerS.north) to (junc.center);
\draw[sArrow] (innerS.east |- bobDown) to node[auto,pos=.75] {$m,\bot$} (bobDown);
\draw[sArrow] (eve4 |- aliceDown) to (eve4 |- innerS.north);

\end{tikzpicture}
\caption[Ideal system]{\label{fig:FS.ideal}The ideal system consists
  of the constructed resources and the simulator $\simE$.}
\end{subfigure}

\caption[The modified FS construction]{\label{fig:FS}The real and ideal
  systems of the modified FS protocol in the case where Eve is
  present.}
\end{figure}

\subsubsection{Ideal system}
\label{sec:FS.construction.ideal}

The goal of the protocol is to construct an authentic channel. But the
channel $\underline{\aA}$ depicted in \figref{fig:resource.auth.ideal}
always delivers the message, and this cannot be achieved, since Eve
can jumble the communication if desired. What is achieved is a
slightly weaker resource $\aA$, which allows Eve to choose if Bob gets
the message or not \--- here too, the receiver must request the
message for it to be output, it is not spontaneously output when Eve
allows it. This is illustrated in \figref{fig:resource.authentic}
without the request arrow for the message to be delivered.

The protocol also recycles keys, which is modeled as key resources in
the ideal system. All the uniform key $(\ell_\sss,\ell_\mac)$ is
recycled, so the ideal system must have a resource $\aK$ as well,
which provides the players with fresh uniform keys (independent of the
messages and ciphers obtained by Eve).

As we show in \thmref{thm:FS} in \secref{sec:FS.analysis}, the
subnormalized key with distribution
$P_{\Theta'}(\theta) \coloneqq \Pr \left[\Theta = \theta \text{ and
    ``accept''}\right]$
that is output by players the if the message is accepted has the same
min\-/entropy as the original normalized key $\Theta$. So the recycled
key in the ideal system is captured by a specification $\aH^k_{\min}$
with the same entropy bound $k$. Note that the renormalized key with
distribution
$P_{\Theta'}(\theta) \coloneqq \Pr \left[\Theta = \theta
  \middle|\text{``accept''}\right]$
may have much less entropy, so there is effectively an entropy loss
during the protocol. But the relevant measure of entropy for future
uses of the key is that of the subnormalized state \--- since the
resource consumed, $\aH^k_{\min}$, measures the entropy of the
subnormalized state.

Putting this together along with a simulator $\simE$ that is needed
for the security proof, we get the ideal system drawn in
\figref{fig:FS.ideal}.

The corresponding constructive statement formulated in the AC
framework is that $\pi$ constructs $\aK \| \aH^k_{\min} \| \aA$ from
$\aK \| \aH^k_{\min} \| \underline{\aA}^1 \| \aQ$ (with some error
$\eps_\adv$),
\begin{equation}
\label{eq:FS.adv}
\aK \| \aH^k_{\min} \| \underline{\aA}^1 \| \aQ \xrightarrow{\pi,\eps_\adv} \aK \|
\aH^k_{\min} \| \aA.\end{equation}
The key resources $\aK$ and $\aH^k_{\min}$ may
be seen as catalysts, since they appear both in the real and ideal
systems.

\subsubsection{Natural Noise}
\label{sec:FS.construction.noise}

In the case where Eve is not present, the players share a channel that
has natural noise instead of the insecure channel $\aQ$. More
specifically, the protocol is designed to correct $\varphi$ errors,
so we assume a channel specification that is $\eps$-close to having
$\varphi$ errors when the states sent are encoded as in the
protocol.\footnote{We model the channel as only adding noise to the
  quantum states sent, but assume that the classical part of the
  cipher is transmitted without errors. One can easily compose a noisy
  classical channel with converters that perform error correction if
  this is not the case (see \figref{fig:ecc}), i.e., construct a
  noiseless channel from a noisy one. We do not want to do this for a
  quantum channel, because quantum error correcting codes are not
  prepare\-/and\-/measure.} We denote this channel specification by
$\aR^\eps_\varphi$, and illustrate this in
\figref{fig:resource.noisy}.

Furthermore, in this case we wish to make a stronger statement, namely
that the message is delivered for sure, i.e., a channel of the type
$\underline{\aA}$ from \figref{fig:resource.auth.ideal} is
constructed. But to achieve this, it is not sufficient that the
communication is not jumble by Eve, we also need to be sure that the
players always get keys. We denote by $\underline{\aK}$ and
$\underline{\aH}^k_{\min}$ resources which always provide keys to the
players. In this case, the recycled keys are of the same type.

The corresponding constructive statement is that $\pi$ constructs
$\underline{\aK} \| \underline{\aH}^k_{\min} \| \underline{\aA}$ from
$\underline{\aK} \| \underline{\aH}^k_{\min} \| \underline{\aA}^1 \|
\aR^\eps_\varphi$ (with some error $\eps_\noise$),
\begin{equation}
\label{eq:FS.noise}
\underline{\aK} \| \underline{\aH}^k_{\min} \| \underline{\aA}^1 \| \aR^\eps_\varphi
\xrightarrow{\pi,\eps_\noise} \underline{\aK} \| \underline{\aH}^k_{\min} \| \underline{\aA}.\end{equation}

\subsection{Protocol}
\label{sec:FS.protocol}

The Fehr-Salvail~\cite{FS17a} protocol requires two different (keyed)
hash functions, which we denote
$\sss : \{0,1\}^{n_\sss} \times \{0,1\}^{r_\sss} \to \{0,1\}^{m_\sss}$
and
$\mac : \{0,1\}^{n_\mac} \times \{0,1\}^{r_\mac} \to
\{0,1\}^{m_\mac}$.
Both of these functions are constructed from extractors\footnote{The
  definition of an extractor is given in \defref{def:extractor}.}
$\Ext_h : \{0,1\}^{n_h} \times \{0,1\}^{r_h - m_h} \to \{0,1\}^{m_h}$,
for $h \in \{\sss,\mac\}$. The hash functions are defined as
$h(x,\ell_1 \| \ell_2) \coloneqq \Ext_h(x,\ell_1) \xor \ell_2$.

Note that not any extractor will do. They have to be quantum\-/proof
$(k,\eps)$\-/strong extractors for subnormalized states for any value
$k$ and an error $\eps = \frac{\nu_h}{2}\sqrt{2^{-k+m_h}}$, where
$\nu_h$ are parameters specific to each extractor. Universal
hashing~\cite{RK05}, almost universal hashing~\cite{TSSR11}, dual
universal hashing~\cite{HT16}, as well as $\delta$\-/biased
masking~\cite{FS08} all satisfy this requirement (for parameters
$\nu_h$ which depend on the specific constructions).  To obtain
composable security, we additionally require that these extractors be
linear functions in the first input~\cite{FS17b}, which is satisfied
by all the constructions cited above.

The reason for this specific form of extractor is that the security
proof requires the following property.
\begin{deff}[$\nu$\-/key\-/privacy~\cite{FS17a}]
\label{def:key-privacy}
A function $h : \cX \times \cL \to \cT$ offers $\nu$\-/key\-/privacy
if for any subnormalized state $\rho_{LXTE} \in \sno{LXTE}$ with the
properties $\rho_{LX} = \tau_L \tensor \rho_X$, $T = h(X,L)$, and
$L \leftrightarrow XT \leftrightarrow E$ forms a Markov chain, it
holds that
\[\trnorm{\rho_{LTE} - \tau_L \tensor \rho_{TE}} \leq
\nu \sqrt{2^{-\Hmin{X|TE} + m}},\]
where $\tau_L$ is the fully mixed state and $\cT = \{0,1\}^m$.
\end{deff}
Fehr and Salvail~\cite{FS17a} prove that if $\Ext_h$ is an extractor of
the type given above, then the corresponding function $h$ provides
$\nu_h$\-/key\-/privacy (see \lemref{lem:key-privacy}).

Note that a function like $h$ obtained by XORing a uniform
key always produces a uniform output.
\begin{deff}[Uniformity~\cite{FS17a}]
\label{def:uniformity}
We say that a function $h : \cX \times \cL \to \cT$ is \emph{uniform}
if for a uniform $L$ and any $x \in \cX$, $T = h(x,L)$ is uniformly
random on $\cT$.
\end{deff}

% The first of these extractors is used to extract a uniform string that
% one-time pads the message in the protocol. We denote by $f : \cX_f
% \times \cL_f \to \cT_f$ the corresponding function defined as
% $f(x,\ell \| t) \coloneqq \Ext_1(x,\ell) \xor t$, and its
% parameter by $\nu_f$. Any extractor with $\eps =
% \frac{\nu_f}{2}\sqrt{2^{-k+m_f}}$ can be used for this function.

The first keyed hash function, (a \emph{secure sketch}) $\sss$, is
used for error correction. Apart from the underlying extractor having
the required properties, we additionally need that for a uniformly
chosen $\ell \in \{0,1\}^{r_{\sss}}$ and any
$x,x' \in \{0,1\}^{n_\sss}$ with Hamming distance
$w(x,x') \leq \varphi n_\sss$ for some fixed $\varphi$, $x$ may be
recovered from $(x',\sss(x,\ell),\ell)$ except with probability
$\eps_\sss$. This can be implemented using universal
hashing~\cite{Ren05}, but the recovery operation is not known to be
efficient. Fehr and Salvail propose another construction based in
$\delta$\-/biased masking~\cite{FS08}, which has an efficient recovery
operation, but has less noise tolerance $\varphi$ (see \cite{FS17a}
for more details).

The second keyed hash function, $\mac$, is used to detect if the adversary tampered with
the message. The additional requirement is that
$\{\mac(\cdot,\ell)\}_\ell$ must be a family of
\emph{$\eps_\mac$\-/strongly universal} hash functions~\cite{WC81,Sti94},
i.e., for any $x_1,x_2 \in \{0,1\}^{n_\mac}$ and $t_1,t_2 \in \{0,1\}^{m_\mac}$ with
$x_1 \neq x_2$,
\[ \Pr_\ell \left[ \mac(x_1,\ell) = t_1 \text{ and } \mac(x_2,\ell) =
  t_2 \right] \leq \frac{\eps_\mac}{2^{m_\mac}}.\]
Constructions for such hash functions are given in \cite{WC81,Sti94},
e.g., $\mac(x,y \| b) = \phi(x*y) + b$ where $\phi$ is any linear
surjective function and $*$ and $+$ are multiplication and addition in
the corresponding finite fields. In this example,
$\Ext_\mac(x,y) = \phi(x*y)$ is a universal hash function, so $\mac$
has the required extractor properties. We denote by $\nu_\mac$ the
corresponding extractor parameter. The probability that tampering is
not detected when this function is used as a standard message
authentication code (MAC) \--- i.e., when a tag $t = \mac(x,\ell)$ is
appended to a message $x$, and upon receiving $(x',t')$, the receiver
checks if $t' = \mac(x',\ell)$~\cite{WC81,Sti94} \--- is
$\eps_\mac$. In the case of the example given, $\nu_\mac = 1$ and
$\eps_\mac = \frac{1}{2^{m_\mac}}$.

In the following we denote the length of the message to be
authenticated by $m$ and the length of the first input to $\sss$ by
$n = n_\sss$. The function $\mac$ is always used with
$n_\mac = n+m+m_\sss$.

The final ingredient needed is a code $\cC \subset \{0,1\}^n$ with
minimum distance $d$ between any two code words.

\vspace{\baselineskip}

The modified Fehr-Salvail protocol works as follows (see
\remref{rem:FS} here below for a description of the differences with
the original protocol).

\paragraph{Encryption.} The players, Alice and Bob, share keys
$\ell_\sss \in \{0,1\}^{r_\sss}$, $\ell_\mac \in \{0,1\}^{r_\mac}$,
and $\theta \in \cC$.  $\ell_\sss$ and $\ell_\mac$ have to be
uniform. $\theta$ has bounded min\-/entropy, i.e., it should be
generated by a min\-/entropy resource with parameter $k$. To
authenticate a message $y \in \{0,1\}^m$, Alice picks a string
$x \in \{0,1\}^n$ uniformly at random, and generates the $n$ qubit
quantum state $H^\theta \ket{x}$, where
$H^\theta = \bigotimes_i H^{\theta_i}$ and $H$ is the Hadamard
matrix. She then computes the values $s = \sss(x,\ell_\sss)$ and
$t = \mac(x \| y \| s, \ell_\mac)$. Finally, she sends the cipher consisting
of $y\|s\|t$ and $\rho = H^\theta \proj{x} H^\theta$ to Bob.

Note that it follows from \lemref{lem:key-privacy.sequarallel} that
the function
$\chi_y : \{0,1\}^n \times \{0,1\}^{r_\sss+r_\mac} \to
\{0,1\}^{m_\sss+m_\mac}$
with $\chi_y(x,\ell_g \| \ell_h) \coloneqq s \| t$ provides
$(\nu_g+\nu_h)$\-/key\-/privacy and uniformity.

\paragraph{Decryption.} Upon receiving $y'\|s'\|t'$ and $\rho'$ from
Alice, Bob first computes $H^\theta \rho' H^\theta$ and measures in
the computational basis, obtaining $\tilde{x}$. He then uses the
recovery procedure to get $x'$ from $(\tilde{x}, s', \ell_\sss)$. He
checks if $w(x',\tilde{x}) \leq \varphi n$ and if
$t' = h(x' \| y' \| s', \ell_\mac)$. If one of these conditions does not
hold, he rejects the message. Otherwise, he accepts $y'$ as the
message that Alice sent.

\paragraph{Key recycling.} If the message was rejected, Bob recycles
$\ell_\sss$ and $\ell_\mac$. If the message was accepted, he recycles
$\ell_\sss$ and $\ell_\mac$, as well as $\theta$. He sends one bit of
information on a backward authentic channel to Alice to tell her if he
accepted or rejected, and she recycles the same keys.

\vspace{\baselineskip}

\begin{rem}
\label{rem:FS}
The protocol described here differs in two ways from the original
Fehr-Salvail protocol~\cite{FS17a}. Firstly, the extractors have to be
linear in their first input, which was introduced in the extended
version~\cite{FS17b} following an initial draft of the current work
pointing out the issue with impersonation attacks. Secondly, we do not
require the players to have fresh key to replace the discarded
$\theta$ in case of a reject, our analysis goes through without
this. Hence, in this version of the protocol, the players can run the
authentication scheme even if they only have the keys
$(\ell_\sss,\ell_\mac,\theta)$, but no access to a key refreshing
function as in \cite{FS17a,FS17b} that generates a new $\theta'$.
\end{rem}

\subsection{Analysis}
\label{sec:FS.analysis}

Our task now is to prove that \eqnsref{eq:FS.adv} and
\eqref{eq:FS.noise} hold and to bound the corresponding errors. We
get similar errors to the original, non\-/composable
proof~\cite{FS17a,FS17b}, but with minor improvements in the
constants. This is stated in the following theorem.

\begin{thm}
\label{thm:FS}
Let $\aK$, $\underline{\aK}$, $\aH^k_{\min} $,
$\underline{\aH}^k_{\min}$, $\underline{\aA}^1$, $\aQ$,
$\aR^\eps_\varphi$, $\aA$, and $\underline{\aA}$, be resource
specification as described above and let $\pi = (\pi_A,\pi_B)$ be the
converters running the modified Fehr-Salvail protocol from
\secref{sec:FS.protocol}. Then
\begin{align*}
\underline{\aK} \| \underline{\aH}^k_{\min} \| \underline{\aA}^1 \| \aR^\eps_\varphi
& \xrightarrow{\pi,\eps_{\noise}} \underline{\aK} \| \underline{\aH}^k_{\min} \| \underline{\aA} \\
\intertext{and}
\aK \| \aH^k_{\min} \| \underline{\aA}^1 \| \aQ &  \xrightarrow{\pi,\eps_{\adv}} \aK \|
\aH^k_{\min} \| \aA,
\end{align*}
with $\eps_{\noise} = \eps + \eps_\sss$
and \begin{equation} \label{eq:thm.error}
\eps_{\adv} = \eps_\mac + 
(\nu_\sss+\nu_\mac) \sqrt{\left(2+\frac{|\cC|}{2^{d/2}} +
    \frac{|\cC|2^{h(\varphi)n}}{2^{d}}\right) 2^{m_\sss+m_\mac-k}}.
\end{equation}
\end{thm}

\begin{proof}
The formal statements we will prove are that for any $\sK \in \underline{\aK}$,
$\sH \in \underline{\aH}^k_{\min}$, $\sA^1 \in \underline{\aA}^1$, and $\sR \in
\aR^\eps_\varphi$, there exists a $\sK' \in \aK$, $\sH' \in
\underline{\aH}^k_{\min}$, $\sA \in \underline{\aA}$, and a simulator
$\simE$ such that
\begin{equation} \label{eq:FS.noise.thm}
\pi_B \pi_A \left(\sK \| \sH \| \sA^1 \| \sR \right) \close{\eps_{\noise}}
\left(\sK' \| \sH' \| \sA \right) \simE;
\end{equation}
and that for any $\sK \in \aK$,
$\sH \in \aH^k_{\min}$, $\sA^1 \in \underline{\aA}^1$, and $\sQ \in
\aQ$, there exists a $\sK' \in \aK$, $\sH' \in
\aH^k_{\min}$, $\sA \in \aA$, and a simulator
$\simE$ such that
\begin{equation}  \label{eq:FS.adv.thm}
\pi_B \pi_A \left(\sK \| \sH \| \sA^1 \| \sQ \right) \close{\eps_{\adv}}
\left(\sK' \| \sH' \| \sA \right) \simE.
\end{equation}

We start with the case of \eqnref{eq:FS.noise.thm}. Since for any
resource $\aR$, $\aR^\eps \xrightarrow{\id,\eps} \aR$, it follows from
\thmref{thm:security} that it is sufficient to consider a
noisy channel specification $\aR_\varphi$ and add an error $\eps$ to
the final statement.

The distinguisher could either first provide Alice with a message $y$
to be encrypted, or first interact with the key resources so that they
generate keys $\ell,\theta$. The second case is a more powerful
distinguisher, since it can choose a message correlated to $\theta$,
due to the side information it has about $\theta$. So w.l.o.g.\ we
consider only this second case, i.e., the distinguisher first
interacts with $\underline{\aH^k_{\min}}$ so that $\theta$ is
generated, and notifies $\underline{\aK}$ to generate the key
$\ell = (\ell_\sss,\ell_\mac)$. It then provides a message $y$ to
Alice, who prepares the cipher, and sends it on the noisy channel
$\aR_\varphi$. Due to the error correction properties of the function
$\sss$, Bob can reconstruct the correct $x$ except with probability
$\eps_\sss$. If he reconstructs the correct $x$, then Alice's message is
always accepted and $\theta$ is recycled.

For every $\sH \in \underline{\aH^k_{\min}}$ we construct a $\sH'$
that behaves identically to $\sH$, except that once $\theta$ has been
generated, $\sH'$ does not accept to output it at the players' interface
if requested, but waits to get a notification at Eve's interface that
it can be output. The simulator $\simE$ allows the distinguisher to
interact directly with $\sH'$, except for the last message. It also
blocks the notification to $\sK'$ to generate the key. Once $\theta$
has been generated (but not output), the activation of $\sK'$ received
from the distinguisher (but not delivered), and the simulator has
received the notification that the message has been received and
delivered by $\sA$, it generates the one bit message for the
transcript of the backwards authentic channel and notifies $\sH'$ and
$\sK'$ that their keys can now be output if requested.

These real and ideal systems behave identically, except if Bob fails
to correctly reconstruct $x$ in the real system. So the final error is
$\eps_{\noise} = \eps + \eps_\sss$.

In the case of \eqnref{eq:FS.adv.thm}, we may also assume w.l.o.g.\
that the distinguisher first interacts with $\aH^k_{\min}$ and
notifies $\aK$ to generate the keys. But it then has two options. It
first provides Alice with a message $y$, intercepts and possibly
changes the cipher, and finally obtains Bob's outcome as well as the
keys \--- a substitution attack. Or, the distinguisher may first input
a cipher on the insecure channel $\aQ$, obtain Bob's output as well as
the recycled keys, then choose a message $y$ that it inputs at Alice's
interface, and finally gets her cipher \--- an impersonation attack.

We start analyzing the substitution attack. This case follows closely
the security proof from \cite{FS17a}. Let $\rho_{\Theta E}$ be the
subnormalized state of the key $\theta$ and Eve's side information
after interacting with $\sH$ when a key $\theta$ is successfully
generated. Eve will now measure the $E$ system to choose her message
$y$, resulting in a new state
$\sigma_{Y\Theta E} = \sum_y p_y \proj{y,\theta} \tensor
\sigma^{y,\theta}_{E}$,
where we have normalized the states so that
$\tr \sigma^y_{\Theta E} = \tr \rho_{\Theta E}$. By definition we have
$\Hmin[\sigma]{\Theta | EY} \geq k$. Note that one also has
\begin{multline}
\label{eq:thm.hmin}
2^{-\Hmin[\sigma]{\Theta|EY}} = \pguess[\sigma]{\Theta|EY} = \\ \sum_y p_y
\pguess[\sigma^y]{\Theta|E} = \sum_y p_y
2^{-\Hmin[\sigma^y]{\Theta|E}}.
\end{multline}
From now one, we take $y$ to be fixed and $\sigma^y_{\Theta E}$ to be
the shared state. At the end of the proof we average over the
different $y$ weighted by $p_y$.

Alice runs her authentication protocol, after which the shared state
between Alice and the distinguisher is given by
$\sigma^y_{L\Theta XZQE}$, where $Z$ contains the classical part of
the cipher, $Q$ contains the quantum part, $X$ is the uniform string
that Alice generated and encoded in $Q$, and $L$ contains the uniform
keys $\ell_\sss \| \ell_\mac$. The $ZQ$ systems are intercepted by the
distinguisher, who applies a map $\cE : \lo{ZQE} \to \lo{ZZ'QE}$ which
leaves $Z$ unmodified and generates a new $Z'$. Let
$\mu^y_{L\Theta XZZ'QE}$ be the resulting state. $Z'Q$ is now sent to
Bob, who measures $Q$ to get $\tilde{X}$, decodes it with
$(S',\ell_\sss)$ to get $X'$, and checks whether
$w(X',\tilde{X}) \leq \varphi n$ and $T' = h(Y'\|S'\|X',\ell_\mac)$,
where $Z' = Y'\|S'\|T'$. He finally sends a bit $D$ to Alice
containing his decision. If $D = 1$ he outputs the received message
$Y'$ when requested, and outputs the keys $L$ and $\Theta$ when
requested. If $D = 0$, he outputs an error $\bot$ and only recycles
$L$. Let the final state be $\rho^y_{DL\Theta XX'\tilde{X}ZZ'E}$.

In the ideal case, we define $\sH'$ and $\simE$ to work as
follows. $\sH'$ first runs $\sH$, so that exactly the same key
$\Theta$ and side information $E$ are generated while interacting with
the distinguisher \--- the simulator $\simE$ lets all these messages
between $\sH'$ and the distinguisher go through. Once the
distinguisher has input the message $y$ in the ideal authentic
channel, the simulator gets a copy, which it forwards to
$\sH'$. $\sH'$ then picks its own keys $(\ell'_\sss, \ell'_\mac)$
uniformly at random and a uniform $X$, and generates a cipher $ZQ$. It
then outputs $ZQ$ at Eve's interface, which the simulator $\simE$
passes on to the distinguisher. Let the state shared between the
different systems at this point be $\tilde{\sigma}^y_{L\Theta XZQE}$,
where $L$ represents the new keys that are to be output by $\aK$ (not
those generated internally by $\sH'$). The distinguisher performs the
same map, $\cE$, as when interacting with the real system, resulting
in a state $\tilde{\mu}^y_{L\Theta XZZ'QE}$, and gives $Z'Q$ to the
simulator $\simE$ who forwards them to $\sH'$. $\sH'$ now measures $Q$
just as Bob would do, gets $\tilde{X}$. It then reconstructs $X'$,
checks that $X' = X$, that $w(X',\tilde{X}) \leq \varphi n$ and that
$Z = Z'$. If one of these three conditions is not satisfied, it
outputs a bit $D = 0$, otherwise $D =1$, which the simulator
intercepts. If $D = 0$, $\simE$ tells the authentic channel to output
an error symbol at Bob's interface when requested, if $D = 1$ the
message $y$ is output instead. Regardless of the value of $D$, $\simE$
outputs $D$ as the bit on the backward authentic channel if requested,
and also notifies $\aK$ to output a random key
$\ell = \ell_\sss \| \ell_\mac$ when requested. And if $D = 0$, it
tells $\sH'$ to output $\bot$ when requested, otherwise $\theta$. Let
the final state be $\tilde{\rho}^y_{DL\Theta XX'\tilde{X}ZZ'E}$.

We now need to prove two things. Firstly, that
$\sH' \in \aH^{k}_{\min}$, and secondly that the real and ideal
systems are indistinguishable except with advantage
$\eps_{\adv}$, i.e., \begin{equation} \label{eq:FS.thm.1}
  \frac{1}{2} \trnorm{\rho^{y,1}_{L\Theta ZZ'E} -
    \tilde{\rho}^{y,1}_{L\Theta ZZ'E}} \leq \eps^{y,1}_{\adv}
  \quad \text{and} \quad \frac{1}{2} \trnorm{\rho^{y,0}_{LZZ'E} -
    \tilde{\rho}^{y,0}_{LZZ'E}} \leq \eps^{y,0}_{\adv}
  ,\end{equation} where
$\rho^y_{L\Theta XX'\tilde{X}ZZ'ED} = \rho^{y,0}_{L\Theta
  XX'\tilde{X}ZZ'E} \tensor \proj{0} + \rho^{y,1}_{L\Theta
  XX'\tilde{X}ZZ'E} \tensor \proj{1}$
and
$\sum_y p_y (\eps^{y,0}_{\adv}+\eps^{y,1}_{\adv}) =
\eps_{\adv}$.

We start with $\sH' \in \aH^{k}_{\min}$. For this, it is sufficient to
show that
\[\Hmin[\tilde{\rho}^{y,1}]{\Theta|ZZ'E} \geq
\Hmin[\sigma^y]{\Theta|E},\]
where $\sigma^{y}_{\Theta E}$ is the state output by $\sH$ (after the
measurement to choose $y$), since it them follows from
\eqnref{eq:thm.hmin} that $\Theta$ has enough min\-/entropy. From
\lemref{lem:tech.min-entropy-inequalities.event} we have
\[\Hmin[\tilde{\rho}^{y,1}]{\Theta|ZZ'E} \geq
\Hmin[\tilde{\rho}^{y}]{\Theta|ZZ'E} =
\Hmin[\tilde{\mu}^{y}]{\Theta|ZZ'E}.\]
Applying a map to the side information can only increase the entropy,
hence
\[\Hmin[\tilde{\mu}^{y}]{\Theta|ZZ'E} \geq
\Hmin[\tilde{\sigma}^{y}]{\Theta|ZQE}.\]
Because $Z = y \|S\|T$ and the function which computes $S\|T$ is
uniform (see \defref{def:uniformity}), $Z$ is independent from the other
systems, hence
\[ \Hmin[\tilde{\sigma}^{y}]{\Theta|ZQE} =
\Hmin[\tilde{\sigma}^{y}]{\Theta|QE}.\]
Finally, $Q$ was generated by applying a unitary $H^\theta$ to a fully
mixed state $\frac{1}{2^n}\sum_x \proj{x}$, so $Q$ is also fully mixed
and independent of the other systems (which contain no information
about $x$), hence
\[ \Hmin[\tilde{\sigma}^{y}]{\Theta|QE} =
\Hmin[\tilde{\sigma}^{y}]{\Theta|E}.\]

We now return to \eqnref{eq:FS.thm.1}. From the properties of $\mac$, it
follows that if $D = 1$, then $Z' = Z$ and $X'=X$ except with
probability $\eps_\mac$, i.e., the state
$\bar{\rho}^{y}_{L\Theta XX'\tilde{X}Y'ZZ'ED}$ obtained by flipping
the value of $D$ from $1$ to $0$ if either $Z' \neq Z$ or $X' \neq X$
must be $\eps_\mac$-close to $\rho^{y}_{L\Theta XX'\tilde{X}ZZ'ED}$. It
now suffices to bound the distance between $\tilde{\rho}^y$ and
$\bar{\rho}^{y}$. To simplify notation, we introduce a new register
$\Theta'$ such that if $D = 0$, $\Theta' = \bot$, and if $D=1$, then
$\Theta' = \Theta$. Thus, we are now trying to prove that
\begin{equation} \label{eq:FS.thm.2}
\frac{1}{2} \trnorm{\bar{\rho}^{y}_{L\Theta' ZZ'ED} -
  \tilde{\rho}^{y}_{L\Theta' ZZ'ED}} \leq \eps^{y}_{\adv} -\eps_\mac.
\end{equation}

In the ideal case we have
$\tilde{\rho}^{y}_{L\Theta' ZZ'ED} = \tau_L \tensor
\bar{\rho}^{y}_{\Theta' ZZ'ED}$,
where $\tau_L$ is the fully mixed state, because by construction $L$
is uniform and independent of the rest, and
$\tilde{\rho}^{y}_{\Theta' ZZ'ED} = \bar{\rho}^{y}_{\Theta' ZZ'ED}$
\--- $\sH'$ runs the real protocol, the only difference is the reject
condition, but by the flip of $D$ that generated $\bar{\rho}^y$ from
$\rho^y$ makes them now abort under the same conditions. Plugging this
in \eqnref{eq:FS.thm.2}, it remains to show that
\begin{equation*} %\label{eq:FS.thm.3}
\frac{1}{2} \trnorm{\bar{\rho}^{y}_{L\Theta' ZZ'ED} -
  \tau_L \tensor
\bar{\rho}^{y}_{\Theta' ZZ'ED}} \leq \eps^{y}_{\adv} -\eps_\mac.
\end{equation*}

From \lemref{lem:key-privacy.sequarallel} it follows that the function
\[\chi_y : \{0,1\}^n \times \{0,1\}^{r_\sss+r_\mac} \to \{0,1\}^{m_\sss+m_\mac}\]
with $\chi_y(x,\ell_\sss \| \ell_\mac) \coloneqq s \| t$,
$s = \sss(x,\ell_\sss)$, and $t = \mac(x\|y\|s,\ell_\mac)$ provides
$(\nu_\sss+\nu_\mac)$\-/key\-/privacy% and uniformity
. Furthermore, the
state $\bar{\rho}^y_{LX\Theta' ZZ'ED}$ satisfies the assumptions of
\defref{def:key-privacy} (with $T = Z$). Hence we have
\begin{multline*} \frac{1}{2}\trnorm{\bar{\rho}^{y}_{L\Theta' ZZ'ED} -
    \tau_L \tensor \bar{\rho}^{y}_{\Theta' ZZ'ED}} \\ \leq
  \frac{\nu_\sss+\nu_\mac}{2}\sqrt{2^{-\Hmin[\bar{\rho}^y]{X|\Theta'
        ZZ'ED}+m_\sss+m_\mac}}.\end{multline*} It now remains to upper
bound
\[ 2^{-\Hmin[\bar{\rho}^y]{X|\Theta' ZZ'ED}} =
\pguess[\bar{\rho}^y]{X|\Theta' ZZ'ED} =
\pguess[\bar{\rho}^y]{X|\Theta' ZZ'E},\] where we have removed $D$
because this register is redundant, it can be inferred from $\Theta'$
(since $\Theta' = \bot \iff D = 0$). 

Let us define a new register $\Omega$ which takes the value
$1$ if $w(X,\tilde{X}) \leq \varphi n$ and $\Omega = 0$ otherwise. Let
$\bar{\rho}^y_{X\Theta' ZZ'E\Omega} =  \bar{\rho}^{y,0}_{X\Theta'
  ZZ'E} \tensor \proj{0} + \bar{\rho}^{y,1}_{X\Theta' ZZ'E} \tensor
\proj{1}$. We have
\begin{align*}
\pguess[\bar{\rho}^y]{X|\Theta' ZZ'E} & \leq \pguess[\bar{\rho}^y]{X|\Theta' ZZ'E\Omega} \\
& = \pguess[\bar{\rho}^{y,0}]{X|\Theta' ZZ'E} + \pguess[\bar{\rho}^{y,1}]{X|\Theta' ZZ'E} \\
& \leq \pguess[\bar{\rho}^{y,0}]{X|ZZ'E} + \pguess[\bar{\rho}^{y,1}]{X|\Theta ZZ'E}.
\end{align*}
The third line follows because adding $\Theta$ to the side information
when $\Theta' = \bot$ can only increase the probability of guessing
$X$.

We now bound these two terms separately.
\begin{align*}
\pguess[\bar{\rho}^{y,0}]{X|ZZ'E} & \leq \pguess[\bar{\rho}^{y}]{X|ZZ'E} \\
& = \pguess[\rho^{y}]{X|ZZ'E} \\
& = \pguess[\mu^{y}]{X|ZZ'E} \\
& \leq \pguess[\sigma^{y}]{X|ZQE} \\
& = \pguess[\sigma^{y}]{X|QE}.
\end{align*}
The first line follows from
\lemref{lem:tech.min-entropy-inequalities.event}. The fourth line
follows because applying a map to the side information can only
decrease the probability of guessing $X$. The fifth line follows
because $\chi_y$ is uniform so $Z$ is independent. Finally, by noting
that
$\sigma^y_{\Theta XQE} = \cM^{\BB}_{\Theta A,X} \left(
  \sigma^y_{\Theta XE} \tensor \Phi^+_{AQ} \right)$,
where $\Phi^+_{AQ}$ are EPR pairs and $\cM^{\BB}_{\Theta A,X}$
measures $A$ according to the basis in $\Theta$ and writes the result
in $X$ (see \appendixref{app:guessing} for a formal definition of
$\cM^{\BB}_{\Theta A,X}$), we can apply \lemref{lem:guessing.one} with
$B=QE$, from which we get
\[ \pguess[\sigma^{y}]{X|QE} \leq
\pguess[\sigma^{y}]{\Theta|E}\left(1+\frac{|\cC|}{2^{d/2}}\right).\]

To bound $\pguess[\bar{\rho}^{y,1}]{X|\Theta ZZ'E}$ we will use
\lemref{lem:guessing.two}. Let $P^\Omega$ be an operator which
projects the state on the event $\Omega = 1$, namely
$w(X,\tilde{X}) \leq \varphi n$. Let $\cM^{\BB}_{\Theta Q,\tilde{X}}$
be Bob's measurement of the cipher $Q$ that yields the string
$\tilde{X}$. Let $\cE : \lo{ZQE} \to \lo{ZZ'QE}$ be the map performed
by the distinguisher on the cipher. We thus have
$\bar{\rho}^{y,1}_{X\Theta ZZ'E} = {\tr_{\tilde{X}}} \circ P^{\Omega}
\circ \cM^{\BB}_{\Theta Q,\tilde{X}} \circ \cE (\sigma^Y_{X\Theta
  ZQE})$.
Note furthermore that by the uniformity of $\chi_y$, $Z = y \| S \| T$
where $S\|T$ is uniform
and independent from the other registers, and $XQ$ may be generated by
measuring halves of EPR pairs $\Phi^+_{AQ}$ according to
$\Theta$ (which commutes with $\cE$). Hence we have
\[ \bar{\rho}^{y,1}_{X\Theta ZZ'E} = \tr_{\tilde{X}} \circ P^{\Omega}
\circ \cM^{\BB}_{\Theta Q,\tilde{X}} \circ \cM^{\BB}_{\Theta A,X}
\circ \cE (\sigma^y_{\Theta E} \tensor \Phi^+_{AQ} \tensor \tau_Z).\]
This puts us in a position to apply \lemref{lem:guessing.two} with
$B = Q$ and $C=ZZ'E$, from which we get
\[ \pguess[\bar{\rho}^{y,1}]{X|\Theta ZZ'E} \leq
\pguess[\sigma^{y}]{\Theta|E}\left(1+\frac{|\cC|2^{h(\varphi)n}}{2^{d}}\right).\]

Finally, taking the average over $p_y$ along with Jensen's inequality,
we get the bound from \eqnref{eq:thm.error}.

The final case to consider is that of impersonation attacks. In the
real system, the distinguisher sends a forged cipher to Bob, who
performs the decoding. He then either accepts and recycles both
$\ell = (\ell_\sss,\ell_\mac)$ and $\theta$, or rejects and recycles
only $\ell = (\ell_\sss,\ell_\mac)$. Let $\sigma_{L\Theta ED}$ denote
the joint state at this point, where, as previously, $D$ is
Bob's decision to accept or reject the message, $L$ contains the
uniform keys $ (\ell_\sss,\ell_\mac)$, $\Theta$ is the non\-/uniform
key (which may have been given to the distinguisher if $D = 1$).

The distinguisher can then choose a message $Y$ to input at Alice's
interface for encryption. Let
$\cE : \lo{L\Theta ED} \to \lo{L\Theta YED}$ be the operator applied
by the distinguisher to generate $Y$, which we define so as to leave
the classical registers $L\Theta D$ unmodified, and furthermore, $\cE$
is only allowed to use information from $\Theta$ if $D=1$. We denote
the resulting state by
$\rho_{L\Theta YED} = \cE(\sigma_{L\Theta ED})$. Finally, Alice's
protocol generates a cipher, resulting in the joint shared state
$\rho_{L\Theta UYQED}$, where $UY$ denotes the classical part of the
cipher with $U = S \| T$ (jointly written $Z = Y \| S \| T$ in the
proof against substitution attacks), and $Q$ is the quantum part of
the cipher.

In the ideal case, if the simulator $\simE$ receives a forged cipher
at its outer interface from the distinguisher before it receives the
message $y$ from the authentic channel, it knows that we are dealing
with an impersonation attack. It then tells the authentic channel to
output an error when requested, the key resource $\sK'$ to output some
fresh uniform key $\ell = (\ell_\sss,\ell_\mac)$ when requested, and
the min\-/entropy resource $\sH'$ to output an error $\bot$ when
requested. Let $\tilde{\sigma}_{L\Theta ED}$ denote the joint state at
this point. As in the real case, the distinguisher applies the map
$\cE$ to generate a message $Y$, inputs this to the authentic channel,
which gives it to the simulator. $\simE$ now picks its own key $\ell'$
and asks $\sH'$ to give it the key $\theta$ that it generated while
running $\sH$ internally. $\simE$ then follows the protocol and
generates a cipher $YUQ$, which it outputs at its outer interface. Let
the final state shared by the different parties be
$\tilde{\rho}_{L\Theta UYQED}$.

As in the case of impersonation attacks, we need to prove two
things. Firstly, that $\sH' \in \aH^k_{\min}$, and secondly, that
\begin{equation}
\label{eq:FS.thm.3}
\frac{1}{2}\trnorm{\rho_{LUYQED} - \tilde{\rho}_{LUYQED}}
\leq \eps_{\adv}.
\end{equation}
Note that \eqnref{eq:FS.thm.3} does not contain $\Theta$. This is
because it is only ever recycled in the real case (the ideal system
cannot get fooled by an impersonation attack), but the bit $D = 1$ is
already enough to perfectly distinguisher real from ideal in this
case, one does not need $\Theta$.

Since $\sH'$ always outputs $\bot$, it trivially satisfies the
definition of a $(k,0)$\-/min\-/entropy resource.  For bounding
\eqnref{eq:FS.thm.3}, note that by the definition of $\mac$, the
probability of accepting the forged cipher is at most $\eps_\mac$. Let
$\bar{\sigma}_{L\Theta ED}$ be the state obtained by flipping $D$ from
$1$ to $0$ on $\sigma_{L\Theta ED}$ and define
$\bar{\rho}_{L\Theta UYQED}$ as the state obtained after the
distinguisher applies $\cE$ to $\bar{\sigma}_{L\Theta ED}$ and Alice
generates the cipher. Then
\[
 \frac{1}{2} \trnorm{\rho_{LUYQED} - \bar{\rho}_{LUYQED}} \leq \eps_\mac.
\]
Furthemore, in the ideal system one has $\tilde{\rho}_{LUYQED} =
\tau_{L} \tensor \tau_U \tensor \bar{\rho}_{YQED}$, where $\tau_{L}$
and $\tau_U$ are fully mixed states. This follows, because $L$ is uniform by definition,
$U$ is uniform because it is generated by the uniform function
$\chi_y$, and $\tilde{\rho}_{YQED} = \bar{\rho}_{YQED}$ because the
simulator uses the same $\theta$ and performs the same operation 
to generate $Q$. Since we always have $D=0$, this register can be
removed, and it remains to show that,
\[
 \frac{1}{2} \trnorm{\bar{\rho}_{LUYQE} - \tau_L \tensor \tau_U \tensor
   \bar{\rho}_{YQE}} \leq \eps_{\adv} - \eps_\mac.
\]

We now look more closely at how the protocol generates $U$. Since the
extractors used are linear, we have $u = s \| t$ with
$s = A_{\ell_\sss} x + b_{\ell_\sss}$ and
$t = A_{\ell_\mac} (x \| y \| s)+ b_{\ell_\mac}$, where
$A_{\ell_\sss}$ and $A_{\ell_\mac}$ are matrices, and $b_{\ell_\sss}$
and $b_{\ell_\mac}$ are strings, which depend on $\ell_\sss$ and
$\ell_\mac$. One can alternatively write this as
$u = \left( A_{\ell_\sss} x \| A^1_{\ell_\mac} x \right) + \left( 0 \|
  A^3_{\ell_\mac} A_{\ell_\sss} x\right)+ B_{\ell_\sss,\ell_\mac} y + c_{\ell_\sss,\ell_\mac}$,
where $A^1_{\ell_\mac}$ and $A^3_{\ell_\mac}$ are the first $n$
columns and last $m_\sss$ columns of $A_{\ell_\mac} =
A^1_{\ell_\mac} \| A^2_{\ell_\mac} \| A^3_{\ell_\mac}$, respectively,
and $B_{\ell_\sss,\ell_\mac}$ and $c_{\ell_\sss,\ell_\mac}$ are a
matrix and string which depend on $\ell_\sss$ and $\ell_\mac$.

Since $\Ext_\mac(x\|y\|s,\ell_\mac) = A_{\ell_\mac}(x\|y\|s)$ is a
$(k,\eps)$\-/strong extractor, then so is
$\Ext'_\mac(x,\ell_\mac) \coloneqq A^1_{\ell_\mac}x$. And from
\lemref{lem:tech.extractors.composition} and the bounds on the errors
of $\Ext_\sss$ and $\Ext_\mac$, we find that
$\Ext(x,\ell_\sss\|\ell_\mac) \coloneqq A_{\ell_\sss} x \|
A^1_{\ell_\mac} x$
is a $(k,\eps)$\-/strong extractor for any $k$ and
$\eps = \frac{\nu_\sss+\nu_\mac}{2}\sqrt{2^{-k+m_\sss+m_\mac}}$.  The
state $\bar{\rho}_{LUYQE}$ may thus be written as
$ \bar{\rho}_{LUYQE} = \cG \circ \cF \left(
  \bar{\rho}_{L\!\Ext(X,L)YQE}\right)$,
where $\cF : \lo{LU} \to \lo{LU}$ with $U=ST$ reads $L$ and $S$ and
XORs $A^3_{\ell_\mac}s$ to the $T$ register, and
$\cG : \lo{LUY} \to \lo{LUY}$ reads $L$ and $Y$ and XORs
$B_{\ell_\sss,\ell_\mac}y$ and $c_{\ell_\sss,\ell_\mac}$ to $U$. Since
$L$ is uniform and independent from $\bar{\rho}_{YQE}$, it follows
from the definition of an extractor that 
\begin{multline*}
\frac{1}{2} \trnorm{\cG \circ \cF \left(
  \bar{\rho}_{L\!\Ext(X,L)YQE}\right)- \cG \circ \cF \left(\tau_L \tensor \tau_U \tensor
   \bar{\rho}_{YQE}\right)} \\ \leq \frac{\nu_\sss+\nu_\mac}{2}\sqrt{2^{-\Hmin[\bar{\rho}]{X|YQE}+m_\sss+m_\mac}}.
\end{multline*}

To finish the proof we need that $\cF$ and $\cG$ XOR a value to a
uniform string, which results in a uniform string, hence
\[\cG \circ \cF \left(\tau_L \tensor \tau_U \tensor
  \bar{\rho}_{YQE}\right) = \tau_L \tensor \tau_U \tensor
\bar{\rho}_{YQE},\]
and we need to upper bound
$2^{-\Hmin[\bar{\rho}]{X|YQE}} = \pguess[\bar{\rho}]{X|YQE}$. This
latter bound is obtained from \lemref{lem:guessing.one} following the
same steps as the bound on $\pguess[\sigma^{y}]{X|QE}$ in the case of
substitution attacks, from which we get
\[ \pguess[\bar{\rho}]{X|YQE} \leq
\pguess[\sigma]{\Theta|E}\left(1+\frac{|\cC|}{2^{d/2}}\right) \leq 2^k
\left(1+\frac{|\cC|}{2^{d/2}}\right). \qedhere\]
\end{proof}

\subsection{Continuing after a reject}
\label{sec:FS.continuing}

Since the construction from \thmref{thm:FS} generates the same
resources $\aK$ and $\aH^k_{\min}$ that it uses, it is trivial to
recursively apply the protocol to authenticate multiple messages using
the same keys. By \thmref{thm:security}, the error of $n$ runs of the
protocol, $\pi \circ \dotsb \circ \pi$, is $n \eps_{\adv}$. But
note that the protocol does not perform anything if it gets an error
from the resource $\aH^k_{\min}$ instead of a key. This way of
composing the protocol with itself aborts all future rounds as soon as
tampering is detected and the key $\theta$ is not
recycled.

One can imagine a different scenario, in which the players have spare
secret key, which they use to replace $\theta$, if it cannot be
recycled. Let $\underline{\aK}$ denote such an extra key resource, and
let $\aH^k_{\min}$ denote a min\-/entropy resource, which might
produce an error $\bot$. Given these two resources, we wish to
construct a new resource $\underline{\aH}^{k'}_{\min}$ which always
outputs a key, by ``giving'' the key from $\underline{\aK}$ to
$\aH^k_{\min}$. We are then left with a resource $\aK$ which might
still output a key if it was not given to $\aH^k_{\min}$. $\aK$ can be
defined as in \figref{fig:qkd.ideal}, with a switch that decides if it
produces a key or not.

\begin{lem}
\label{lem:continuing}
Let $\pi^{\text{new}} = (\pi^{\text{new}}_A,\pi^{\text{new}}_B)$ be a
protocol where both $\pi^{\text{new}}_A$ and $\pi^{\text{new}}_B$ get
keys $(k,\theta)$ from $\underline{\aK}$ and $\aH^k_{\min}$. If
$\theta \neq \bot$, they output $(k,\theta)$ at their outer interfaces
when requested. If $\theta = \bot$, they output $(\bot,k)$ at their
outer interfaces instead. Then
\[  \underline{\aK} \| \aH^k_{\min} \xrightarrow{\pi^{\text{new}},0}  \aK \|
\underline{\aH}^{k'}_{\min}\]
with $k' = - \log \left(2^{-n}+ 2^{-k}\right)$.
\end{lem}

\begin{proof}
Let $\sH \in \aH^k_{\min}$, and after interacting with this, let the
shared state of the key and the 
distinguisher be $\rho_{\Theta E} = \proj{\bot} \tensor \mu_E +
\sigma_{\Theta E}$ with $\Hmin[\sigma]{\Theta|E} \geq k$.

We define $\sH'$ to run $\sH$, and if it gets $\bot$, it generates a
fresh uniform key $\tau_{\Theta}$, which it outputs when
requested. The shared state is then
$\rho'_{\Theta E} = \tau_{\Theta} \tensor \mu_E + \sigma_{XE}$. One
has
\[
\pguess[\rho']{\Theta|E} \leq \pguess[\tau \tensor \mu]{\Theta|E} +
                           \pguess[\sigma]{\Theta|E} \leq 2^{-n} + 2^{-k}.\]
Hence
\[ \Hmin[\rho']{\Theta|E} = -\log \pguess[\rho']{\Theta|E} \geq - \log
\left(2^{-n}+ 2^{-k}\right). \]
This shows that $\sH' \in \underline{\aH}^k_{\min}$. We additionally
need $\sH'$ to output at Eve's interface whether $\sH$ generated
$\bot$ or not, and for a simulator to activate the corresponding
switch on $\aK$.
\end{proof}

If one has enough spare key $\underline{\aK}$, recursively composing
$\pi^{\text{new}}$ with the modified Fehr-Salvail protocol $\pi$ \---
i.e., running
$\pi \circ \pi^{\text{new}} \circ \pi \circ \dotsb \circ
\pi^{\text{new}} \circ \pi$
\--- allows one to encrypt multiple messages by using keys from
$\underline{\aK}$ to replace lost $\theta$. Note that since the
entropy of the key decreases slightly when this is done, the error
$\eps_{\adv}$ will increase slightly with each run (if one does not
change the parameters to compensate).

\appendix
\appendixpage
\phantomsection
\label{app}
%\pdfbookmark[1]{Appendices}{app}
\addcontentsline{toc}{section}{Appendices}

\section{Notation and basic concepts}
\label{app:basic}

We assume the reader is familiar with basic quantum information
theory, e.g., textbooks such as \cite{NC00,Wat16}. In this appendix we
explain the notation, and introduce the distance and entropy measures
that we use in this work.

\subsection{Quantum states, maps, and norms}
\label{app:basic.states}

We use standard notation for quantum information theory. $\hilbert$
denotes a Hilbert space. In all sections except
\secref{sec:min-entropy.box} and \appendixref{app:sys.boxes}, where
causal boxes are used, Hilbert spaces are finite dimensional. In
finite dimensions, $\lo{}$ denotes the set of linear operators on
$\hilbert$, $\no{}$ is the set of normalized positive operators \---
density operators \--- and $\sno{}$ is the set of subnormalized
positive operators, i.e., if $\rho \in \sno{}$, then
$0 \leq \tr \rho \leq 1$. In infinite dimensions, $\tcop{\hilbert}$
denotes the set of trace class
operators.\footnote{$V \in \tcop{\hilbert}$ if
  $\trnorm{V} = \sum_i \bra{i} \sqrt{\hconj{V}V} \ket{i} < \infty$,
  where $\{\ket{i}\}$ is an orthonormal basis of $\hilbert$. A density
  operator is a non\-/negative self\-/adjoint operator
  $\rho \in \tcop{\hilbert}$ with $\tr \rho = 1$.}

We write $\hilbert_{AB} = \hilbert_A \tensor \hilbert_B$ for a
bipartite quantum system and $\rho_{AB} \in \sno{AB}$ for a bipartite
(subnormalized) quantum state. $\rho_A = \trace[B]{\rho_{AB}}$ and
$\rho_B = \trace[A]{\rho_{AB}}$ denote the corresponding reduced
density operators.  Let $\cE : \lo{A} \to \lo{B}$ be a completely
positive, trace\-/preserving (CPTP) map.\footnote{In the case of
  infinite dimensional systems, a CPTP map acts on trace class
  operators, $\cE : \tcop{\hilbert_A} \to \tcop{\hilbert_B}$.} When it
is applied to a state $\rho \in \sno{AC}$, we write $\cE(\rho)$ as
shorthand for $\left(\cE \tensor \id_C\right) (\rho)$, where $\id_C$
is the identity on system $C$.

The only norm we need in this work is the trace norm (or Schatten
$1$\-/norm), defined as $\trnorm{A} \coloneqq \tr \sqrt{\hconj{A}A}$.

\subsection{Distance measures}
\label{app:basic.distance}

The trace distance between two states $\rho$ and $\sigma$ is given by 
\[D(\rho,\sigma) \coloneqq \frac{1}{2} \trnorm{\rho - \sigma}.\]
This corresponds to the maximum advantage one has in distinguishing
between the two states, i.e., if given $\rho$ or $\sigma$ chosen
uniformly at random one has to guess which one we hold, then the
probability of guessing correctly is \cite{Wat16}
\[p = \frac{1}{2} + \frac{1}{2} D(\rho,\sigma).\]

Another widely used
measure is the fidelity, defined as
\[ F(\rho,\sigma) \coloneqq \trace{\sqrt{\rho^{1/2}\sigma\rho^{1/2}}}.\]

When dealing with subnormalized states, we need to generalize these
measures to retain their properties. The following distance notions
are treated in detail in \cite{TCR10}, and we refer to that work for
more information.

For any two subnormalized states $\rho,\sigma \in \sno{}$, we define
the generalized trace distance as
\[\bar{D}(\rho,\sigma) \coloneqq D(\rho,\sigma) +
\frac{1}{2}|\tr \rho - \tr \sigma|,\] and the generalized fidelity as
\[\bar{F}(\rho,\sigma) \coloneqq F(\rho,\sigma) +
\sqrt{(1-\tr \rho)(1 - \tr \sigma)}.\]

The (generalized) fidelity has a useful property, known as Uhlmann's
theorem (see \cite{NC00,TCR10}), which states that for any two states
$\rho$ and $\sigma$, there exist purifications of these states which
have the same fidelity. The \emph{purified distance} is defined based
on the fidelity, so as to have the same property~\cite{TCR10}:
 \[P(\rho,\sigma) \coloneqq \sqrt{1-\bar{F}^2(\rho,\sigma)}.\]

This metric coincides with the generalized distance for pure states,
and is larger otherwise.
\begin{lem}[\protect{\cite[Lemma 6]{TCR10}}]
\label{lem:purified}
Let $\rho,\sigma \in \sno{}$. Then
\[\bar{D}(\rho,\sigma) \leq P(\rho,\sigma) \leq
\sqrt{2\bar{D}(\rho,\sigma)}.\]\end{lem}
The purified distance is used to define smooth min\-/entropy in the
following section.

\subsection{Min-entropy}
\label{app:basic.min-entropy}

The \emph{smooth conditional min\-/entropy} that we use throughout
this work to define the randomness of a quantum system was first
proposed by Renner~\cite{Ren05}. It represents the optimal measure for
randomness extraction in the sense that it is always possible to
extract that amount of almost uniform randomness from a source, but
never more. Before defining this notion, we first state a
\emph{non-smooth} version.

\begin{deff}[conditional min-entropy~\cite{Ren05}]
  \label{def:min-entropy}
  Let $\rho_{AB} \in \sno{AB}$. The
  \emph{min\-/entropy} of $A$ conditioned on $B$ is defined as
  \begin{equation*}
%    \label{eq:min-entropy}
    \Hmin[\rho]{A|B} \coloneqq \max \{\lambda
    \in \reals :  \exists \sigma_B \in \no{B}\text{\ s.t.\ } 2^{-\lambda} \1_A \tensor \sigma_B \geq \rho_{AB}\},
  \end{equation*}
  where $\1_A$ denotes the identity operator and $A \geq B$ if and
  only if $A-B$ is positive semi\-/definite ($A - B \geq 0$).
\end{deff}

We will often drop the subscript $\rho$ when there is no doubt about what
underlying state is meant.

This definition has a simple operational interpretation when the
first system is classical, which is the case we consider. K\"onig et
al.~\cite{KRS09} showed that for a state $\rho_{XB} = \sum_{x \in
  \cX} p_x \proj{x} \tensor \rho^x_B$ classical on
$X$, \begin{equation} \label{eq:hmin=pguess} \Hmin[\rho]{X|B} = - \log
  \pguess[\rho]{X|B},\end{equation} where $\pguess{X|B}$ is the
maximum probability of guessing $X$ given $B$, namely
\begin{equation*}
%  \label{eq:pguess} 
  \pguess[\rho]{X|B} \coloneqq
  \max_{\{\Gamma_x\}_{x \in \cX}} \left( \sum_{x \in \cX}
    p_x \trace{\Gamma_x\rho^x_B} \right),
\end{equation*}
where the maximum is taken over all positive\-/operator valued
measures (POVMs) $\{\Gamma_x\}_{x \in \cX}$ on $B$ (i.e., $\Gamma_x$
is positive and $\sum_x \Gamma_x = I$).  If the system $B$ is empty,
then the min\-/entropy of $X$ reduces to the R\'enyi entropy of order
infinity, $\Hmin{X} = - \log \max_{x \in \cX} p_x$ (often written
$H_\infty(X)$). In this case the connection to the guessing
probability is particularly obvious: when no side information is
available, the best guess we can make is simply the value $x \in \cX$
with highest probability.

The \emph{smooth} min\-/entropy then consists in maximizing the
min\-/entropy over all subnormalized states $\delta$-close in the
purified distance to the actual state $\rho_{XB}$ of the system
considered. Thus by introducing an extra error $\delta$, we have a
state with potentially much more entropy.

\begin{deff}[smooth min-entropy~\cite{Ren05,TCR10}]
  \label{def:smooth-min-entropy}
  Let $\delta \geq 0$ and $\rho_{AB} \in \sno{AB}$, then the
  \emph{$\delta$-smooth min\-/entropy} of $A$ conditioned on $B$ is defined as
  \begin{equation*}
%    \label{eq:smooth-min-entropy}
    \HminSmooth[\rho]{\delta}{A|B} \coloneqq \max_{\sigma_{AB} :
      P(\sigma,\rho) \leq \delta} \Hmin[\sigma]{A|B}.
  \end{equation*}
\end{deff}

\section{Formalizing Information-Processing Systems}
\label{app:sys}

\subsection{Quantum Combs}
\label{app:sys.combs}

A \emph{quantum comb}~\cite{GW07,Gut12,CDP09} (see also
\cite{Har11,Har12,Har15}) models an interactive quantum
information\-/processing system that receives an input, sends an
output, receives an input, sends an output, etc, and terminates after
a fixed number of steps. Such a system is drawn in
\figref{fig:comb}. One trivial way to model such systems is to
explicitly denote their internal memory, e.g., let $\hilbert_M$ be the
space of the internal memory, then a system is given by a sequence of
CPTP maps:
\[ \cE_i : \lo{X_iM} \to \lo{MY_i}.\]
Let the internal memory start in some initial state $\zero_M$.  Upon
receiving an input $\rho_{X_1} \in \lo{X_1}$ one can evaluate the
output and new memory state as
$\sigma_{MY_1} = \cE_1(\rho_{X_1} \tensor \proj{0}_M)$. Upon receiving
the next input $\sigma_{X_2}$, the new memory state and output are
$\tau_{MY_2} = \cE_2(\sigma_{X_2M})$, etc.

\begin{figure}[tb]
\begin{centering}
\begin{tikzpicture}[scale=1.2,
wire/.style={->,>=stealth,thick},
circle1/.style={draw,circle,fill,thick,inner sep=1.5pt}]

\node[circle1] (u1) at (0,0) {};
\node[circle1] (u2) at (2,0) {};
\node[circle1] (u3) at (4,0) {};
\node[circle1] (u4) at (6,0) {};

\node (v11) at (-.75,-1.1) {};
\node (v12) at (.75,-1.1) {};
\node (v21) at (1.25,-1.1) {};
\node (v22) at (2.75,-1.1) {};
\node (v31) at (3.25,-1.1) {};
\node (v32) at (4.75,-1.1) {};
\node (v41) at (5.25,-1.1) {};
\node (v42) at (6.75,-1.1) {};

\draw[wire] (u1) to (u2);
\draw[wire] (u2) to (u3);
\draw[wire] (u3) to (u4);

\draw[wire] (v11) to (u1);
\draw[wire] (v21) to (u2);
\draw[wire] (v31) to (u3);
\draw[wire] (v41) to (u4);
\draw[wire] (u1) to (v12);
\draw[wire] (u2) to (v22);
\draw[wire] (u3) to (v32);
\draw[wire] (u4) to (v42);

\draw[dashed,thick] (-.35,-1.2) to ++(0,1.6) to ++(6.7,0) to
++(0,-1.6) to ++(-.7,0) to ++(0,.8) to ++(-1.3,0) to ++(0,-.8) to
++(-.7,0) to ++(0,.8) to ++(-1.3,0) to ++(0,-.8) to ++(-.7,0) to
++(0,.8) to ++(-1.3,0) to ++(0,-.8) -- cycle;
\end{tikzpicture}

\end{centering}
\caption[Quantum comb]{\label{fig:comb}A single
  information\-/processing system modeled as a comb. The nodes
  represent an operation and the arrows capture a quantum state. Each
  tooth of the comb corresponds to a pair of an input and an output
  message.}
\end{figure}
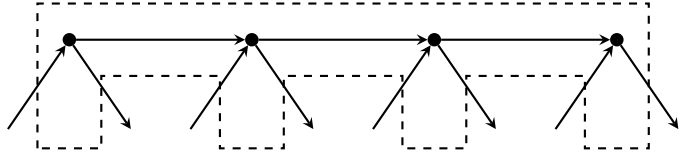

A more compact representation of a comb which omits the internal
memory is given by a single CPTP map
\[ \cE : \lo{X_1 \dotsb X_n} \to \lo{Y_1 \dotsb Y_n},\]
which, for all $i \in [n]$ and all
$\rho_{X_1 \dotsb X_n}, \sigma_{X_1 \dotsb X_n} \in \lo{X_1 \dotsb
  X_n}$
such that $\rho_{X_1 \dotsb X_i} = \sigma_{X_1 \dotsb X_i}$ satisfies
the relation
\[ \ktrace[Y_{i+1} \dotsb Y_n]{\cE\left(\rho_{X_1 \dotsb X_n}\right)} = \ktrace[Y_{i+1}
  \dotsb Y_n]{\cE\left(\sigma_{X_1 \dotsb X_n}\right)},\]
i.e., if the inputs are identical up to position $i$, then the outputs
must be identical up to position $i$ as well. This can be seen as a
causality condition, namely that an input after $i$ cannot influence
an output before $i$.

A very convenient representation of combs is given by the
\cjrep~\cite{Wat16} of the map $\cE$, namely the operator
$\sR_{Y_1\dotsb Y_nX_1\dotsb X_n} \in \lo{Y_1\dotsb Y_nX_1\dotsb X_n}$
given by
\[ \sR_{Y_1\dotsb Y_nX_1\dotsb X_n} \coloneqq \sum_{i,j} \cE\left(\ketbra{i}{j}\right) \tensor
\ketbra{i}{j}.\] The causality condition then becomes
\[ \trace[Y_{i+1} \dotsb Y_n]{\sR_{Y_1\dotsb Y_nX_1\dotsb X_n}} =
\sR_{Y_1\dotsb Y_iX_1\dotsb X_i} \tensor I_{X_{i+1} \dotsb X_n}.\]

\begin{rem}[Sequential scheduling]
\label{rem:scheduling}
If the output value in register $Y_i$ corresponds to one message being
sent to one (random) party \--- where the name of the party might be a
\emph{classical} part of the message\--- then the composition of
combs obtained by routing messages to the correct parties is alway a
new well\-/defined comb.\footnote{This follows because in such a
  network of systems, there is always ever only one active party, the
  one that has just received a message. Thus, the next message to be
  output and the next active party is always uniquely defined.} The
systems used in this work in all sections except
\secref{sec:min-entropy.box} follow this rule, and are thus modeled as
quantum combs. In \appendixref{app:sys.boxes} we discuss some settings
that do not use sequential scheduling, and thus need a more complex
model of systems to be captured.
\end{rem}

\subsection{Causal Boxes}
\label{app:sys.boxes}

The quantum combs introduced in \appendixref{app:sys.combs} are
well\-/suited for modeling finite systems with sequential scheduling,
which is the case for all the concrete protocols analyzed in this
work. There are however situations which require a more developed
model of systems. Consider the example drawn in \figref{fig:ndsystem}:
two players, Alice and Bob, each send a message to a third player,
Charlie, who outputs the first message he receives and ignores the
second. Each of the systems is a well\-/defined comb. Alice and Bob
just output a single message. When Charlie receives the first message,
$m = (v,p)$ \--- value $v$ from player $p$ \--- he outputs $v$ and
ignores all further inputs. But the composition of all three systems
(depicted as a dashed box in \figref{fig:ndsystem}) is not defined: it
is a system with no input and one output, but this output is
undetermined.

\begin{figure}[tb]
\begin{centering}

\begin{tikzpicture}
  \node[draw,thick,minimum height=1cm,minimum width=1.5cm] (A) at (0,1) {Alice};
  \node[draw,thick,minimum height=1cm,minimum width=1.5cm] (B) at (0,-1) {Bob};
  \node[draw,thick,minimum height=1cm,minimum width=1.5cm] (C) at (3,0)  {Charlie};
  \draw[->,>=stealth,thick] (A) to node[auto,sloped,pos=.02] {\tt Alice!} (C);
  \draw[->,>=stealth,thick] (B) to node[auto,sloped,pos=.85] {\tt Bob!} (C);
  \draw[->,>=stealth,thick] (C) to node[auto] {\tt Alice}
  node[auto,swap] {or {\tt Bob}?} (5.4,0);
   \node[fit=(A)(B)(C),dashed,draw,thick] {};
\end{tikzpicture}

\end{centering}
\caption[Order dependent system]{\label{fig:ndsystem}Alice and Bob
  both send messages to Charlie, who outputs the first message he
  receives. Although each system can be described by a comb, the
  composition of the three, depicted as the dashed box, is undefined.}
\end{figure}
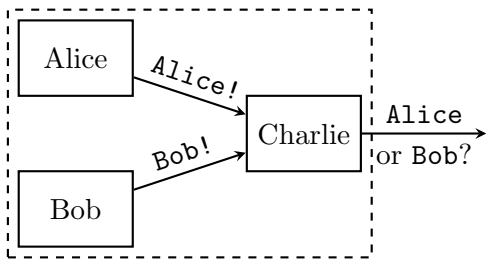

The composition of these three systems is undefined, because Charlie's
output depends on the order of the messages he receives, but this is
not specified by the individual systems of Alice and Bob. The
\emph{causal box} framework~\cite{PMMRT17} was developed to model such
systems. It achieves this by assigning a tag $t \in \T$ to values
$v \in \cV$ that are output (and input) by quantum
information\-/processing systems. $\T$ is a partially ordered set
(e.g., $\T = \rationals$), and $t \in \T$, which can be thought of as
a time, denotes the order of $v$ with respect to other messages. This
also allows superpositions of causal structures to be modeled by
allowing a message to be in a superposition of different positions,
e.g., $\ket{t_1} + \ket{t_2}$. Additionally, when the cardinality of
$\T$ is infinite, the resulting causal boxes can process an unbounded
number of messages, e.g., a beacon which outputs a qubit every second is
valid causal box, but cannot be modeled as a quantum comb.

More precisely, a quantum message that is either input to or output
from a causal box is an element of a Hilbert space with an orthonormal
basis given by $\{\ket{v,t}\}_{v \in \cV, t \in \T}$. For a finite
$\cV$ and infinite $\T$, this Hilbert space corresponds to
\begin{equation} \label{eq:hilbertspace}
  \Ltwo{\T}{\C^{|\cV|}}\,,\end{equation}
where
$\LtwoOp(\T) = \{ (x_t)_{t \in \T} : x_t \in \C, \norm{x} <
\infty\}$
is the sequence space with bounded $2$\=/norm with
$\norm{x} = \sqrt{\braket{x}{x}}$ and
%\begin{equation*} %\label{eq:scalarproduct}
$\braket{x}{y} = \sum_{t \in \T} \overline{x_t}y_t$. %\end{equation*}

An arrow used in figures such as \figref{fig:min-entropy.box}
captures an unbounded number of messages being output or input. We
refer to such an object as a wire (which may connect two systems), and
its Hilbert space is given by a Fock space,
\begin{equation} \label{eq:generalfockspace} \fock{\hilbert} \coloneqq
  \bigoplus_{n = 0}^\infty \vee^n \hilbert\,, \end{equation} where
$\vee^n \hilbert$ denotes the symmetric subspace of
$\hilbert^{\tensor n}$, and $\hilbert^{\tensor 0}$ is the one
dimensional space containing the vacuum state $\vacuum$.  Plugging
\eqnref{eq:hilbertspace} into \eqnref{eq:generalfockspace} for a $d$
dimensional value ($|\cV| = d$), we find that the message space of a
wire is
\begin{equation} \label{eq:quditfockspace}
\fock{\LtwoTC} = \bigoplus_{n = 0}^\infty \vee^n \left(\LtwoTC\right)\,.
\end{equation}
The orthogonal subspaces $\vee^n \left(\LtwoTC\right)$ for
$n \in \naturals_0$ capture $n$ messages being sent on a wire. The
restriction to the symmetric space guarantees that there is no order
amongst the messages other than what might be defined from their state,
e.g., their position in $\T$.

Let $A$ denote a wire that carries $d_A$\=/dimensional
messages. Abusing language, we will refer to this as a
$d_A$\=/dimensional wire.\footnote{The Hilbert space of the wire is in
  fact infinite dimensional.} We write $\cF^{\T}_A$ for the
corresponding state space, namely
\[\cF^{\T}_A \coloneqq \fock{\Ltwo{\T}{\C^{d_A}}}\,.\] Note that for any
Hilbert spaces $\hilbert_A$ and $\hilbert_B$,
\begin{equation} \label{eq:fockisomorphism}
\fock{\hilbert_A} \tensor \fock{\hilbert_B} \cong \fock{\hilbert_A \oplus
    \hilbert_B}\,,\end{equation}
where the isomorphism preserves the
meaning associated with the bases of the Fock spaces, i.e., a tensor
product of two vacuum states on the left in
\eqnref{eq:fockisomorphism} is mapped to a vacuum state on the right,
a tensor product of a vacuum state and one message on the left is
mapped to a single message with the same value and position on the
right, etc.\ (see \cite{PMMRT17} for the exact isomorphism). This
allows the tensor product of two wires with dimensions $d_A$ and $d_B$
to be written as one wire with dimension $d_A + d_B$,
\[ \cF^{\T}_A \tensor \cF^{\T}_B \cong \cF^{\T}_{AB}.\]
Similarly, for any $\cP \subset \T$, a wire may be split in two parts
corresponding to messages in $\cP$ and $\tilde{\cP} \coloneqq \cT
\setminus \cP$, respectively:
\[\cF^{\T}_{A} \cong \cF^{\cP}_A \tensor \cF^{\tilde{\cP}}_A.\]
We use this in particular to trace out messages that are not before
some position $t
\in \T$ by taking $\cP = \{ p \in \cT : p \leq t\}$, e.g., \[
\rho^{\leq t}_A = \trace[\nleq t]{\rho_A}.\]

Let $\cF^\T_X$ and $\cF^\T_Y$ denote the Hilbert spaces of an input
wire $X$ and output wire $Y$. And let $\tcop{\cF^\T_X}$ and
$\tcop{\cF^\T_Y}$ be the corresponding sets of trace class
operators. For the specical case of a totally ordered set
$\T$,\footnote{The case for a partial order on $\T$ can be found in
  \cite{PMMRT17}.} a causal box is defined as a set of mutually consistent
CPTP maps
\[ \sR = \left\{ \cE^{\leq t} : \tcop{\cF^\T_X} \to \tcop{\cF^{\leq
    t}_Y}\right\}_{t \in \T},\] i.e., for $t \leq u$,
\[ \cE^{\leq t} = {\tr_{> t}} \circ \cE^{\leq u}.\]
In some cases the map $\cE = \lim_{t \to \infty} \cE^{\leq t}$ is
well\-/defined, in which case a causal box can be defined by this
limit instead of by the sequence.

As for quantum combs, causal boxes must satisfy a notion of causality,
so that an input before position $t$ cannot influence an output after
position $t$. The exact definition is not needed in this work, so we
omit it, and refer the interested reader to \cite{PMMRT17}. Similarly,
causal boxes can be represented using the \cj isomorphism, which we
omit here as well.

Causal boxes may be connected in arbitrary ways, i.e., any output wire
of one box can be ``plugged into'' an input wire of the same dimension
from a different box, or may be looped back and connected to one if
its own inputs. This always results in a new well\-/defined causal
box, even if loops are present. It does not create any causality
conflicts since messages are ordered and an output can only depend on
inputs that arrived before.

\begin{rem}[From combs to causal boxes]
\label{rem:combs2boxes}
  Quantum combs may be seen as special cases of causal boxes, which do
  not specify the positions $t \in \T$ of the outputs, only the
  (local) order with respect to the other outputs it generates. One
  can easily ``upgrade'' a quantum comb to a causal box by assigning
  some (fixed) processing time $\delta_i$ to produce the output $Y_i$
  after receiving the input $X_i$, in which case it inherits all the
  properties (e.g., closure under composition) of causal
  boxes. Alternatively, a comb can be modeled as a specification of
  causal boxes, namely those that produce the same outputs in the same
  order, but at undetermined times.
\end{rem}

\section{Quantum Key Distribution}
\label{app:distil.qkd}

A QKD protocol typically has three phases.\footnote{A detailed review
  of QKD can be found in \cite{SBCDLP09}.} In the first, the players
exchange some quantum states \--- either one player (Alice) generates
them and sends them to the other (Bob), or an untrusted third party
(Eve) prepares states that are sent to both players. We assume for
simplicity that the players measure the quantum states upon reception,
but the same analysis holds for protocols that require quantum memory
as well. In \figref{fig:qkd.real} we illustrated the case where the
players have access to an (insecure) quantum channel that they use to
send quantum states from Alice to Bob.  In the second phase of the
protocol, the players compare some of the measurement results to
estimate the amount of noise on the channel. At the end of this phase,
they obtain a bound on the entropy of the remaining undisclosed bit
strings conditioned on the adversary's information.  In the final
phase, the players run some (classical) post\-/processing protocols on
the strings they hold to extract a secret key. More precisely, they
first need to correct errors between the strings held by Alice and
Bob. Then they run a privacy amplification step to extract a secret
key from their strings.

In \secref{sec:distil} we used these error correction and privacy
amplification procedures to get a secret key from any min\-/entropy
resource. Here, we show that a standard QKD security proof, but with
the post\-/processing omitted, constructs such a min\-/entropy
resource. More precisely, let
$\pi^{\dis}_{AB} = (\pi^\dis_A,\pi^\dis_B)$ be a protocol that
distributes quantum states using an insecure channel $\aQ$, then uses
an authentic channel $\aA^c$ to compare the measurement results, and
either aborts or produces two (random) strings $X$ and $Y$. Let Eve
have access to the quantum channel and apply any operation allowed by
quantum mechanics to the states being sent, and let her obtain a
transcript of the messages sent on the classical authentic channel. We
prove below that if we can bound the information Eve has about $X$,
then $\pi^\dis_{AB}$ constructs a min\-/entropy resource.

Proving that we can bound the information Eve has about $X$ is the
difficult part of QKD security proofs. Here we only show that if this
can be done, then the protocol can be written as a constructive
statement in the AC framework. An overview of how one can actually
bound Eve's information can be found in \cite{SBCDLP09}, and
detailed security proofs for BB84 and BBM92 that compute these bounds
are given in \cite{TL15}.

\begin{lem}
\label{lem:distil.min-entropy}
Let $\pi^{\dis}_{AB} = (\pi^\dis_A,\pi^\dis_B)$ be a protocol as
described above. Suppose that one can prove that the subnormalized
state $\sigma_{XYE}$ resulting from Alice and Bob not aborting is such
that $\HminSmooth[\sigma]{\delta}{X|E} \geq k$. And let
$\overline{\aH}^{k,\delta}_{\min}$ be a $(k,\delta)$\-/min\-/entropy
resource where the output $Y$ at Bob's interface is arbitrary (but is
always $\bot$ if Alice's is $\bot$). Then $\pi^{\dis}_{AB}$
(perfectly) constructs $\overline{\aH}^{k,\delta}_{\min}$ from
$\aQ \| \aA^c$,
\[
  \aQ \| \aA^c \xrightarrow{\pi^{\dis}_{AB},0} \overline{\aH}^{k,\delta}_{\min}.
\]
\end{lem}

\begin{proof}
  We will show that
  $\pi^{\dis}_{AB} \left( \aQ \| \aA^c \right) \subset
  \overline{\aH}^{k,\delta}_{\min}$.
  By contradiction, suppose there exists
  $\sH \in \pi^{\dis}_{AB} \left( \aQ \| \aA^c \right)$ such that
  $\sH \notin \overline{\aH}^{k,\delta}_{\min}$. This means there must exist an
  $\sS$ such that the output
  $\rho_{XE} = \proj{\bot} \tensor \tau_{E} + \sigma_{XE}$ after
  interacting with $\sH$ has $\HminSmooth[\sigma]{\delta}{X|E} <
  k$.
  By running such a system $\sS$, Eve would thus obtain more
  information about the string $X$ than allowed.
\end{proof}

If no adversary is eavesdropping on the quantum channel $\aQ$, but
instead it is only subject to natural noise, then one can make
stronger statements than \lemref{lem:distil.min-entropy} in which $Y$ is
not arbitrary, but a bound on the number of errors between $X$ and $Y$
is known, and the probability of aborting is also bounded. Such
statements are necessary for proving the robustness of the protocols,
i.e., the probability that they terminate with a shared secret key
when only natural noise is present (see~\cite{PR14} for a formal
treatment of robustness in QKD.)

For example, instead of an insecure channel $\aQ$, let the players
share a noisy channel specification $\aC$ such as the depolarizing
channels depicted in \figref{fig:channel.noisy}. Then one can often
prove statements such as
\begin{equation}
\label{eq:distil.qkd}
  \aC \| \aA^c \xrightarrow{\pi^{\dis}_{AB},\eps_\noise} \overline{\aR}^{k,\delta}_{\min},
\end{equation}
where
$\overline{\aR}^{k,\delta}_{\min} \subset
\overline{\aH}^{k,\delta}_{\min}$
is a specification of min\-/entropy resources that never aborts and
outputs strings $X$ and $Y$ with specific correlations (e.g., no more
than $t$ bit flips).

Composing \corref{cor:distil.comp}, \lemref{lem:distil.min-entropy},
and \eqnref{eq:distil.qkd}, we recover the standard QKD security
statement~\cite{PR14}.

\begin{cor}
\label{cor:qkd}
Let $\pi^{\qkd} = \pi^\pa \circ \pi^\ec \circ \pi^\dis$ consist of
the composition of the protocols described here and
\secref{sec:distil}, then
\begin{gather*}
%\label{eq:qkd.cor}
 \aC \| \aA^c \xrightarrow{\pi^{\qkd},\eps_\noise+\eps_\pa+2\delta} \underline{\aK}^m,
\intertext{and}
 \aQ \| \aA^c \xrightarrow{\pi^{\qkd},\eps_\verif+\eps_\pa+2\delta} \aK^m,
\end{gather*}
where $\aA^c = \aA^{c_1} \| \aA^{c_2} \| \aA^{c_3}$, and
$\aA^{c_i}$ are the authentic channels used by the different parts of
the protocol.
\end{cor}

% \begin{rem}
% \label{rem:qkd}
% As pointed out in \remref{rem:distil.ec}, if instead of using authentic
% channels, the players use secure channels $\aS^c$, then no loss of
% entropy occurs during error correction. This directly results in a
% longer secret key being generated, i.e., we get
% \[  \aQ \| \aS^c \xrightarrow{\pi^{\qkd},\nu} \aK^{m+r+t},\]
% instead of \eqnref{eq:qkd.cor}, where as extractor one chooses a
% function $\Ext : \{0,1\}^n \times \{0,1\}^d \to \{0,1\}^{m+r+t}$.
% \end{rem}

\section{Technical Lemmas}
\label{app:tech}

\subsection{Min-entropy inequalities}
\label{app:tech.min-entropy-inequalities}

The following inequality shows that if an additional $r$ bit string
$Z$ is given to the adversary, then she gets $r$ bits of information.

\begin{lem}[\protect{\cite[Lemma 11]{WTHR11}}]
\label{lem:tech.min-entropy-inequalities.chain}
Let $\rho \in \sno{ABZ}$ be any subnormalized state where $Z$ is a
classical system over the alphabet $\cZ$. Then
  \[\HminSmooth[\rho]{\delta}{A|BZ} \geq \HminSmooth[\rho]{\delta}{A|B} - \log
  |\cZ|.\]
\end{lem}

The next inequality shows that the min\-/entropy of a
state conditioned on an event is bounded by the min\-/entropy before
this conditioning.

\begin{lem}
\label{lem:tech.min-entropy-inequalities.event}
Let $\rho \in \sno{ABZ}$ be any subnormalized state with a binary
classical register $Z$. This state may be written as \[
\rho_{ABZ} = \rho^0_{AB} \tensor \proj{0} + \rho^1_{AB}
\tensor \proj{1}.\]  Then
\[ \HminSmooth[\rho^0]{\delta}{A|B} \geq
\HminSmooth[\rho]{\delta}{A|B}\ .\]
\end{lem}

\begin{proof}
  Follows from \cite[Lemma 10]{TL15} by taking $X$ to be empty,
  $Y = Z$ and $\Omega : \cY \to \{0,1\}$ is the identity. We also
  remove the condition $\delta \in [0,\sqrt{\tr \rho})$ by defining
  $\Hmin[\rho]{A|B} = +\infty$ if
  $\tr \rho = 0$.
\end{proof}

\subsection{Extractors}
\label{sec:tech.extractors}

Extractors are usually defined on normalized states. We show here than
any extractor for normalized states is also an extractor for
subnormalized states. This proof follows the steps of an equivalent
proof for multi\-/source extractors in the Markov model from
\cite[Lemma~37]{AFPS16}, but is adapted to seeded extractors.

\begin{lem}
\label{lem:tech.extractors.sub}
If $\Ext$ is a quantum\-/proof $(k,\eps)$\-/strong extractor for
normalized states, then it is a quantum\-/proof $(k+1,2\eps)$\-/strong extractor for
subnormalized states.
\end{lem}

\begin{proof}
  Let $\rho_{XE} \in \sno{XE}$ be a subnormalized state with
  $\Hmin[\rho]{X|E} \geq k+1$, and let $p = \trace{\rho_{XE}}$. We
  define $\tilde{\rho}_E \coloneqq \frac{1-p}{p} \rho_E$ and the normalized state
  \begin{equation}
  \label{eq:tech.extractors.sub}
  \sigma_{XEP} \coloneqq \rho_{XE} \tensor \proj{0}_P +
    \tau_{X} \tensor \tilde{\rho}_E \tensor \proj{1}_P,
  \end{equation}
  where $\tau_{X}$ is the fully mixed state.

  This state satisfies a slightly modified min\-/entropy condition:
  \begin{align*}
    \pguess[\sigma]{X|EP} & = \pguess[\rho]{X|E} +
                            \pguess[\tau \tensor \tilde{\rho}]{X|E} \\
                          & \leq 2^{-k-1} +(1-p) 2^{-n} \leq 2^{-k},
  \end{align*}
  where $n$ is the length of the string $X$. Hence
  $\Hmin[\sigma]{X|EP} \geq k$. And because $\Ext$ is an extractor
  for normalized states it follows that
  \[
  \frac{1}{2}\trnorm{\sigma_{\Ext(X,Z)ZEP} - \tau_K \tensor \tau_Z
    \tensor \sigma_{EP}} \leq \eps.
  \]

  Plugging in \eqnref{eq:tech.extractors.sub} and tracing out $P$, we get
  \[
    \frac{1}{2}\trnorm{\rho_{\Ext(X,Z)ZE} - \tau_{\Ext(X,Z)Z} \tensor
    \tilde{\rho}_E - \tau_K \tensor \tau_Z \tensor \rho_E +
    \tau_K \tensor \tau_Z \tensor  \tilde{\rho}_E} \leq \eps.
  \]
  Thus starting from the expression
  $\trnorm{\rho_{\Ext(X,Z)ZE} - \tau_K \tensor \tau_Z \tensor \rho_E}$
  and then adding and subtracting the term
  $\tau_{\Ext(X,Z)Z} \tensor \tilde{\rho}_E - \tau_K \tensor \tau_Z \tensor \tilde{\rho}_E$
  as well as applying the triangle inequality leaves us with
  \begin{align*}
    & \frac{1}{2}\trnorm{\rho_{\Ext(X,Z)ZE} - \tau_K \tensor \tau_Z \tensor
      \rho_E} \\
    & \qquad \qquad \qquad \qquad \leq \eps +\frac{1}{2}\trnorm{\tau_{\Ext(X,Z)Z} \tensor \tilde{\rho}_E -
      \tau_K \tensor \tau_Z \tensor \tilde{\rho}_E} \\
    & \qquad \qquad \qquad \qquad \leq \eps +\frac{1}{2}\trace{\tilde{\rho}_{E}}\trnorm{\tau_{\Ext(X,Z)Z} -
      \tau_K \tensor \tau_Z} \leq 2\eps.	\qedhere
\end{align*}
\end{proof}

The following lemma shows that any extractor defined for subnormalized
states can be used to extract from states with a bound on the smooth
min\-/entropy instead of on the min\-/entropy with a small adjustment
to the error parameter. Note that the original lemma from \cite[Lemma
3.5]{DPVR12} omitted to specify that the extractor has to be defined
for subnormalized states for the proof to go through.

\begin{lem}[\protect{\cite[Lemma 3.5]{DPVR12}}]
\label{lem:tech.extractors.smooth}
If $\Ext$ is a quantum\-/proof $(k,\eps)$\-/strong extractor for
subnormalized states, then for any subnormalized
$\rho_{XE} \in \sno{XE}$ with classical $X$ and
$\HminSmooth[\rho]{\delta}{X|E} \geq k$, and a uniform $Z$,
\[\frac{1}{2} \trnorm{\rho_{\Ext(X,Z)ZE} - \tau_K \tensor \tau_Z \tensor
  \rho_E} \leq \eps + 2 \delta,\]
where $\tau_K$ is the fully mixed state.
\end{lem}

The final extractor lemma that we need states that the composition of
two extractors is also an extractor, and is also taken from
\cite{DPVR12}.

\begin{lem}[\protect{\cite[Lemma A.4]{DPVR12}}]
\label{lem:tech.extractors.composition}
Let $\Ext_1: \{0,1\}^n \times \{0,1\}^{d_1} \to \{0,1\}^{m_1}$ and
$\Ext_2: \{0,1\}^n \times \{0,1\}^{d_2} \to \{0,1\}^{m_2}$ be
quantum-proof $(k,\eps_1)$- and $(k-m_1,\eps_2)$-strong extractors for
subnormalized states. Then the composition of the two, namely
  \begin{align*}
    \Ext_3 : & \{0,1\}^n \times \{0,1\}^{d_1+d_2} \to
    \{0,1\}^{m_1+m_2} \\
    & (x,y_1\|y_2) \mapsto \Ext_1(x,y_1)\|\Ext_2(x,y_2),
  \end{align*}
  is a quantum-proof $(k,\eps_1+\eps_2)$-strong extractor for
  subnormalized states.
\end{lem}

\subsection{Key-Privacy and uniformity}

The lemmas in this section are all from \cite{FS17a}. The first states
that a specific kind of extractor can provides $\nu$\-/key\-/privacy
(\defref{def:key-privacy}) and uniformity (\defref{def:uniformity}).

\begin{lem}[\protect{\cite[Proposition~2]{FS17a}}]
\label{lem:key-privacy}
Let $\Ext : \{0,1\}^{n} \times \{0,1\}^{r-m} \to \{0,1\}^{m}$ be a
quantum\-/proof $(k,\eps)$\-/strong extractors for subnormalized
states for any $k$ and $\eps = \frac{\nu}{2}\sqrt{2^{-k+m}}$. And
let $h : \{0,1\}^{n} \times \{0,1\}^{r} \to \{0,1\}^{m}$ be
defined as
$h(x,\ell_1 \| \ell_2) \coloneqq \Ext(x,\ell_1) \xor \ell_2$. Then
$h$ provides $\nu$\-/key\-/privacy.
\end{lem}

The next lemma shows that the composition of two functions that
provide key\-/privacy results in a new function also providing
key\-/privacy. The same holds for uniformity.

% \begin{lem}[\protect{\cite[Proposition~3]{FS17a}}]
% \label{lem:key-privacy.parallel}
% Let $h_1 : \cX \times \cL_1 \to \cT_1$ and
% $h_2 : \cX \times \cL_2 \to \cT_2$ provide $\nu_1$ and
% $\nu_2$\-/key\-/privacy, respectively. Then
% $h : \cX \times \left(\cL_1 \times \cL_2 \right) \to \left(\cT_1
%   \times \cT_2 \right)$
% with $h(x,\ell_1\|\ell_2) \coloneqq h_1(x,\ell_1) \| h_2(x,\ell_2)$
% provides $\nu_1+\nu_2$\-/key\-/privacy.
% \end{lem}

\begin{lem}[\protect{\cite[Proposition~4]{FS17a}}]
\label{lem:key-privacy.sequarallel}
Let $h_1 : \cX \times \cL_1 \to \cT_1$ and
$h_2 : \left(\cX \times \cT_1\right) \times \cL_2 \to \cT_2$ be two
functions. We define 
$h : \cX \times \left(\cL_1 \times \cL_2 \right) \to \left(\cT_1
  \times \cT_2 \right)$
with $h(x,\ell_1\|\ell_2) \coloneqq t \| h_2(x \| t,\ell_2)$ where
$t \coloneqq h_1(x,\ell_1)$. If $h_1$ and $h_2$ provide
$\nu_1$ and $\nu_2$\-/key\-/privacy, respectively, then $h$
 provides $\nu_1+\nu_2$\-/key\-/privacy. And if $h_1$ and $h_2$ are
 both uniform, then $h$ is uniform.
\end{lem}

\subsection{Guessing Games}
\label{app:guessing}

The lemmas in this section are also from \cite{FS17a}. They are
concerned with a setting in which different players sharing a quantum
state are trying to guess the outcome of a measurement performed by
one of them.

In the following we denote $n$ EPR
pairs by $\Phi^+_{AB}$, and $P^\theta_x = H^\theta \proj{x} H^\theta$
denotes a projector, which we use to measure qubits in either the
computational or diagonal basis. Note that by measuring half of EPR
pairs with these projectors, the other half results in the state
$H^\theta \proj{x} H^\theta$, i.e., \[ \ktrace[A]{\left(P^\theta_x \tensor
  I_B \right) \Phi^+_{AB} \left(P^\theta_x \tensor  I_B \right)} =
\frac{1}{2^n} H^\theta \proj{x} H^\theta.\]
Given two registers $\Theta$ and $A$, where $\Theta$ is classical and
$A$ is quantum, we denote the map which performs
this measurement on $A$ according to $\Theta$ and writes the result in a register $X$ as
$\cM^{\BB}_{\Theta A,X}$, i.e., 
\[ \cM^{\BB}_{\Theta A,X} \left( \rho_{\Theta A} \right) = \sum_{\theta,x} \tr_A
\proj{x} \tensor \left[ \left( \proj{\theta} \tensor P^\theta_x
\right) \rho_{\Theta A B} \left( \proj{\theta} \tensor P^\theta_x
\right)\right].\]

The first lemma considers a two player setting, where one player obtains
the measurement outcome $X$ and the second player wants to guess it.

\begin{lem}[\protect{\cite[Corollary~2]{FS17a}}]
\label{lem:guessing.one}
  Let $\rho_{\Theta E}$ be any cq-state, where the strings
  $\theta \in \cC \subset \{0,1\}^n$ are taken from a code $\cC$ with
  minimal distance $d$, let $\cE : \lo{E} \to \lo{AB}$ be any CPTP map
  where $\lo{A}$ is an $n$ qubit Hilbert space, and let
  $\sigma_{X\Theta B} \coloneqq \cM^{\BB}_{\Theta A,X} \circ \cE
  \left(\rho_{\Theta E}\right)$. Then
\[\pguess[\sigma]{X|B} \leq \pguess[\rho]{\Theta | E} \left(1 +
  \frac{|\cC|}{2^{d/2}} \right).\]
\end{lem}

The second lemma considers a three player setting. The first player
obtains the measurement outcome $X$, the second one has to get an
outcome $X'$ that is close to $X$ and the third one wants to guess
$X$. Note that in this setting, the all players have access to
$\Theta$.

\begin{lem}[\protect{\cite[Corollary~3, Remark~6]{FS17a}}]
\label{lem:guessing.two}
Let $\rho_{\Theta E}$ be any cq-state, where the strings
$\theta \in \cC \subset \{0,1\}^n$ are taken from a code $\cC$ with
minimal distance $d$. Let $\cE : \lo{E} \to \lo{ABC}$ be any CPTP map
where $\lo{A}$ and $\lo{B}$ are an $n$ qubit Hilbert spaces. Let
$\sigma_{XX'\Theta C} \coloneqq \cM^{\BB}_{\Theta A,X} \circ
\cM^{\BB}_{\Theta B,X'} \circ \cE \left(\rho_{\Theta E}\right)$,
i.e., both $A$ and $B$ are measured according to $\Theta$ and the
results are written in $X$ and $X'$, respectively. Finally, let
$\sigma^1_{XX'\Theta C}$ be the projection of $\sigma_{XX'\Theta C}$ on
the space with $w(X,X') \leq \varphi n$. Then
\[\pguess[\sigma^1]{X|\Theta C} \leq \pguess[\rho]{\Theta | E} \left(1 +
  \frac{|\cC|2^{h(\varphi)n}}{2^{d}} \right).\]
\end{lem}

\section*{Acknowledgments}
\pdfbookmark[1]{Acknowledgments}{sec:ack}
%\addcontentsline{toc}{section}{Acknowledgments}

CP if grateful to Serge Fehr for lively conversations on composable
security, and for a lot of help understanding his paper. He would also
like to thank Ueli Maurer for discussions on Constructive
Cryptography, and Renato Renner for commenting on an initial draft of
this work and proposing many of the ideas used in this paper.

CP is partially supported by the US Air Force Office of Scientific
Research (AFOSR) via grant~FA9550-16-1-0245, the Swiss National
Science Foundation (via the National Centre of Competence in Research
`Quantum Science and Technology'), and the Zurich Information Security
and Privacy Center. %It represents the views of the authors.

% \bibliographystyle{eprintunsrt} %\bibliographystyle{tocplainedited} %\bibliographystyle{eprintalpha}
% \bibliography{quantum,classical} %,Akinori

% Define empty bibhead if not already defined
\providecommand{\bibhead}[1]{}
% Define tocrefpdfbookmark if not already defined
\expandafter\ifx\csname pdfbookmark\endcsname\relax%
  \providecommand{\tocrefpdfbookmark}{}
\else\providecommand{\tocrefpdfbookmark}{%
   \phantomsection%
   \addcontentsline{toc}{section}{\refname}}%
\fi

\tocrefpdfbookmark

\end{document}